\documentclass[11pt]{article}
\usepackage{amssymb,amsfonts,amsmath,amscd,dsfont,mathrsfs,graphicx,float,psfrag,epsfig,color,enumerate,hyperref,breakurl,amsthm}
\definecolor{red2}{RGB}{176,0,0}

\newtheorem{claim}{Claim}
\newtheorem{lemma}{Lemma}

\newtheorem{theorem}{Theorem}
\newtheorem{proposition}[claim]{Proposition}
\newtheorem{corollary}[claim]{Corollary}
\newtheorem{Definition}[claim]{Definition}
\newtheorem{Remark}[claim]{Remark}

\def\cA{{\cal A}}
\newcommand{\enp} {}

\newcommand{\p}{\partial}

\newcommand{\mc}{\mathcal}

\newcommand{\f}{\frac}

\newcommand{\alf}[2]{\alpha_{#1\backslash #2}}

\newcommand{\off}[2]{m_{#1\rightarrow #2}}

\newcommand{\baralf}{\underline{\alpha}}
\newcommand{\barbet}{\underline{\beta}}

\newcommand{\bargamma}{\underline{\gamma}}

\newcommand{\baroff}{\underline{m}}

\newcommand{\surp}{\mathcal{S}\textup{urp}}
\newcommand{\BALOPT}{\mathcal{BALOPT}}
\newcommand{\sT}{{\sf{T}}}

\def\reals{{\mathds R}}

\def\di{{\partial i}}
\def\cC{{\cal C}}

\def\eps{{\epsilon}}

\def\damp{\kappa}

\def\dj{{\partial j}}
\def\u0{\underline{0}}

\def\NB{{\textup {\tiny NB}}}

\def\rfl{{\cal R}}
\def\Wmax{{W_{\rm max}}}

\def\reb{{\fontfamily{cmr`}\selectfont \textup{\tiny reb}}}

\makeatletter
\def\equalsfill{$\m@th\mathord=\mkern-7mu
\cleaders\hbox{$\!\mathord=\!$}\hfill
\mkern-7mu\mathord=$}
\makeatother

\definecolor{Red}{rgb}{1,0,0}
\definecolor{Blue}{rgb}{0,0,1}
\definecolor{Olive}{rgb}{0.41,0.55,0.13}
\definecolor{Green}{rgb}{0,1,0}
\definecolor{MGreen}{rgb}{0,0.8,0}
\definecolor{DGreen}{rgb}{0,0.55,0}
\definecolor{Yellow}{rgb}{1,1,0}
\definecolor{Cyan}{rgb}{0,1,1}
\definecolor{Magenta}{rgb}{1,0,1}
\definecolor{Orange}{rgb}{1,.5,0}
\definecolor{Violet}{rgb}{.5,0,.5}
\definecolor{Purple}{rgb}{.75,0,.25}
\definecolor{Brown}{rgb}{.75,.5,.25}
\definecolor{Grey}{rgb}{.5,.5,.5}
\definecolor{Pink}{rgb}{1,0,1}
\definecolor{DBrown}{rgb}{.5,.34,.16}

\definecolor{Black}{rgb}{0,0,0}

\newcommand{\bmax}[1]{{(b_{#1}^{\rm th}\!\operatorname{-max})}}

\begin{document}

\title{Bargaining dynamics in exchange networks}

\author{
Mohsen Bayati\thanks{Operations and Information Technology, Graduate School of Business, Stanford University}
\and
Christian Borgs\thanks{Microsoft Research New England}
\and
Jennifer Chayes${}^{2}$
\and
Yashodhan Kanoria\thanks{Department of Electrical Engineering,
Stanford University, \texttt{ykanoria@stanford.edu}, (650) 353-0476}
\and
Andrea Montanari\thanks{Department of Electrical Engineering and Department of Statistics,
Stanford University}
}

\maketitle

\begin{abstract}
We consider a one-sided assignment market or exchange network with transferable utility
and propose a  model
for the dynamics of bargaining in such a market. Our dynamical model is local, involving iterative updates of `offers' based on estimated best alternative matches, in the spirit of pairwise Nash bargaining. We establish that when a balanced outcome (a generalization of the pairwise Nash bargaining solution to networks) exists, our dynamics converges rapidly to such an outcome.
We extend our results to the cases of (i) general agent `capacity constraints', i.e., an agent may be allowed to participate in multiple matches, and (ii) `unequal bargaining powers' (where we also find a surprising change in rate of convergence).\\

\noindent{\bf Keywords:} Nash bargaining, network, dynamics, convergence, matching, assignment, exchange, balanced outcomes. \\[3pt]
\noindent{\bf JEL classification:} C78
\end{abstract}

\section{Introduction}
\label{sec:intro}

Bargaining has been heavily studied in the economics and sociology literature, e.g.,  \cite{Nash,sequential_bargaining,KS,RubinBrown}. While the case of bargaining between two agents is fairly well understood \cite{Nash,sequential_bargaining,KS}, less is known about the results of bargaining on networks. We consider exchange networks \cite{CookY,KT}, also called assignment markets \cite{ShapleyShubik72,Rochford}, where agents occupy the nodes of a network, and edges represent potential partnerships between pairs of agents, which can generate some additional value for these agents. To form a partnership,  the pair of agents must reach an agreement on how to split the value of the partnership. Agents are constrained on the number of partnerships they can participate in, for instance, under a matching constraint, each agent can participate in at most one partnership. The fundamental question of interest is: Who will partner with whom, and on what terms? Such a model is relevant to the study of the housing market, the labor market, the assignment of interns to hospitals, the marriage market and so on.
An assignment model is suitable for markets with heterogeneous indivisible goods, where trades are constrained by a network structure.

Balanced outcomes\footnote{Also called \emph{symmetrically pairwise-bargained allocations} in \cite{Rochford} or \emph{Nash bargaining solutions} \cite{KT}.} generalize the pairwise Nash bargaining solution to this setting. The key issue here is the definition of the threats or best alternatives of participants in a match --- these are defined by assuming the incomes of other potential partners to be fixed. In a balanced outcome on a network, each pair plays according to the local Nash bargaining solution thus defined. Balanced outcomes have been found to possess some favorable properties:

\emph{Predictive power.} The set of balanced outcomes refines the set of stable outcomes (the core), where players have no incentive to deviate from their current partners. For instance, in the case of a two player network, all possible deals are stable, but there is a unique balanced outcome. Balanced outcomes have been found to capture various experimentally observed effects on small networks \cite{CookY,KT}.

\emph{Computability.} Kleinberg and Tardos \cite{KT} provide an efficient centralized algorithm to compute balanced outcomes. They also show that balanced outcomes exist if and only if stable outcomes exist.

\emph{Connection to cooperative game theory.} The set of balanced outcomes is identical to the core intersection prekernel of the corresponding cooperative game \cite{Rochford,Bateni}.

However, this leaves unanswered the question of whether balanced outcomes can be predictive on large networks, since there was previously no dynamical description of how players can find such an outcome via a bargaining process.

In this work, we consider an assignment market with arbitrary network structure
and edge weights that correspond to the value of the potential partnership on that edge.
We present a model for the bargaining process in such a network
setting, under a capacity constraint on the number of exchanges
that each player can establish. We define a model that satisfies two
important requirements: locality and convergence.

\vspace{0.1cm}

\emph{Locality.} It is reasonable to assume that agents have accurate
information about their neighbors in the network (their negotiation
partners). Specifically, an agent is expected to know the weights of the edges with each of her negotiation partners. Further, for $(ij) \in E$, agent $i$ may estimate the `best alternative' of agent $j$ in the current `state' of negotiations.\footnote{For instance, $i$ may form this estimate based on her conversation with $j$. (We do not present a game theoretic treatment with fully rational/strategic agents in this work, cf. Section \ref{sec:discussion}.)}
On the other hand, it is
unlikely that the offer made to an agent during a negotiation depends
on arbitrary other agents in the network. We thus require that agent
make `myopic' choices on the basis of their neighborhood in the
exchange networks. This is consistent with the bulk of the game theory literature on
learning \cite{FLbook,KMR}.

\vspace{0.1cm}

\emph{Convergence.} The duration of a negotiation
is unlikely to depend strongly on the overall network size.
For instance, the negotiation on the price of a house, should not
depend too much on the size of the town in which it takes place, all
other things being equal. Thus a realistic model for negotiation should converge to
a fixed  point (hence to a set of exchange agreements) in a time
roughly independent of the network size.

\vspace{0.1cm}

Our dynamical model is fairly simple. Players compute the current best
alternative to each exchange, both for them, and for their partner.
On the basis of that, they make a new offer to their partner
according to the pairwise Nash bargaining solution.
This of course leads to a change in the set of best alternatives
at the next time step.
We make the assumption that `pairing' occurs at the end, or after several iterative updates, thus suppressing the effect of agents pairing up and leaving.

This dynamics is of course local. Each agent only needs to know the
offers she is receiving, as well as the offers that her potential
partner is receiving. The technical part of this paper is therefore
devoted to the study of the convergence properties of this dynamics.

Remarkably, we find that it converges rapidly, in time
roughly independent of the network size. Further its fixed points are
not arbitrary, but rather a suitable generalization of Nash bargaining
solutions  \cite{CookY,Rochford} introduced in the context of assignment markets and
exchange networks   (also referred to as `balanced
outcomes' \cite{KT} or `symmetrically pairwise-bargained allocations'
\cite{Rochford}).
Thus, our work provides a dynamical justification of Nash bargaining solutions in assignment markets.

We now present the mathematical definitions of bargaining networks and balanced outcomes.

The network consists of a graph $G=(V,E)$, with positive weights
$w_{ij}>0$ associated  to the edges $(i,j)\in E$.
A player sits at each node of this network,
and two players connected by edge $(i,j)$ can
share a profit of $w_{ij}$ dollars if they agree to
trade with each other.
Each player can trade with at most one of her neighbors
(this is called the \emph{$1$-exchange rule}), so that
a set of valid trading pairs forms a matching $M$ in the graph $G$.

We define an \textit{outcome} or \textit{trade outcome} as a pair
$(M,\bargamma)$ where $M\subseteq E$ is a matching of $G$, and
$\bargamma = \{\gamma_i\, :\, i\in V\}$ is the vector of players'
profits.
This means, $\gamma_i\ge 0$, and
$(i,j) \in M$ implies $\gamma_i + \gamma_j = w_{ij}$, whereas for every unmatched node $i \notin M$ we have $\gamma_i=0$.

A balanced outcome, or Nash bargaining (NB) solution, is a trade outcome that satisfies the additional requirements of \textit{stability} and \textit{balance}.
Denote by $\partial i$ the set of neighbors of node $i$ in $G$.

\noindent\textit{Stability.} If player $i$ is trading with $j$, then
she cannot earn more by simply changing her trading partner.
Formally $\gamma_i + \gamma_j \ge w_{ij}$ for all
$(i,j) \in E\setminus M$.

\noindent\textit{Balance.} If player $i$ is trading with $j$, then the surplus of $i$ over her best alternative must be
equal to the surplus of $j$ over his best alternative. Mathematically,
\begin{align}
\gamma_i - \max_{k \in \di \backslash j} (w_{ik}- \gamma_k)_+ =
\gamma_j - \max_{l \in \dj \backslash i} (w_{jl}- \gamma_l)_+
\label{eq:balance}
\end{align}
for all $(i,j) \in M$. Here $(x)_+$ refers to the non-negative part of $x$, i.e. $(x)_+\equiv \max(0,x)$.

It turns out that the interplay between the $1$-exchange rule
and the stability and balance conditions results in highly non-trivial
predictions regarding the influence of network structure on
individual earnings.

\begin{figure}[ht]
\centering
\includegraphics[scale=0.5]{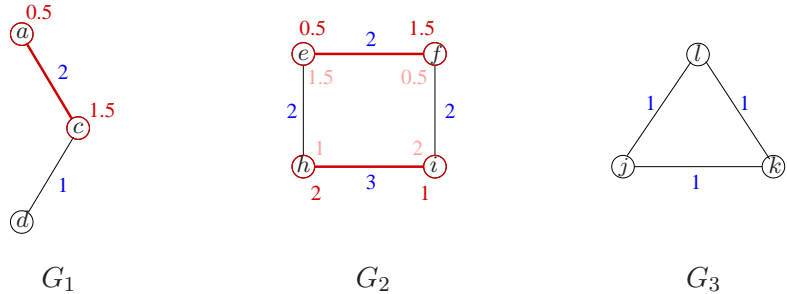}
\put(-282, -20){$G_1$}
\put(-163, -20){$G_2$}
\put(-38, -20){$G_3$}
\put(-291.5,74){{\footnotesize $a$}}
\put(-270,38){{\footnotesize $c$}}
\put(-291.5,2.5){{\footnotesize $d$}}
\put(-185,66){{\footnotesize $e$}}
\put(-135,66){{\scriptsize $f$}}
\put(-185.5,23.5){{\footnotesize $h$}}
\put(-135,23.5){{\footnotesize $i$}}
\put(-63.5,23.5){{\scriptsize $j$}}
\put(-7,23){{\footnotesize $k$}}
\put(-35,65.5){{\footnotesize $l$}}
\vskip10pt
\caption{Examples of networks and corresponding balanced outcomes. The network $G_1$ admits a unique balanced outcome, $G_2$ admits multiple balanced outcomes, and $G_3$ admits no balanced outcome. For $G_2$ one solution is shown inside the square and the other solution is outside.}
\label{fig:NB_examples}
\end{figure}

We conclude with some examples of networks and corresponding balanced outcomes (see Figure \ref{fig:NB_examples}).

 The network $G_1$ has a unique balanced outcome with the nodes $a$ and $c$ forming a partnership with a split of $\gamma_a = 0.5, \gamma_c = 1.5$. Node $d$ remains isolated with $\gamma_d = 0$. The best alternative of node $c$ is $(w_{cd} - \gamma_d)_+=1$, whereas it is $0$ for node $a$, and the excess of $2-1 = 1$ is split equally between $a$ and $c$, so that each earns a surplus of $0.5$ over their outside alternatives.

 The network $G_2$ admits multiple balanced outcomes. Each balanced outcome involves the pairing $M= \{(e,f), (h,i) \}$. The earnings $\gamma_e = 0.5, \gamma_f = 1.5, \gamma_h = 2, \gamma_i = 1$ are balanced, and so is the symmetric counterpart of this earnings vector $\gamma_e = 1.5, \gamma_f = 0.5, \gamma_h = 1, \gamma_i = 2$. In fact, every convex combination of these two earnings vectors is also balanced.

The network $G_3$ does not admit any stable outcome, and hence does not admit any balanced outcomes. To see this, observe that for any outcome, there is always a pair of agents who can benefit by deviating.

%
%
\subsection{Related work}

Recall the LP relaxation to the maximum weight matching problem
\begin{eqnarray}
\label{prob:mwm_relaxation}
\textup{maximize} && \sum_{(i,j) \in E} w_{ij} x_{ij},\\
\textup{subject to}&&
\sum_{j\in \di} x_{ij} \le 1 \;\;\; \forall i \in V,\nonumber\\
&& x_{ij}\ge 0\;\;\;\forall (i,j) \in E\,.\nonumber
\end{eqnarray}

The dual problem to (\ref{prob:mwm_relaxation}) is
\begin{eqnarray}
\label{prob:mwm_dual}
\textup{minimize} && \sum_{i \in V} y_i,\\
\textup{subject to}&&
y_i+y_j \ge w_{ij} \;\;\; \forall (i,j) \in E,\nonumber\\
&& y_i \ge 0 \;\;\;\forall i \in V. \nonumber
\end{eqnarray}
Stable outcomes were studied by Sotomayor \cite{Sotomayor}.
\begin{proposition}{\em \cite{Sotomayor}}
\label{prop:sotomayor}
Stable outcomes exist if and only if the linear programming relaxation \eqref{prob:mwm_relaxation}
of the maximum weight matching problem on $G$ admits an integral optimum.
Further, if $(M,\bargamma)$ is a stable solution then $M$ is a maximum weight matching and $\bargamma$ is an optimum solution to the dual LP \eqref{prob:mwm_dual}.
\end{proposition}

Following \cite{Rochford,CookY}, Kleinberg and Tardos
\cite{KT} first considered balanced outcomes on general exchange networks and proved that:
\emph{a network $G$ admits a balanced outcome if and only if it admits a stable outcome.}

The same paper describes a polynomial algorithm for constructing balanced outcomes. This is in turn based on the dynamic programming algorithm of Aspvall and Shiloach \cite{Aspvall}
for solving systems of linear inequalities. 
However, \cite{KT} left open the question of whether the actual bargaining process converges to balanced outcomes.

  Rochford \cite{Rochford}, and recent work by Bateni et al \cite{Bateni}, relate the assignment market problem to the extensive literature on cooperative game theory. They find that balanced outcomes correspond to the core intersect pre-kernel of the corresponding cooperative game. A consequence of the connection established is that the results of Kleinberg and Tardos \cite{KT}
  are implied by previous work in the economics literature. The existence result follows from Proposition \ref{prop:sotomayor}, and the fact that if the core of a cooperative game is non-empty then the core intersect prekernel is non-empty. Efficient computability follows from work by Faigle et al \cite{Faigle01}, who provide a polynomial time algorithm for finding balanced outcomes. \footnote{In fact, Faigle et al \cite{Faigle01} work in the more general setting of cooperative games. The algorithm involves local `transfers', alternating with a \textit{non-local} LP based step after every $O(n^2)$ transfers.}

  However, \cite{Rochford,Bateni} also leave open the twin questions of finding (i) a natural model for bargaining, and (ii) convergence (or not) to NB solutions.

Azar and co-authors \cite{Azar} studied the question as to whether a balanced outcome can be produced by a local dynamics,
and were able to answer it positively\footnote{Stearns \cite{Stearns} defined a very similar dynamics and proved convergence for general cooperative games. The dynamics can be interpreted in terms of the present model using the correspondence with cooperative games discussed in \cite{Bateni}.}. Their results left,
however, two outstanding challenges:
$(I)$ The algorithm analyzed by these authors
first selects a matching $M$ in $G$ using the message passing algorithm studied in \cite{Bayati,JHu07,BayatiB,SMW07}, corresponding to the
pairing of players that trade.
In a second phase the algorithm determines the profit of each player.
While such an algorithm
can be implemented in a distributed way, Azar et al. point out
that it is not entirely realistic. Indeed the rules of the dynamics change
abruptly after the matching is found. Further, if the pairing
is established at the outset, the players lose their
bargaining power;
$(II)$ The bound on the convergence time proved in \cite{Azar}
is exponential in the network size, and therefore does not
provide a solid justification for convergence to NB solutions in large
networks.

\subsection{Our contribution}
\label{subsec:contribution}

The present paper (based partly on our recent conference papers \cite{OurSODA,myWINE}) aims at tackling these challenges. First we
introduce a natural dynamics that is interpretable as a
realistic negotiation process. We show that the fixed points of
the dynamics are in one to one correspondence with NB
solutions, and prove that it converges to such solutions.
Moreover, we show that the convergence to approximate NB
solutions is fast. Furthermore we are able to treat the more
general case of nodes with unsymmetrical bargaining powers and
generalize the result of \cite{KT} on existence of NB solutions to this
context. These results are obtained through a new and
seemingly general analysis method, that builds on powerful
quantitative estimates on mappings in the Banach spaces
\cite{rate}. For instance, our approach allows us to prove that a
simple variant of the edge balancing dynamics of \cite{Azar}
converges in polynomial time (see Section \ref{sec:convergence}).

We consider various modifications to the model and analyze the results. One direction is to allow arbitrary integer `capacity constraints' that capture the maximum number of deals that a particular node is able to simultaneously participate in (the model defined above corresponds to a capacity of one for each node). Such a model would be relevant, for example, in the context of a job market, where a single employer may have more than one opening available.  We show that many of our results generalize to this model in Section \ref{sec:bmatching}.

Another well motivated modification is to depart from the assumption of symmetry/balance and allow nodes to have different `bargaining powers'. Rochford and Crawford \cite{RochfordCrawford86} mention this modification in passing, with the remark that  it ``\dots seems to yield no new
insights". Indeed, we show here (see Section \ref{sec:UD}) that our asymptotic convergence results generalize to the unsymmetrical case. However, surprisingly, we find that the natural dynamics may now take exponentially
long to converge. We find that this can occur even in a two-sided network with the `sellers' having slightly more bargaining power than the `buyers'. Thus, a seemingly minor change in the model appears to drastically change the convergence properties of our dynamics. Other algorithms like that of Kleinberg and Tardos \cite{KT} and Faigle et al \cite{Faigle01} also fail to generalize, suggesting that, in fact, we may lose computability of solutions in allowing asymmetry.
However, we show that a suitable modification to the bargaining process yields a fully polynomial time approximation scheme (FPTAS) for the unequal bargaining powers. The caveat is that this algorithm, though local, is not a good model for bargaining because it fixes the matching at the outset (cf. comment $(I)$ above).


Our dynamics and its analysis have similarities with a
series of papers on using max-product belief propagation for
the weighted matching problems \cite{Bayati,JHu07,BayatiB,SMW07}. We discuss that connection
and extensions of our results to those settings in one of our conference papers \cite[Appendix F]{OurSODA}. We obtain a class of new message passing algorithms to compute the maximum weight matching, with belief propagation and our dynamics being special cases.

\subsection*{Related work in sociology}

Besides economists, sociologists have been interested in such markets, called \emph{exchange networks} in that literature. The key question addressed by \emph{network exchange theory} is that of
how network structure influences the power balance between
agents. Numerous predictive frameworks have been suggested in this context including generalized Nash bargaining solutions \cite{CookY}.
Moreover, controlled experiments \cite{NET,Lucas,Skvoretz} have been carried out by
sociologists. The typical experimental set-up studies exactly the model of assignment markets proposed by economists \cite{ShapleyShubik72,Rochford}.
It is often the case that players are provided information
only about their immediate neighbors. Typically, a number of `rounds'
of negotiations are run, with no change in the network, so as to allow
the system to reach an `equilibrium'. Further, players are usually not provided much information beyond who their immediate neighbors are, and the value of the corresponding possible deals.

In addition to balanced outcomes \cite{CookY}, other frameworks have been suggested to predict/explain the outcomes of these experiments \cite{Kearns1,BraunGautschi06,Skvoretz}.

\subsection{Outline of the paper}

Our dynamical model of bargaining in a network is described in Section \ref{sec:natural_dynamics}. We state our main results characterizing fixed points and convergence of the dynamics in Section \ref{sec:main_results}.

We present two extensions of our model. In Section
\ref{sec:UD} we investigate the
unsymmetrical case with nodes having different bargaining
powers. We find that the (generalized) dynamics may take exponentially long to converge in this case, but we provide a modified local algorithm that provides an FPTAS.
We consider the case of general capacity constraints in Section \ref{sec:bmatching}, and show that our main results generalize to this case.


In Section \ref{sec:convergence}, we prove Theorems
\ref{thm:convergence_to_FP} and \ref{thm:rate_of_convergence}
on convergence of our dynamics. We characterize
fixed points in Section \ref{sec:FixedPoint} with a proof of
Theorem \ref{thm:fp_dualopt} (the proof of Theorem \ref{thm:approx_fp} is deferred
to Appendix \ref{sec:eps_fixed_point_prop_proofs}).  Section
\ref{sec:bipartite_pairing} shows polynomial time convergence
on bipartite graphs (proof of Theorem \ref{thm:bipartite_pairing}).

We present a discussion of our results in Section \ref{sec:discussion}.

Appendix
\ref{app:variations_of_dynamics} contains a discussion on
variations of the natural dynamics including time and node
varying damping factors and asynchronous updates.

%
%
\section{Dynamical model}
\label{sec:natural_dynamics}

%

Consider a bargaining network $G=(V,E)$, where the vertices represent agents, and the edges represent potential partnerships between them.  There is a positive weight
$w_{ij}>0$ on each edge $(i,j)\in E$, representing the fact that players connected by edge $(i,j)$ can
share a profit of $w_{ij}$ dollars if they agree to
trade with each other.
Each player can trade with at most one of her neighbors
(this is called the \emph{$1$-exchange rule}), so that
a set of valid trading pairs forms a matching $M$ in the graph $G$. We define a \emph{trade
outcome} as in Section \ref{sec:intro}, in accordance with the above constraints.

We expect natural dynamical description of a bargaining network to have the following properties: It should be
\emph{local},
i.e. involve limited information exchange along edges and processing
at nodes; It should be \emph{time invariant}, i.e. the players' behavior
should be the same/similar on identical local information at different times;
It should be \emph{interpretable}, i.e. the information  exchanged
along the edges should have a meaning for the players involved, and
should be consistent with reasonable behavior for players.

In the model we propose, at each time $t$, each player sends a message to
each of her neighbors. The message has the meaning of `best current
alternative'. We denote the message from player $i$ to player $j$
by $\alf{i}{j}^t$. Player $i$ is telling player $j$ that she (player $i$)
currently estimates earnings of $\alf{i}{j}^t$ elsewhere,
if she chooses not to trade with $j$.

The vector of all such messages is denoted by $\baralf^t \in \reals_+^{2|E|}$.
Each agent $i$ makes an `offer' to each of
her neighbors, based on her own `best alternative' and that of her neighbor. The offer from node $i$ to $j$ is denoted by $\off{i}{j}^t$
and is computed according to
\vspace{-0.5cm}

\begin{align}
    \off{i}{j}^t =
 (w_{ij}-\alf{i}{j}^t)_+ - \frac{1}{2}(w_{ij}-\alf{i}{j}^t-\alf{j}{i}^t)_+\, .
    \label{eq:off_def}
\end{align}
\vspace{-0.5cm}

It is easy to deduce that this definition corresponds to the following
policy: $(i)$ An offer is always non-negative, and a positive offer is never
larger than $w_{ij}-\alf{i}{j}^t$ (no player is interested in earning less
than her current best alternative);
$(ii)$ Subject to the above constraints, the surplus
$(w_{ij}-\alf{i}{j}^t-\alf{j}{i}^t)$ (if non-negative) is shared equally.
We denote by $\baroff^t\in\reals_+^{2|E|}$ the vector of offers.

Notice that $\baroff^t$ is just a deterministic
function of $\baralf^t$. In the rest
of the paper we shall describe the network status uniquely
through the latter vector,  and use $\baroff|_{\baralf^t}$ to
denote $\baroff^t$ defined by (\ref{eq:off_def}) when required
so as to avoid ambiguity.

Each node can estimate its potential earning based on the
network status, using
\begin{align}
\gamma_i^t \equiv \max_{k \in \di} \, \off{k}{i}^t,\label{eq:bargamma}
\end{align}
the corresponding vector being denoted by $\bargamma^t\in\reals_+^{|V|}$. Notice that $\bargamma^t$ is also a function of $\baralf^t$.

Messages are updated synchronously through the network,
according to the rule
\begin{align}
\alf{i}{j}^{t+1}= (1-\damp)\,\alf{i}{j}^{t} +
\damp \max_{k \in \di \backslash j}
\off{k}{i}^{t}\, .
\label{eq:update}
\end{align}
Here $\damp \in (0,1]$ is a `damping' factor: $(1-\damp)$ can be thought of as
the inertia on the part of the nodes to update their current estimates (represented by outgoing messages).
The use of  $\damp<1$ eliminates pathological behaviors
related to synchronous updates. In particular, we observe oscillations on even-length
cycles in the undamped synchronous version.
We mention here that in Appendix \ref{app:variations_of_dynamics} we present extensions of our results to various update schemes (e.g., asynchronous updates, time-varying damping factor).

\begin{Remark}\label{rem:cost_per_iteration}
An update under the natural dynamics requires agent $i$ to perform $O(|\partial i|)$ arithmetic operations, and $O(|E|)$ operations in total.
\end{Remark}

Let $\Wmax \equiv \max_{(ij)\in E}w_{ij}$. Often in the paper we
take $\Wmax=1$, since this can always be achieved by rescaling the
problem, which is the same as changing units.  It is easy to
see that $\baralf^{t} \in [0,\Wmax]^{2|E|}$,
$\baroff^t \in [0,\Wmax]^{2|E|}$ and $\bargamma^t \in [0,\Wmax]^{|V|}$
at all times (unless the initial condition violates this
bounds). Thus we call $\baralf$ a `valid' message vector if
$\baralf \in [0,\Wmax]^{2|E|}$.

\subsection{An Example}

We consider a simple graph $G$ with $V=\{A, B, C, D\}$, $E = \{(A,B), (B,C), (C,D)\}$, $w_{AB} = 8$, $w_{BC} = 6$
and $w_{CD} = 2$.
\begin{figure}
\label{fig:dynamics_eg}
\centering
\includegraphics[scale = 0.5]{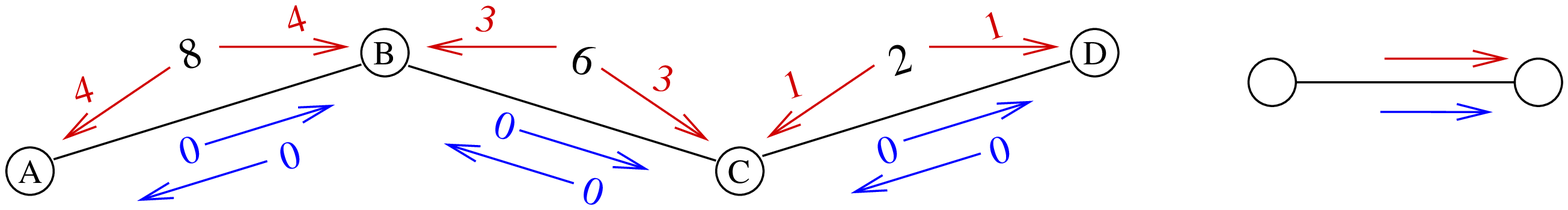}
\put(-72,28){{\footnotesize $i$}}
\put(-8.5,28){{\footnotesize $j$}}
\put(-44,14){\textcolor{blue}{{\small $\alf{i}{j}$}}}
\put(-43,42){\textcolor{red2}{{\small $\off{i}{j}$}}}
\put(-240,-30){\Large $\downarrow$}
\put(-430, 20){$t=0$}
\vskip15pt
\includegraphics[scale = 0.5]{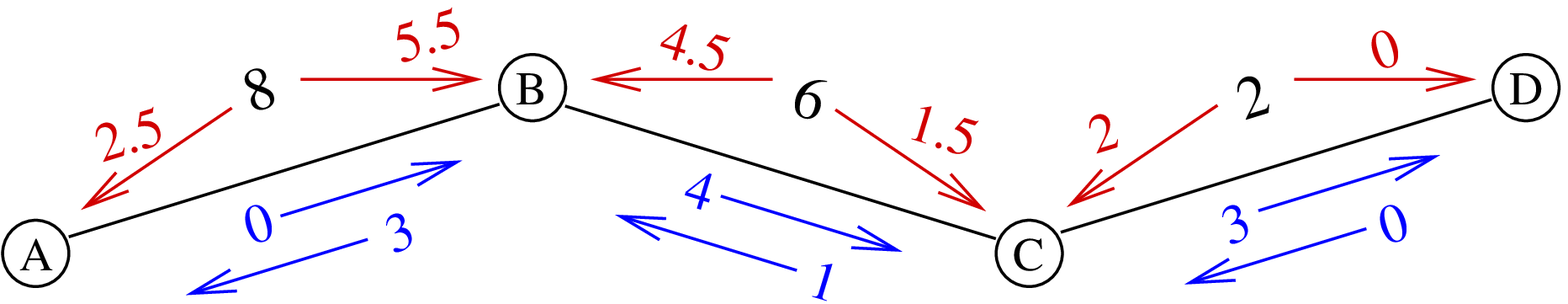}
\put(-240,-30){\Large $\downarrow$}
\put(-430, 20){$t=1$}
\vskip10pt
\includegraphics[scale = 0.5]{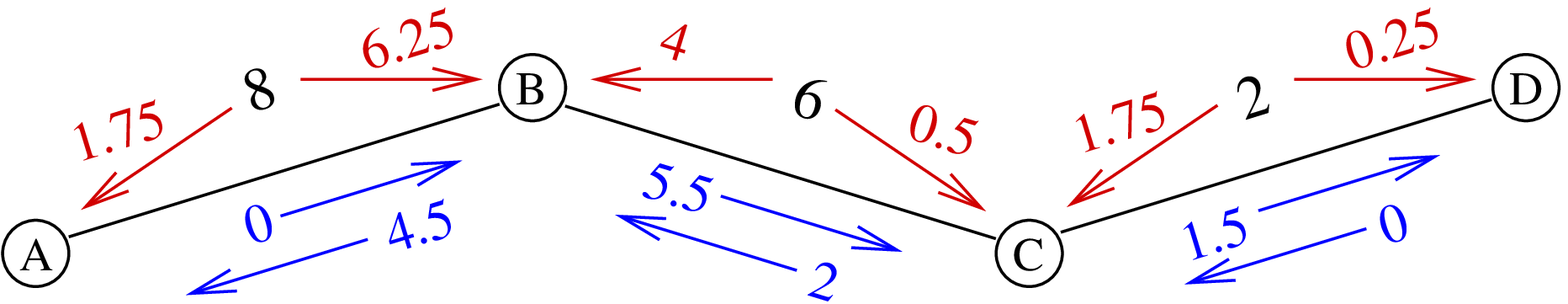}
\put(-430, 20){$t=2$}
\put(-240,-30){\Large $\downarrow$}
\put(-238.5,-55){\Large $\vdots$}
\vskip10pt
\includegraphics[scale = 0.5]{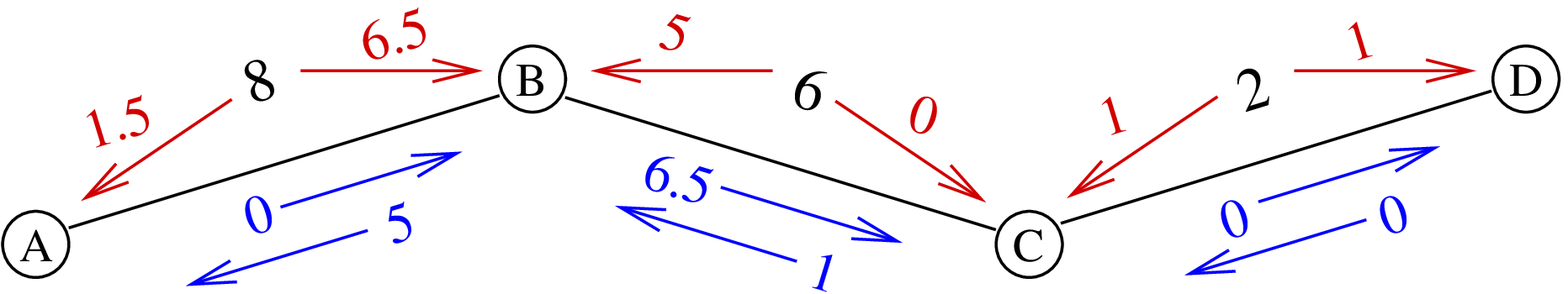}
\put(-430, 20){$t=6$}
\caption{Progress of the natural dynamics on a graph with four nodes and three edges. We (arbitrarily) choose the initialization $\baralf^0 = \underline{0}$. A fixed point is reached at $t=6$.}
\end{figure}
The unique maximum weight matching on this graph is $M= \{(A, B), (C, D) \}$. By Proposition \ref{prop:sotomayor}, stable outcomes correspond to matching $M$ and can be parameterized as
\begin{align*}
\bargamma &= (8-\gamma_B, \gamma_B, \gamma_C, 2-\gamma_C)
\end{align*}
where $(\gamma_B,\gamma_C)$ are constrained as
\begin{align*}
\gamma_B &\in [0,8]\\
\gamma_C &\in [0,2]\\
\gamma_B+\gamma_C &\geq 6
\end{align*}
For instance, the set of stable outcomes (all on matching $M$) includes $(0,8,2,0)$, $(4,4,2, 0)$, $(3,5,1,1)$ and so on. Now suppose we impose the balance condition Eq.~\eqref{eq:balance} in addition, i.e., we look for balanced outcomes. Using the algorithm of Kleinberg and Tardos \cite{KT}, we find that the network admits a unique balanced outcome $\bargamma = (1.5, 6.5, 1, 1)$.


Now we consider the evolution of the natural dynamics proposed above on the graph $G$. We arbitrarily choose to study the initialization $\baralf^0 = \underline{0}$, i.e., each node initially estimates its best alternatives to be $0$ with respect to each neighbor. We set $\damp = 1$ for simplicity\footnote{Our results assume $\damp<1$ to avoid oscillatory behavior. However, it turns out that  on graphs with no even cycles, for instance the graph $G$ under consideration, oscillations do not occur. We choose to consider $\damp = 1$ for simplicity of presentation.}. The evolution of the estimates and offers under the dynamics is shown  in Figure \ref{fig:dynamics_eg}. We now comment on a few noteworthy features demonstrated by this example. In the first step, nodes $A$ and $B$ receive their best offers from each other, node $C$ receives its best offer from $B$ and node $D$ receives its best offer from $C$. Thus, we might expect nodes $A$ and $B$ to be considering the formation of a partnership already (though the terms are not yet clear), but this is not the case for $C$ and $D$. After one iteration, at $t=1$, both pairs $(A, B)$ and $(C, D)$ receive their best offers from each other. In fact, this property remains true at all future times (the case $t=2$ is shown). However, the vectors $\baralf$ and $\baroff$ continue to evolve from one iteration to the next. At iteration $t=6$, a fixed point is reached, i.e., $\baralf$ and $\baroff$ remain unchanged for $t\geq 6$. Moreover, we notice that the fixed point captures the unique balanced outcome on this graph, with the matching $M$ and the splits $(\gamma_A = 1.5, \gamma_B=6.5)$ and $(\gamma_C = 1, \gamma_D = 1)$ emerging from the fixed point $\baroff^*$.

We remark here that convergence to a fixed point in finite number of iterations is not a general phenomenon. This occurs as a consequence of the simple example considered and the choice $\damp = 1$. However, as we prove below, we always obtain rapid convergence of the dynamics, and fixed points always correspond to balanced outcomes, on \emph{any graph} possessing balanced outcomes, and for \emph{any initialization}.

%
%
\section{Main results: Fixed point properties and convergence}
\label{sec:main_results}

Our first result is that fixed points of the update equations
(\ref{eq:off_def}), (\ref{eq:update}) (hereafter referred to as
`natural dynamics') are indeed in correspondence with Nash
bargaining solutions when such solutions exist. \textit{Note
that the fixed points are independent of the damping factor
$\damp$.} The correspondence with NB solutions includes
pairing between nodes, according to the following notion of
induced matching.
%

%
\begin{Definition}
\label{def:induces_matching}
We say that a state $(\baralf,
\baroff, \bargamma)$ (or just $\baralf$) \emph{induces a matching} $M$
if the following happens.  For each node $i\in V$ receiving
non-zero offers ($\off{\cdot}{i}>0$), $i$ is matched under $M$ and gets its unique
best offer from node $j$ such that $(i,j) \in M$. Further, if
$\gamma_i = 0$ then $i$ is not matched in $M$. In other words,
pairs in $M$ receive unique best offers that are positive from
their respective matched neighbors whereas unmatched nodes
receive no non-zero offers.
\end{Definition}

Consider the LP relaxation to the maximum weight matching problem
\eqref{prob:mwm_relaxation}.
A feasible point $\underline{x}$ for LP \eqref{prob:mwm_relaxation} is called \emph{half-integral} if for all $e\in E$, $x_e\in\{0,1,\frac{1}{2}\}$. It is well known that problem (\ref{prob:mwm_relaxation}) always has an optimum $\underline{x}^*$ that is half-integral \cite{Schrijver}.
An LP with a fully integer $\underline{x}^*$ ($x_e^*\in\{0,1\}$) is called \emph{tight}.

\begin{theorem}\label{thm:fp_dualopt}
Let $G$ be an instance admitting one or more Nash bargaining solutions,  i.e. the LP (\ref{prob:mwm_relaxation})
admits an integral optimum.\\
(a) \emph{Unique LP optimum (generic case):} Suppose the optimum is unique corresponding to matching $M^*$.
Let $(\baralf,\baroff,\bargamma)$ be a fixed point of
the natural dynamics. Then $\baralf$
induces matching $M^*$ and $(M^*, \bargamma)$ is a Nash bargaining solution.
Conversely, every Nash bargaining solution $(M,\bargamma_{\NB})$ has $M=M^*$ and
corresponds to a unique
fixed point of the natural dynamics with $\bargamma=\bargamma_{\NB}$.\\
\\
(b) Let $(\baralf,\baroff,\bargamma)$ be a fixed point of
the natural dynamics. Then $(M^*, \bargamma)$ is a Nash bargaining solution for any integral maximum weight matching $M^*$. Conversely, if $(M,\bargamma_{\NB})$ is a  Nash bargaining solution, $M$ is a maximum weight matching and there is a unique fixed point of the natural dynamics  with $\bargamma=\bargamma_{\NB}$.
\end{theorem}

We prove Theorem \ref{thm:fp_dualopt} in Section
\ref{sec:FixedPoint}. Theorem \ref{thm:fp_dualopt_unmatched} in
Appendix \ref{sec:fixed_point_prop_proofs} extends this
characterization of fixed points of the natural dynamics to
cases where Nash bargaining solutions do not exist.

\begin{Remark}
The condition that a tight LP \eqref{prob:mwm_relaxation} has
a unique optimum is generic (see Appendix \ref{sec:fixed_point_prop_proofs}, Remark \ref{rem:unique_opt_generic}).
Hence, fixed points induce a matching for almost all instances (cf. Theorem \ref{thm:fp_dualopt}(a)).
Further, in the non-unique optimum case, we cannot expect an induced matching,
since there is always some node with two equally good alternatives.
\end{Remark}

The existence of a fixed point of the natural dynamics is
immediate from Brouwer's fixed point theorem. Our next result
says that the natural dynamics always converges to a fixed
point. The proof is in Section \ref{sec:convergence}.

\begin{theorem}\label{thm:convergence_to_FP}
The natural dynamics has at least one fixed point. Moreover, for any initial condition with $\baralf^0\in[0,W]^{2|E|}$, $\baralf^t$ converges to a fixed point.
\end{theorem}

Note that
Theorem \ref{thm:convergence_to_FP} does not require any condition on LP (\ref{prob:mwm_relaxation}). It also does not require uniqueness of the fixed point.

With Theorems \ref{thm:fp_dualopt} and \ref{thm:convergence_to_FP}, we know that in the limit of a large number of iterations, the natural dynamics yields a Nash bargaining solution. However, this still leaves unanswered the question of the rate of convergence of the natural dynamics. Our next theorem
addresses this question, establishing fast convergence to an approximate fixed point.

However, before stating the theorem we define the notion of approximate fixed point.

\begin{Definition}
\label{def:approx_FP}
We say that $\baralf$ is an \emph{$\eps$-fixed point}, or \emph{$\eps$-FP} in short, if, for all $(i,j) \in E$ we have
\begin{align}\label{eq:eps-fp}
\big |\alf{i}{j} - \max_{k \in \partial i \backslash j} \off{k}{i} \big| \; \leq \; \eps \, ,
\end{align}
and similarly for $\alf{j}{i}$. Here, $\baroff$ is obtained from $\baralf$ through Eq.~\eqref{eq:off_def} (i.e., $\baroff = \baroff |_{\baralf}$).
\end{Definition}
Note that $\eps$-fixed points are also defined independently of the
damping $\damp$.

\begin{theorem}
\label{thm:rate_of_convergence}
Let $G=(V, E)$ be an instance with weights  $(w_e, e \in E) \in [0, 1]^{|E|}$. Take any initial condition
$\baralf^0\in[0,1]^{2|E|}$.  Take any $\eps > 0$. Define
\begin{align}
T^*(\eps) = \frac{1}{\pi \damp (1- \damp) \eps^2}\, .
\label{eq:convergence_time_Teps}
\end{align}
Then for all $t \geq T^*(\eps)$,  $\baralf^t$ is an $\eps$-fixed point. (Here $\pi = 3.14159\ldots$)
\end{theorem}

Thus, if we wait until time $t$, we are guaranteed to obtain an $\left(1/\sqrt{\pi \damp (1- \damp) t}\right)$-FP.
Theorem \ref{thm:rate_of_convergence} is proved in Section \ref{sec:convergence}.

\begin{Remark}
For any $\eps > 0$, it is possible to construct an example such that it takes $\Omega(1/\eps)$
iterations to reach an $\eps$-fixed point. This lower bound can be improved to $\Omega(1/\eps^2)$ in the unequal
bargaining powers case (cf. Section \ref{sec:UD}). However, in our constructions, the size of the example graph grows with decreasing
$\eps$ in each case.
\end{Remark}

We are left with the problem of relating approximate fixed points to approximate Nash bargaining solutions. We use the following definition of $\eps$-Nash bargaining solution, that is analogous to the standard definition of $\eps$-Nash equilibrium (e.g., see \cite {approx_NE}).
\begin{Definition}
We say that $(M,\bargamma)$ is an \emph{$\eps$-Nash bargaining solution} if it is a valid trade outcome that is stable and satisfies $\eps$-balance. $\eps$-balance means that for every $(i,j) \in M$ we have
\begin{align}
\big | [\gamma_i - \max_{k \in \di \backslash j} (w_{ik}\!-\! \gamma_k)_+ ] -[\gamma_j - \max_{l \in \dj \backslash i} (w_{jl}\!- \!\gamma_l)_+] \big| \leq \eps \, .
\label{eq:eps_balance}
\end{align}
\end{Definition}


A subtle issue needs to be addressed. For an approximate fixed point to
yield an approximate Nash bargaining solution, a suitable
pairing between nodes is needed. Note that our dynamics does
not force a pairing between the nodes. Instead, a pairing
should emerge quickly from the dynamics. In other words, nodes
on the graph should be able to identify their trading partners
from the messages being exchanged. As before, we use the notion of an
induced matching (see Definition \ref{def:induces_matching}).

\begin{Definition}
\label{def:LP_gap}
Consider LP \eqref{prob:mwm_relaxation}. Let $\mathcal{H}$ be the set of half integral points in the primal polytope. Let $\underline{x}^* \in \mathcal{H}$ be an optimum. Then the LP gap $g$ is defined as $g= \min_{\underline{x} \in \mathcal{H}\backslash \{\underline{x}^*\}} \sum_{e \in E} w_e x_e^*-\sum_{e \in E}w_e x_e$.
\end{Definition}


\begin{theorem}\label{thm:approx_fp}
Let $G$ be an instance for which the LP
(\ref{prob:mwm_relaxation}) admits a unique optimum, and this
is integral, corresponding to matching $M^*$. Let the gap be
$g>0$. Let $\baralf$ be an $\eps$-fixed point of the natural
dynamics, for some $\eps < g/(6n^2)$. Let $\bargamma$ be the
corresponding earnings estimates. Then $\baralf$ induces the
matching $M^*$ and $(\bargamma, M^*)$ is an $(6\eps)$-Nash bargaining
solution. Conversely, every $\eps$-Nash bargaining solution
$(M,\bargamma_{\NB})$ has $M=M^*$ for any $\eps>0$.
\end{theorem}
Note that $g>0$ is equivalent to the unique optimum condition (cf. Remarks 2, 5). The proof of this theorem requires generalization of the
analysis used to prove Theorem \ref{thm:fp_dualopt} to the case
of approximate fixed points. Since its proof is similar to the proof of Theorem \ref{thm:fp_dualopt}, we defer it to Appendix \ref{sec:eps_fixed_point_prop_proofs}. We stress, however, that Theorem \ref{thm:approx_fp} is not, in any sense, an obvious
strengthening of Theorem \ref{thm:fp_dualopt}. In fact, this is a delicate property of approximate fixed points
that holds {\em only in the case of balanced outcomes}. This characterization breaks down in the face
of a seemingly benign generalization to unequal bargaining powers
 (cf. Section \ref{sec:UD} and \cite[Section 4]{myWINE}).

Theorem \ref{thm:approx_fp} holds for all graphs, and is, in a
sense, the best result we can hope for. To see this, consider
the following immediate corollary of Theorems
\ref{thm:rate_of_convergence} and \ref{thm:approx_fp}.
\begin{corollary}
 \label{coro:fast_matching}
Let $G=(V, E)$ be an instance with weights  $(w_e, e \in E) \in [0, 1]^{|E|}$. Suppose LP
(\ref{prob:mwm_relaxation}) admits a unique optimum, and this is integral,
corresponding to matching $M^*$. Let the gap be $g>0$. Then for
any $\baralf^0 \in [0,1]^{2|E|}$, there exists $T^*= O(n^4/g^2)$ such
that for any $t \geq T^*$, $\baralf^t$ induces the matching $M^*$
and $(\bargamma^t, M^*)$ is an $(6/\sqrt{\pi \damp(1-\damp)t})$-NB solution.
\end{corollary}
\begin{proof}
Choose $T^*$ as $T^*(g/(10n^2))$ as defined in \eqref{eq:convergence_time_Teps}. Clearly, $T^*=O(n^4/g^2)$. From Theorem \ref{thm:rate_of_convergence}, $\baralf^t$ is an $\eps(t)$-FP for $\eps(t) = 1/\sqrt{\pi \damp(1-\damp)t}$. Moreover, for all $t\geq T^*$, $\eps(t) \leq g/(10n^2)$. Hence, by Theorem \ref{thm:approx_fp}, $\baralf^t$ induces the matching $M^*$ and  $(\bargamma^t, M^*)$ is a $(6\eps(t))$-NB solution for all $t\geq T^*$.
\end{proof}

Corollary \ref{coro:fast_matching} implies that for any $\eps > 0 $, the natural dynamics finds an $\eps$-NB solution in time $O\left( \max\left(n^4/g^2 , 1/\eps^2\right)\right)$.

This result is the essentially the strongest bound we can hope
for in the following sense. First, note that we need to find
$M^*$ (see converse in Theorem \ref{thm:approx_fp})
\textit{and} balance the allocations. Max product belief
propagation, a standard local algorithm for computing the
maximum weight matching, requires $O(n/g)$ iterations to
converge, and \textit{this bound is tight} \cite{BayatiB}.
Similar results hold for the Auction algorithm \cite{auction}
which also locally computes $M^*$. Moreover, max product BP and
the natural dynamics are intimately related (see \cite{OurSODA}), with the exception that \textit{max
product is designed to find $M^*$, but this is not true for the
natural dynamics}. Corollary \ref{coro:fast_matching} shows
that natural dynamics only requires a time that is polynomial
in the same parameters $n$ and $1/g$ to find $M^*$, while it
simultaneously takes rapid care of balancing the outcome.

\subsection{Example: Polynomial convergence to $\eps$-NB solution on bipartite graphs.}
The next result further shows a concrete setting in which Corollary \ref{coro:fast_matching} leads to a strong guarantee on quickly reaching an approximate NB solution.

%
\begin{theorem}
\label{thm:bipartite_pairing} Let $G=(V, E)$ be a bipartite
graph with weights  $(w_e, e \in E) \in [0, 1]^{|E|}$. Take any
$\xi \in (0,1), \eta \in (0, 1)$.
Construct a perturbed problem instance with weights $\bar{w}_e = w_e + \eta U_e$, where $U_e$ are
independent identically distributed
random variables uniform in $[0,1]$. Then there exists $C=C(\damp)<\infty$, such that for
\begin{align}
T^* = C \left( \frac{n^2 |E|}{\eta \xi} \right)^2 ,
\end{align}
the following
happens for all
$t \geq T^*$ with probability at least $1-\xi$.
State $\baralf^t$ induces a matching $M$ that is independent of $t$. Further, $(\bargamma^t, M)$ is a $\eps(t)$-NB solution for the perturbed problem, with $\eps(t)= 12/\sqrt{\pi \damp(1-\damp)t}$.
\end{theorem}

$\xi$ represents our target in the probability that a pairing does not emerge, while
$\eta$ represents the size of perturbation of the problem instance.

Theorem \ref{thm:bipartite_pairing} implies that for any fixed $\eta$ and $\xi$, and any $\eps>0$,  we find an $\eps$-NB solution in time
$\tau(\eps)=K\max(n^4 |E|^2, 1/\eps^2)$  with probability at least $1-\xi$, where $K = K(\eta, \xi, \damp) < \infty$. Theorem \ref{thm:bipartite_pairing} is proved in Section \ref{sec:bipartite_pairing}.
%

\subsection{Other results}
A different analysis allows us to prove {\em exponentially fast convergence}
to a unique Nash bargaining solution. We describe this briefly
in Section \ref{sec:exp_fast_convergence}, referring to an unpublished manuscript \cite{OurPrevious} for the proof, in the interest of space.

We investigate the case of \textit{nodes with unsymmetrical bargaining powers} in Section \ref{sec:UD}. We show that generalizations of the Theorems
\ref{thm:fp_dualopt}, \ref{thm:convergence_to_FP} and \ref{thm:rate_of_convergence} hold for a suitably modified dynamics. Somewhat surprisingly, we find that the modified dynamics may take exponentially long to converge (so Theorem \ref{thm:rate_of_convergence} does \emph{not} generalize). However, we find a different local procedure that efficiently finds approximate solutions.

In Section \ref{sec:bmatching} we consider the case where agents have arbitrary integer capacity constraints on the number of partnerships they can participate in, instead of the one-matching constraint. We generalize our dynamics and the notion of balanced outcomes to this case. We show that Theorems \ref{thm:fp_dualopt}, \ref{thm:convergence_to_FP} and \ref{thm:rate_of_convergence} generalize. As a corollary, we establish the existence of balanced outcomes whenever stable outcomes exist (Corollary \ref{coro:b_balance_existence}) in this general setting\footnote{The caveat here is that Corollary \ref{coro:b_balance_existence} does not say anything about the corner case of non-unique maximum weight $\mathbf b$-matching.}.

Appendix \ref{app:variations_of_dynamics} presents extensions of our convergence results to cases where the damping factor varies in time or from node to node, and when updates are asynchronous. This shows that our insights are robust to variations in the natural of damping and the timing of iterative updates.

\section{Unequal bargaining powers}
\label{sec:UD}


It is reasonable to expect that not all edge surpluses on
matching edges are divided equally in an exchange network
setting. Some nodes are likely to have more `bargaining power'
than others. This bargaining power can arise, for example, from
`patience'; a patient agent is expected to get more than half
the surplus when trading with an impatient partner. This
phenomenon is well known in the Rubinstein game
\cite{sequential_bargaining} where nodes alternately make
offers to each other until an offer is accepted -- the node
with a smaller discount factor earns more in the subgame perfect Nash equilibrium.
Moreover, a recent experimental study of bargaining in exchange
networks \cite{Kearns_exp} found that patience correlated
positively with earnings.

A reasonable approach to model this effect would be to assign a positive
`bargaining power' to each node, and
postulate that if a pair of nodes trade with each other, then the
edge surplus is divided in the ratio of their bargaining powers. We choose
instead, a more general setting where on each edge $(ij)$ there
is an expected surplus split fraction quantified by
$r_{ij}\in(0,1)$. Namely, $r_{ij}$ is the fraction of surplus that goes to $i$
if $i$ and $j$ trade with each other, and similarly for $r_{ji}$.
Note that we have $r_{ij} + r_{ji}=1$. We
call a weighted graph $G$ along with the postulated split
fraction vector $\underline{r}$ an \textit{unequal division
(UD)} instance.

The balance condition is replaced by \textit{correct division}
condition
\begin{align}
\label{eq:correct_division}
[r_{ij}]^{-1}&\big[\gamma_i - \max_{k \in \di \backslash j}
(w_{ik}- \gamma_k)_+\big ]\\
&\stackrel{\text{\tiny UD}}{\hbox{\equalsfill}}\,
[r_{ji}]^{-1}\big[\gamma_j - \max_{l \in \dj \backslash i} (w_{jl}- \gamma_l)_+\big]\, ,\nonumber
\end{align}
on matched edges $(ij)$. We retain the stability condition. We
call  trade outcomes satisfying (\ref{eq:correct_division}) and
stability \textit{unequal division (UD) solutions}. A natural modification to our dynamics in this situation
consists of the following redefinition of offers.
\begin{align}
    \off{i}{j}^t \stackrel{\text{\tiny UD}}{\hbox{\equalsfill}}
 (w_{ij}-\alf{i}{j}^t)_+ - r_{ij}(w_{ij}-\alf{i}{j}^t-\alf{j}{i}^t)_+\, .
    \label{eq:unequal_off_def}
\end{align}
We call the dynamics resulting from \eqref{eq:unequal_off_def} and the update rule \eqref{eq:update} the \emph{unsymmetrical natural dynamics}.
One can check that $\sT$ defined in (\ref{eq:Tdef}) is
nonexpansive for offers defined as in
(\ref{eq:unequal_off_def}). {\em It follows that Theorems
\ref{thm:convergence_to_FP} and \ref{thm:rate_of_convergence}
hold for the UD-natural dynamics with damping.} (We retain Definition \ref{def:approx_FP} of an $\eps$-FP).
Further, Theorem \ref{thm:fp_dualopt}
can also be extended to this case.
The proof involves exactly the same
steps as for the natural dynamics (cf. Section \ref{sec:FixedPoint}). Properties 1-6 in the direct part all hold (proofs nearly verbatim) and an identical construction works for the converse.

\begin{theorem}\label{thm:unequal_fp_UDsoln}
Let $G$ be an instance for which the LP (\ref{prob:mwm_relaxation})
admits an integral optimum.
Let $(\baralf,\baroff,\bargamma)$ be a fixed point of
the UD-natural dynamics. Then  $(M^*, \bargamma)$ is a UD solution for any maximum weight matching $M^*$.
Conversely, for any UD solution $(M,\bargamma_{\textup{\tiny UD}})$, matching $M$ is a maximum weight matching and
there is a unique
fixed point of the UD-natural dynamics with $\bargamma=\bargamma_{\textup{\tiny UD}}$.

Further, if the LP \eqref{prob:mwm_relaxation} has a \emph{unique} integral optimum, corresponding to matching $M^*$, then any fixed point $\baralf$
induces matching $M^*$.
\end{theorem}

We note that the following generalization of the result on existence of Nash bargaining solutions \cite{KT} follows from Theorem \ref{thm:unequal_fp_UDsoln} and the existence of fixed points.
\begin{lemma}
\label{lemma:UD_existence}
UD solutions exist if and only if a stable outcome exists (i.e. LP \eqref{prob:mwm_relaxation} has an integral optimum.)
\end{lemma}

\begin{proof}
The direct part of Theorem \ref{thm:unequal_fp_UDsoln}, along with the existence of fixed points of the UD natural dynamics (from Brouwer's fixed point theorem, also first part of Theorem \ref{thm:convergence_to_FP} for UD) shows that UD solutions exist if LP \eqref{prob:mwm_relaxation} has an integral optimum. The converse is trivial since if LP \eqref{prob:mwm_relaxation} has no integral optimum, then there are no stable solutions (see Proposition \ref{prop:sotomayor}) and hence no UD solutions.
\end{proof}

\subsection{Exponential convergence time in the UD case}
\label{subsec:UD_slow_convergence}

It is possible to derive
a characterization similar to Theorem \ref{thm:approx_fp} also for the UD case. However, the bound
on $\eps$ needed to ensure that  the right pairing emerges in an $\eps$-FP turns out to be
exponentially small in $n$. As such, we are only able to show that a pairing emerges in time
$2^{O(n)}/g^2$. In fact, as the example below shows, it {\em does} take exponentially long
for a pairing to emerge in worst case.

Let $n=|V|$. In Appendix \ref{subapp:UD_construction}, we construct a sequence of instances $(I_n, n \geq 16)$, such that for
each instance in the sequence the following holds.
\begin{enumerate}[(a)]
\item The instance admits a UD solution, and the gap $g \geq 1$.
\item There is a message vector $\baralf$ that is an $\eps$-FP for $\eps =  2^{-cn}$ such that $\baralf$ does \emph{not} induce any matching. In fact, for any matching in our example, there exists $j$ such that $j$ gets an offer from outside the matching that exceeds the offer of its partner under the matching (if any) by at least 1. Thus, $\baralf$ is `far' from inducing a matching.
\end{enumerate}
Further, split fractions are bounded within $[r, 1-r]$ for arbitrary desired
$r \in (0, 1/2)$ ($c$ depends on $r$). Also, the weights are uniformly bounded by a constant $W(r)$.

Now, if we start at an $\eps$-FP, the successive iterates of our dynamics are all $\eps$-FPs (this follows from the nonexpansivity of the update operator, cf. Section \ref{sec:FixedPoint}). Hence, no offer can change by more than $\eps$ in each iteration. Thus, in our constructed instances $I_n$, it requires at least $2^{\Omega(n)}$ iterations starting with $\baralf^0 = \baralf$, before $\baralf^t$ induced a matching. Thus, our construction with the above properties implies that the unsymmetrical natural dynamics can take exponentially long to induce a matching, even on well behaved instances ($g = \Omega(1)$). Moreover, as discussed in Appendix \ref{subapp:UD_construction}, our construction corresponds to a plausible two-sided market, so this is not to be dismissed as a unrealistic special case that can be ignored.

\subsection{A fully polynomial time approximation scheme}
\label{subsec:UD_FPTAS}

We show that though the natural dynamics may take exponentially long to converge, there is a polynomial time iterative `re-balancing' algorithm that enables us to compute an approximate UD solution.

First we define an approximate version of correct division, asking that
Eq.~\eqref{eq:correct_division} be satisfied to within an additive $\eps$, for all matched edges.
For each edge $(ij) \in M$,  we define the
`edge surplus' as the excess of $w_{ij}$ over the sum of best alternatives, i.e.,
\begin{align}
\surp_{ij}(\bargamma) =  w_{ij} - \max_{k \in \partial i \backslash j} (w_{ik} - \gamma_k)_+
- \max_{l \in \partial j \backslash i} (w_{jl} - \gamma_l)_+ \, .
\label{eq:surplus_defn}
\end{align}

\begin{Definition}[$\eps$-Correct division]
\label{def:eps_correct_div}
An outcome $(\bargamma, M)$ is said to satisfy \emph{$\eps$-correct division} if, for all $(ij) \in M$,
\begin{align}
|\gamma_i - \max_{k \in \partial i \backslash j} (w_{ik} - \gamma_k)_+ - r_{ij} \surp_{ij}(\bargamma)| \leq \eps \, ,
\label{eq:eps_correct_div}
\end{align}
where $\surp_{ij}(\cdot)$ is defined by Eq.~\eqref{eq:surplus_defn}.
\end{Definition}

We define approximate UD solutions as follows:
\begin{Definition}[$\eps$-UD solution]
An outcome $(\bargamma,M)$ is an \emph{$\eps$-UD solution} for $\eps \geq 0$
if it is \emph{stable} and it satisfies \emph{$\eps$-correct division} (cf. Definition \ref{def:eps_correct_div}).
\end{Definition}

It follows from Lemma \ref{lemma:UD_existence} that $\eps$-UD solutions exist if and only if the LP
\eqref{prob:mwm_relaxation} admits an integral optimum, which is the same as the requirement for
existence of UD solutions. We prove the following:
\begin{theorem}
\label{thm:FPTAS}
There is an  algorithm that is polynomial in the input and $1/\eps$, such that for any problem instance with weights uniformly bounded by $1$,
i.e., $(w_e, e \in E) \in (0,1]^{|E|}$:
\begin{itemize}
\item If the instance admits a UD solution, the algorithm finds an $\eps$-UD solution.
\item If the instance does not admit a UD solution, the algorithm
returns the message {\sc unstable}.
\end{itemize}
\end{theorem}

Our approach to finding an $\eps$-UD solution consists of two main steps:
\begin{enumerate}
\item Find a maximum weight matching $M^*$ and a dual optimum $\bargamma$ (solution to the dual LP \eqref{prob:mwm_dual}) .
Thus, form a stable outcome $(\bargamma,M^*)$. Else certify that the instance has no UD solution.
\item Iteratively update the allocation $\bargamma$ without changing the matching.
Updates are local, and are designed to converge fast to an allocation satisfying the $\eps$-correct division solution
\emph{while maintaining stability}. Thus, we arrive at an $\eps$-UD solution.
\end{enumerate}
As mentioned earlier, this is similar to the approach of \cite{Azar}. The crucial differences (enabling our results)
are:  (i) we stay within the space of stable outcomes, (ii) our analysis of convergence.

The algorithm leading to Theorem \ref{thm:FPTAS} and the corresponding analysis are presented in Appendix \ref{subsec:UD_FPTAS}. We remark that our algorithm can easily be made \emph{local}, without sacrificing efficiency (see \cite{myWINE} for details).

\section{General capacity constraints}
\label{sec:bmatching}

In several situations, agents may be less restricted: Instead of an agent being allowed to enter at most one agreement, for each agent $i$, there may be an integer \emph{capacity} constraint $b_i$ specifying the maximum number of partnerships that $i$ can enter into. For instance, in a labor market for full time jobs, an employer $j$ may have 4 openings for a particular role ($b_j=4$), another employer may have 6 openings for a different role, and so on, but the job seekers can each accept at most one job.
In this section, we describe a generalization of our dynamical model to the case of general capacity constraints, in an attempt to model behavior in such settings. We find that most of our results from the one-matching case, cf. Section \ref{sec:main_results}, generalize.

\subsection{Preliminaries}

Now a bargaining network is specified by an undirected graph $G=(V,E)$ with positive weights on the edges $(w_{ij})_{(ij) \in E}$, and integer capacity constraints associated to the nodes $(b_i)_{i \in V}$. We generalize the notion of `matching' to sets of edges that satisfy the given capacity constraints:
Given capacity constraints $\mathbf b=(b_i)$,  we call a set of edges $M \subseteq E$ a \emph{$\mathbf b$-matching} if the degree $d_i(M)$ of $i$ in the graph $(V,M)$ is at most $b_i$, for every $i\in V$. We say that $i$ is \emph{saturated} under $M$ if $d_i(M)=b_i$.

We assume that there are no double edges between nodes\footnote{This assumption was not needed in the one-exchange case since, in that case,  utility maximizing agents $i$ and $j$ will automatically discard all but the heaviest edge between them. This is no longer true in the case of general capacity constraints.}. Thus, an agent can use at most one unit of capacity with any one of her neighbors in the model we consider.

A \emph{trade outcome} is now a pair $(M,\Gamma)$, where $M$ is a $\mathbf b$-matching
and $\Gamma\in [0,1]^{2|E|}$ is a splitting of profits $\Gamma = (\gamma_{i \to j}, \gamma_{j \to i})_{(ij) \in E}$, with
$\gamma_{i\to j}=0$ if $(ij)\notin M$, and $\gamma_{i\to j}+\gamma_{j\to i}=w_{ij}$ if $(ij)\in M$.

Define $\gamma_i=\min_{j:(ij) \in M} \gamma_{j\to i}$ if $i$ is saturated (i.e. $d_i(M)=b_i$)
and $\gamma_i=0$ if $i$ is not saturated. Note that this definition is equivalent to $\gamma_i=\bmax{i}_{j \in \partial i} \gamma_{j\to i}$. Here $\bmax{}: \reals_+^* \rightarrow \reals_+$ denotes the $b$-th largest of a set of non-negative reals, being defined as $0$ if there are less than $b$ numbers in the set. It is easy to see that our definition of $\gamma_i$ here is consistent with the definition for the one-exchange case. (But $\Gamma$ is not consistent with $\bargamma$, which is why we use different notation.)


We say that a trading outcome is \emph{stable} if $\gamma_i+\gamma_j\geq w_{ij}$ for all $ij\notin M$. This definition is natural; a selfish agent would want to switch partners if and only if he can gain more utility elsewhere.

An outcome $(M, \gamma)$ is said to be \emph{balanced} if
\begin{align}
\gamma_{j \to i} - \bmax{i}_{k \in \partial i \backslash j} (w_{ik} -\gamma_k)_+
=
\gamma_{i \to j} - \bmax{j}_{l \in \partial j \backslash i} (w_{jl} -\gamma_l)_+
\label{eq:b_balance}
\end{align}
for all $(ij)\in M$.

Note that the definitions of stability and balance generalize those for the one-exchange case.

An outcome $(M, \Gamma)$ is a Nash bargaining solution if it is both stable and balanced.

Consider the problem of finding the maximum weight (not necessarily perfect) $\mathbf b$-matching on a weighted graph $G=(V,E)$.  The LP-relaxation of this problem and it's dual
are given by
\begin{equation}
\begin{array}{rcclcrccl}
&&&&&&&&\\
\textrm{max }&&\sum_{(ij)\in E}x_{ij}w_{ij}&&|&\textrm{min }&&\sum_{i\in V} b_iy_i+\sum_{{(ij)\in E}} y_{ij}&\\
\textrm{subject to}&&&&|&\textrm{subject to}&&&\\
&&\sum_{j\in N(i)} x_{ij} \leq b_i&\forall~~ i&|&&&y_{ij}+y_i + y_j -w_{ij}\geq 0&\forall~~ (ij)\in E\\
&&0\leq x_{ij} \leq 1&\forall~(ij)\in E&|&&&y_{ij}\geq0&\forall~(ij)\in E\\
&&&&|&&&y_{i}\geq0&\forall~i\in V\\
&&&&|&&&&\\
&&\textrm{Primal LP}&&|&&&\textrm{Dual LP}.&\\
\end{array}
\label{prob:b_LP_primal_dual}
\end{equation}
\vspace{2mm}
Complementary slackness says that a pair of feasible solutions is optimal if and only if
\begin{itemize}
\item For all $ij\in E;~~~x_{ij}^*(-w_{ij}+y_{ij}^*+y_i^*+y_j^*)=0$.
\item For all $ij\in E;~~~(x_{ij}^*-1)y_{ij}^*=0$.
\item For all $i\in V;~~~(\sum_{j\in N(i)}x_{ij}^*-b_i)y_i^*=0$.
\end{itemize}

\begin{lemma}
Consider a network $G=(V,E)$ with edge weights $(w_{ij})_{(ij)\in E}$ and capacity constraints $\mathbf b=(b_i)$. There exists a stable solution if and only if the primal LP \eqref{prob:b_LP_primal_dual} admits an integer optimum. Further, if $(M, \Gamma)$ is a stable outcome, then $M$ is a maximum weight $\mathbf b$-matching, and $y_i = \gamma_i$ for all $i \in V$ and $y_{ij} = (w_{ij} - y_i -y_j)_+$ for all $(ij) \in E$ is an optimum solution to the dual LP.
\label{lemma:b_stable_characterization}
\end{lemma}

\begin{proof}
If $x^*$ is an integer optimum for the primal LP, and
$H^*\subset G$ is the corresponding $\mathbf b$-matching), the complementary slackness conditions  read
\begin{itemize}
\item[(i)] For all $ij\in E(H^*);~~~w_{ij}=y_{ij}^*+y_i^*+y_j^*$.
\item[(ii)] For all $ij\notin E(H^*);~~~y_{ij}^*=0$.
\item[(iii)] For all $i$ with $d_i(H^*)<b_i;~~~y_i^*=0$.
\end{itemize}
We can construct a stable outcome $(H^*, \gamma)$ by setting $\gamma_{i \to j}= y_j^* + y_{ij}^*/2$ for $(ij) \in H^*$, and $\gamma_{i \to j} =0$ otherwise: Using (iii) above, $\gamma_i \geq y_i^*$ (cf. definition of $\gamma_i$ above), so for any $(ij) \notin H^*$, we have $\gamma_i+\gamma_j \geq y_i^*+y_j^* \geq w_{ij}$, using (ii) above. It is easy to check
that $\gamma_{i \to j}+\gamma_{j \to i}=w_{ij}$ for any $(ij)\in H^*$ using (i) above. Thus,
$(H^*, \gamma)$ is a stable outcome.

For the converse, consider a stable allocation $(M, \gamma)$. We claim that $M$ forms a (integer) primal optimum. For this we simply demonstrate that there is a feasible point in the dual with the same value as the primal value at $M$: Take $y_i = \gamma_i$, and $y_{ij} = w_{ij} - y_i - y_j$ for edges in $M$, and $0$ otherwise. The dual objective is then exactly equal to the weight of $M$. This also proves the second part of the lemma.
\end{proof}

\subsection{Dynamical model}

We retain the notation $\alf{i}{j}^t$ for the `best alternative' estimated in iteration $t$. As before, `offers' are determined as
\begin{align*}
    \off{i}{j}^t =
 (w_{ij}-\alf{i}{j}^t)_+ - \frac{1}{2}(w_{ij}-\alf{i}{j}^t-\alf{j}{i}^t)_+\, ,
\end{align*}
in the spirit of the pairwise Nash bargaining solution.

Now the best alternative $\alf{i}{j}^t$ should be the estimated income from the `replacement' partnership, if $i$ and $j$ do not reach an agreement with each other.  This `replacement' should be the one corresponding to the $b_i^{\rm th}$ largest offer received by $i$ from neighbors other than $j$. Hence, the update rule is modified to
\begin{align}
\alf{i}{j}^{t+1}= (1-\damp)\,\alf{i}{j}^{t} \,+
\,\damp \,\bmax{i}_{k \in \di \backslash j}\;
\off{k}{i}^{t}\, ,
\label{eq:b_update}
\end{align}
where $\damp \in (0,1)$ is the damping factor.

Further, we define $\Gamma = (\gamma_{i \to j}, \gamma_{j \to i})_{(ij) \in E}$ by
\begin{align}
\gamma_{j \to i}^{t} \equiv \left \{
\begin{array}{ll}
 \off{j}{i}^{t} & \mbox{if } \off{j}{i}^t \mbox{ is among top $b_i$ incoming offers to $i$}\\
 0 & \mbox{otherwise.}
\end{array} \right .
\label{eq:b_Gamma}
\end{align}
Here ties are broken arbitrarily in ordering incoming offers. Finally, we define
\begin{align}
\gamma_i^t \equiv \bmax{i}_{k \in \di} \;\gamma_{k \to i}^t = \bmax{i}_{k \in \di}\;\off{k}{i}^t
\label{eq:b_gamma}
\end{align}

\subsection{Results}

Our first result is that fixed points of the new update equations
(\ref{eq:off_def}), (\ref{eq:b_update})
are again in correspondence with Nash
bargaining solutions when such solutions exist (analogous to Theorem \ref{thm:fp_dualopt}). \textit{Note
that the fixed points are independent of the damping factor
$\damp$.} First, we generalize the notion of an
induced matching.
%

%
\begin{Definition}
\label{def:b_induces_matching}
We say that a state $(\baralf,
\baroff, \Gamma)$ (or just $\baralf$) \emph{induces a $\mathbf{b}$-matching} $M$
if the following happens.  For each node $i\in V$ receiving at least
$b_i$ non-zero offers ($\off{\cdot}{i}>0$): there is no tie for the $\bmax{i}$ incoming offer to $i$, and node $i$ is matched under $M$ to the $b_i$ neighbors from whom it is receiving its $b_i$ highest offers. For each node $i\in V$ receiving less than
$b_i$ non-zero offers: node $i$ is matched under $M$ to all its neighbors from whom it is receiving positive offers.
\end{Definition}

Consider the LP relaxation to the maximum weight matching problem
\eqref{prob:mwm_relaxation}.
A feasible point $\underline{x}$ for LP \eqref{prob:mwm_relaxation} is called \emph{half-integral} if for all $e\in E$, $x_e\in\{0,1,\frac{1}{2}\}$. Again, it can be easily shown that the primal LP (\ref{prob:b_LP_primal_dual}) always has an optimum $\underline{x}^*$ that is half-integral \cite[Chapter 31]{Schrijver}.
As before, an LP with a fully integer $\underline{x}^*$ (i.e., $x_e^*\in\{0,1\}$ for all $e \in E$) is called \emph{tight}.

\begin{theorem}\label{thm:b_FP_NB}
Let $G=(V,E)$ with edge weights $(w_{ij})_{(ij)\in E}$ and capacity constraints $\mathbf b=(b_i)$ be an instance such that the primal LP (\ref{prob:b_LP_primal_dual}) has a unique optimum that is integral, corresponding to matching $M^*$.
Let $(\baralf,\baroff,\Gamma)$ be a fixed point of
the natural dynamics. Then $\baralf$
induces matching $M^*$ and $(M^*, \Gamma)$ is a Nash bargaining solution.
Conversely, every Nash bargaining solution $(M,\Gamma_{\NB})$ has $M=M^*$ and
corresponds to a unique
fixed point of the natural dynamics with $\Gamma=\Gamma_{\NB}$.\\
\end{theorem}

We prove Theorem \ref{thm:b_FP_NB} in Appendix \ref{app:b_FP_NB_proof}.

\begin{Remark}
The condition that a tight primal LP \eqref{prob:b_LP_primal_dual} has
a unique optimum is generic (analogous to Appendix \ref{sec:fixed_point_prop_proofs}, Remark \ref{rem:unique_opt_generic}).
Hence, Theorem \ref{thm:b_FP_NB} applies to `almost all' problems with for which there exists a stable solution (cf. Lemma \ref{lemma:b_stable_characterization}).
\end{Remark}

\begin{corollary}
Let $G=(V,E)$ with edge weights $(w_{ij})_{(ij)\in E}$ and capacity constraints $\mathbf b=(b_i)$ be an instance such that the primal LP (\ref{prob:b_LP_primal_dual}) has a unique optimum that is integral. Then the instance possesses a Nash bargaining solution.
\label{coro:b_balance_existence}
\end{corollary}
Thus, we obtain a (almost) tight characterization of the when NB solutions exist in the case of general capacity constraints\footnote{For simplicity, we have stated and proved, in Theorem \ref{thm:b_FP_NB}, a generalization of only part (a) of Theorem \ref{thm:fp_dualopt}. However, we expect that part (b) also generalizes, which would then lead to an exact characterization of when NB solutions exist in this case.}.

Our convergence results, Theorems \ref{thm:convergence_to_FP} and \ref{thm:rate_of_convergence}, generalize immediately, with the proofs (cf. Section \ref{sec:convergence}) going through nearly  verbatim:

\begin{theorem}\label{thm:b_convergence_to_FP}
Let $G=(V,E)$ with edge weights $(w_{ij})_{(ij)\in E}$ and capacity constraints $\mathbf b=(b_i)$ be any instance. The natural dynamics has at least one fixed point. Moreover, for any initial condition with $\baralf^0\in[0,W]^{2|E|}$, $\baralf^t$ converges to a fixed point.
\end{theorem}

We retain the Definition \ref{def:approx_FP}
for an $\eps$-fixed point.

\begin{theorem}
\label{thm:b_rate_of_convergence}
Let $G=(V, E)$ with weights  $(w_e, e \in E) \in [0, 1]^{|E|}$ and capacity constraints $\mathbf b=(b_i)$ be any instance. Take any initial condition
$\baralf^0\in[0,1]^{2|E|}$.  Take any $\eps > 0$. Define
\begin{align}
T^*(\eps) = \frac{1}{\pi \damp (1- \damp) \eps^2}\, .
\label{eq:b_convergence_time_Teps}
\end{align}
Then for all $t \geq T^*(\eps)$,  $\baralf^t$ is an $\eps$-fixed point. (Again $\pi = 3.14159\ldots$)
\end{theorem}

\section{Discussion}
\label{sec:discussion}

Our results provide a dynamical justification for balanced outcomes, showing that agents bargaining with each other in a realistic, local manner can find such outcomes quickly. Refer to Section \ref{subsec:contribution} for a summary of our contributions.

Some caution is needed in the interpretation of our results. Our dynamics avoids the question of how and when a pair of agents will cease to make iterative updates, and commit to each other. We showed that the right pairing will be found in time polynomial in the network size $n$ and the parameter LP parameter $g$. But how will agents find out when this convergence has occurred? After all, agents are not likely to know $n$, and even less likely to know $g$. Further, why should agents wait for the right pairing to be found? It may be better for them to strike a deal after a few iterative updates because (i) they may estimate that they are unlikely to get a better deal later, (ii) they may be impatient, (iii) the convergence time may be very large on large networks. If a pair of agents \emph{do} pair up and leave, then this changes the situation for the remaining agents, some of whom may have lost possible partners (\cite{ManeaA_markovEq} studies a model with this flavor). Our dynamics does not deal with this. A possible approach to circumventing some of these problems is to interpret our model in the context of a repeated game, where agents can pair up, but still continue to renegotiate their partnerships. Formalizing this is an open problem.

Related to the above discussion is the fact that our agents are not strategic. Though our dynamics admits interpretation as a bargaining process, it is unclear how, for instance, agent $j$ becomes aware of the best alternative $\alf{i}{j}$ of a neighbor $i$. In the case of a fixed best alternative, the work of Rubinstein \cite{sequential_bargaining} justifies the pairwise Nash bargaining solution, but in our case the best alternative estimates evolve in time. Thus, it is unclear how to explain our dynamics game theoretically. However, we do not consider this to be a major drawback of our approach. Non-strategic agent behavior is commonly assumed in the literature on learning in games \cite{FLbook}, even in games of only two players. Alternative recent approaches to bargaining in networks assume strategic agents, but struggle to incorporate reasonable informational assumptions (e.g. \cite{ManeaA_markovEq} assumes common knowledge of the network and perfect information of all prior events). Prima facie, it appears that bounded rationality models like ours may be more realistic.

Several examples admit multiple balanced outcomes (see Example 3 in Section \ref{sec:intro}). In fact, this is a common feature of two-sided assignment markets, which typically contain multiple even cycles. It would be very interesting to investigate whether our dynamics favors some balanced outcomes over others. If this is the case, it may improve our ability to predict outcomes in such markets.

Our model assumes the network to be exogenous, which does not capture the fact that agents may strategically form links. It would be interesting (and very challenging) to endogenize the network. A perhaps less daunting proposition is to characterize bargaining on networks that experience shocks, like the arrival of new agents, the departure of agents or the addition/deletion of links. Our result showing convergence to an approximate fixed point in time independent of the network size provides hope of progress on this front.

Finally, we remark on the computational problem of computing exact UD solutions in the unsymmetric case (recall that we give an FPTAS). We conjecture that the problem is computationally hard (cf. Section \ref{subsec:contribution}), with the recently introduced complexity class \emph{continuous local search} \cite{CLS} providing a possible way forward. We leave it as a challenging open problem to prove or refute this conjecture.

\section{Convergence to fixed points: Proofs of Theorems
\ref{thm:convergence_to_FP} and \ref{thm:rate_of_convergence}}
\label{sec:convergence}

Theorems \ref{thm:convergence_to_FP} and
\ref{thm:rate_of_convergence} admit a surprisingly simple
proofs, that build on powerful results in the theory of
nonexpansive mappings in Banach spaces.

\begin{Definition} Given a normed linear space $L$, and a bounded domain
$D\subseteq L$, a \emph{nonexpansive mapping} $\sT:D\to L$ is a
mapping satisfying $\|\sT x-\sT y\|\le \|x-y\|\,$ for all $x,y\in D$.
\end{Definition}
Mann \cite{Mann} first considered the
iteration $x^{t+1} = (1-\kappa)\, x^t+\kappa\, \sT x^t$ for
$\kappa\in (0,1)$, which is equivalent to iterating  $\sT_{\kappa} = (1-\kappa)\,I+\kappa\, \sT$.
Ishikawa \cite{Ishikawa} and Edelstein-O'Brien \cite{Edelstein}
proved the surprising result that, if the sequence
$\{x^t\}_{t\ge 0}$ is bounded, then $\|\sT x^t-x^t\|\to 0$ (the
sequence is \emph{asymptotically regular}) and indeed $x^t\to
x^*$ with $x^*$ a fixed point of $\sT$.

Baillon and Bruck \cite{rate} recently proved a powerful
quantitative version of Ishikawa's theorem: If $\|x^0-x^t\|\le
1$ for all $t$, then
\begin{eqnarray}
\|\sT x^t-x^t\| < \frac{1}{\sqrt{\pi \kappa(1-\kappa)t}}\, .
\end{eqnarray}
The surprise is that such a result holds irrespective of the
mapping $\sT$ and of the normed space (in particular, of its
dimensions). Theorems \ref{thm:convergence_to_FP} and
\ref{thm:rate_of_convergence} immediately follow from this
theory once we recognize that the natural dynamics can be cast
into the form of a Mann iteration for a  mapping which is
nonexpansive with respect to a suitably defined norm.

Let us stress that the nonexpansivity property does not appear
to be a lucky mathematical accident, but rather an intrinsic
property of bargaining models under the one-exchange
constraint. It loosely corresponds to the basic observation
that if earnings in the neighborhood of a pair of trade
partners change by amounts $N_1, N_2, ... , N_k$, then the
balanced split for the partners changes at most by $\max(N_1,
N_2, \ldots, N_k)$, i.e., the largest of the neighboring
changes.

Our technique seems therefore applicable in a broader context.
(For instance, it can be applied successfully to prove fast
convergence of a synchronous and damped version of the
edge-balancing dynamics of \cite{Azar}.)
\begin{proof}[Proof (Theorem \ref{thm:convergence_to_FP})]
We consider the linear space $L=\reals^{2|E|}$ indexed by
directed edges in $G$. On the bounded domain $D = [0,W]^{2|E|}$
we define the mapping $\sT:\baralf\mapsto \sT\baralf$ by letting,
for $(i,j)\in E$,
\begin{equation}
\label{eq:Tdef}
(\sT\baralf)_{i\backslash j} \equiv \max_{k \in \di \backslash j} \off{k}{i}|_{\baralf}\, ,
\end{equation}
where $\off{k}{i}|_{\baralf}$ is defined by Eq.~\eqref{eq:off_def}.
It is easy to check that the sequence of best alternatives
produced by the natural dynamics corresponds to the Mann
iteration $\baralf^t = \sT_{\kappa}^t\baralf^0$. Also, $\sT$ is
nonexpansive for the $\ell_{\infty}$ norm
\begin{eqnarray}
\|\baralf-\barbet\|_{\infty}=
\max_{(i,j)\in E}|\alf{i}{j}-\beta_{i\setminus j}|\, .
\end{eqnarray}
Non-expansivity follows from:\\
(i) The `$\max$' in
Eq.~(\ref{eq:Tdef}) is non expansive.\\
(ii) \hskip-4pt An offer $\off{i}{j}$ as defined
by Eq.~\eqref{eq:off_def} is nonexpansive. To see this, note that
$\off{i}{j}=f(\alf{i}{j}, \alf{j}{i})$,
where $f(x,y):\reals_+^2 \rightarrow \reals_+$ is given by
\begin{eqnarray}
\ f(x,y) = \left\{
    \begin{array}{ll} \frac{w_{ij}-x+y}{2} & x+y \leq w_{ij}\\
    (w_{ij}-x)_+ &\mbox{ otherwise.}\\
\end{array}\right.
\label{eq:regionwise_offers}
\end{eqnarray}
It is easy to check that $f$ is continuous everywhere in $\reals_+^2$.
Also, it is differentiable except in $\{(x,y) \in \reals_+^2:x+y=w_{ij} \mbox{ or } x=w_{ij}\}$, and
satisfies $||\nabla f ||_1=|\frac{\partial f}{\partial x}|+|\frac{\partial f}{\partial y}| \leq 1$.
Hence, $f$ is Lipschitz continuous in the $L^\infty$ norm, with Lipschitz constant 1,
i.e., it is nonexpansive in sup norm.

Notice that $\sT_{\kappa}$ maps $D\equiv [0,W]^{2|E|}$ into
itself. The thesis follows from \cite[Corollary
1]{Ishikawa}.
\end{proof}

\begin{proof}[Proof (Theorem \ref{thm:rate_of_convergence})]
With the definitions given above, consider $W=1$ (whence
$\|\sT\baralf^t-\baralf^0\|_{\infty}\le 1$ for all $t$) and apply
\cite[Theorem 1]{rate}.
\end{proof}
%
%
\subsection{Exponentially fast convergence to unique Nash
bargaining solution} \label{sec:exp_fast_convergence}

Convergence of the natural dynamics was studied in an earlier
version of this paper using a different (and much more
laborious) technique \cite{OurPrevious}. While the results in
Section \ref{sec:main_results} constitute a large improvement
in elegance and generality over those of  \cite{OurPrevious},
the latter retain an independent interest.  Indeed the analysis
of \cite{OurPrevious} shows that convergence is exponentially
fast in a well defined class of instances. We decided therefore
to retain the main result of that analysis (recast from \cite{OurPrevious}).
\begin{theorem}\label{thm:ConvergenceOld}
Assume $W=1$. Let $G$ be an instance having unique Nash bargaining solution $(M, \bargamma)$
with \emph{KT gap} $\sigma>0$, and let $\bargamma$ denote the
corresponding allocation. Then, for any $\eps \in (0, \sigma/4)$, there exists
$T_*(n,\sigma,\eps)= C\, n^7
\big[1/\sigma\, +\log (\sigma/\eps)\big]\, ,$%
such that, for any initial condition with
$\baralf^0\in[0,1]^{2|E|}$, and any $t\ge T_*$ the natural
dynamics yields earning estimates $\bargamma^t$, with
$|\gamma^t_i-\gamma_i|\le \eps$ for all $i\in V$. Moreover,
$\baralf^t$ induces the matching $M$ and $(M, \bargamma^t)$ is a
$(4 \eps)$-NB solution for any $t\ge T_*$.
\end{theorem}
We refer to Appendix \ref{sec:KT+gap} for a definition
of the KT gap $\sigma$ (here KT stands for Kleinberg-Tardos).
Suffice it to say that it is related to the Kleinberg-Tardos
decomposition of $G$ and that it is polynomially computable
\cite{KT}.

As mentioned above, the proof is based on a very different
technique, namely on `approximate decoupling' of the natural
dynamics on different KT structures under the assumptions
$\sigma>0$ (which is generic) and that there is a unique NB
solution. See preprint \cite{OurPrevious} for a
complete proof.

Let us stress here that, for fixed $\sigma$,
 $T_*(n,\sigma,\eps)$ is only logarithmic in
$(1/\eps)$ while it is proportional to $1/\eps^2$ in Theorem
\ref{thm:rate_of_convergence}. In other words, for instances
with KT gap bounded away from $0$, the natural dynamics
converges exponentially fast, while Theorem
\ref{thm:rate_of_convergence} guarantees inverse polynomial
convergence in the general case.

\section{Fixed point properties: Proof of Theorem
\ref{thm:fp_dualopt}} \label{sec:FixedPoint}

%

 Let $\mc{S}$ be the set of optimum solutions of LP (\ref{prob:mwm_relaxation}).
We call  $e\in E$ a \emph{strong-solid edge} if $x_e^*=1$ for all $x^* \in \mc{S}$ and a
\emph{non-solid edge} if $x_e^*=0$ for all $x^* \in \mc{S}$.  We call $e\in E$ a \emph{weak-solid edge} if it is neither strong-solid nor non-solid.
%
\vspace{2mm}

\noindent\textbf{Proof of Theorem \ref{thm:fp_dualopt}: From fixed points to NB solutions.}
The direct part follows from the following set of fixed point
properties.
The proofs of these properties are given in Appendix
\ref{sec:fixed_point_prop_proofs}. Throughout
$(\baralf,\baroff,\bargamma)$ is a fixed point of the dynamics
(\ref{eq:off_def}), (\ref{eq:update}) (with $\bargamma$ given
by (\ref{eq:bargamma})).

(1) Two players $(i,j) \in E$ are called \emph{partners}
    if $\gamma_i + \gamma_j = w_{ij}$. Then the following
    are equivalent: (a) $i$ and $j$ are partners, (b)
    $w_{ij}-\alf{i}{j} - \alf{j}{i}\geq 0$, (c)
    $\gamma_i=\off{j}{i}$ and $\gamma_j=\off{i}{j}$.

(2) Let $P(i)$ be the set of all partners of $i$. Then
    the following are equivalent: (a) $P(i) = \{ j \}$ and $\gamma_i>0$, (b)
    $P(j) = \{ i \}$ and $\gamma_j>0$, (c) $w_{ij}-\alf{i}{j} - \alf{j}{i}
    >0$, (d) $i$ and $j$ receive unique best positive offers from each other.

(3) We say that $(i,j)$ is a \emph{weak-dotted edge} if
    $w_{ij}-\alf{i}{j} - \alf{j}{i}=0$, \emph{a
    strong-dotted edge} if $w_{ij}-\alf{i}{j} - \alf{j}{i}
    > 0$, and a \emph{non-dotted edge} otherwise. If $i$
    has no adjacent dotted edges, then $\gamma_i=0$.

(4) An edge is strong-solid (weak-solid) if and only if it is
    strongly (weakly) dotted.

(5) The balance property \eqref{eq:balance}, holds at
    every edge $(i,j) \in E$ (with both sides being
    non-negative).

(6) $\bargamma$ is an optimum solution for the dual LP
    (\ref{prob:mwm_dual}) to LP \eqref{prob:mwm_relaxation} and
    $\off{i}{j}=(w_{ij}-\gamma_i)_+$ holds for all $(i,j)
    \in E$.

\begin{proof}[Proof of Theorem \ref{thm:fp_dualopt} (a), direct implication]
Assume that the LP (\ref{prob:mwm_relaxation}) has a unique optimum that is integral. Then,
by property 4, the set of strong-dotted edges form the unique
maximum weight matching $M^*$ and all other edges are non-dotted.
By property $3$ for $i$ that is unmatched under $M^*$, $\gamma_i=0$.
Hence by property $2$, $\baralf$ induces the matching $M^*$.
Finally, by properties 6 and 5, the pair
$(M^*,\bargamma)$ is stable and balanced respectively, and thus forms a NB
solution.
\end{proof}
The corresponding result for the non-unique optimum case (part (b)) can be proved similarly: it follows immediately Theorem \ref{thm:fp_dualopt_unmatched}, Appendix \ref{sec:fixed_point_prop_proofs}.

\begin{Remark}
\label{rem:FP_properties_hold_always}
Properties 1-6 hold for any instance. This leads to the general result Theorem \ref{thm:fp_dualopt_unmatched}
in Appendix \ref{sec:fixed_point_prop_proofs} shows that in general, fixed points correspond to dual optima satisfying the
unmatched balance property (\ref{eq:balance}).
\end{Remark}
%
%
\vspace{2mm}

\noindent\textbf{Proof of Theorem \ref{thm:fp_dualopt}: From NB solutions to fixed points.}
\begin{proof}
Consider any NB solution $(M,\bargamma_{\NB} )$. Using Proposition \ref{prop:sotomayor},
 $M$ is a maximum weight matching.
Construct a corresponding FP as follows. Set
$\off{i}{j}=(w_{ij}-\gamma_{\NB,i})_+$ for all $(i,j) \in E$.
Compute $\baralf$ using $\alf{i}{j}=\max_{k \in \di \backslash
j} \off{k}{i}$. We claim that this is a FP and that the
corresponding $\bargamma$ is $\bargamma_{\NB}$. To prove that
we are at a fixed point, we imagine updated offers
$\baroff^{\textup{upd}}$ based on $\baralf$, and show
$\baroff^{\textup{upd}}=\baroff$.

Consider a matching edge $(i,j) \in M$. We know that
$\gamma_{\NB,i} + \gamma_{\NB,j}=w_{ij}$. Also stability and
balance tell us $\gamma_{\NB,i} - \max_{k \in \di \backslash j}
(w_{ik}- \gamma_{\NB,k})_+ = \gamma_{\NB,j} - \max_{l \in \dj
\backslash i} (w_{jl}- \gamma_{\NB,l})_+$ and both sides are
non-negative. Hence, $\gamma_{\NB,i} - \alf{i}{j} =
\gamma_{\NB,j} - \alf{j}{i} \geq 0$. Therefore
$\alf{i}{j}+\alf{j}{i} \leq w_{ij}$,
\begin{align*}
\off{i}{j}^\textup{upd}&= \frac{w_{ij}-\alf{i}{j}+\alf{j}{i}}{2} = \frac{w_{ij}-\gamma_{\NB,i}+\gamma_{\NB,j}}{2} \\
&= \gamma_{\NB,j} = w_{ij}-\gamma_{\NB,i} = \off{i}{j}\, .
\end{align*}
By symmetry, we also have $\off{j}{i}^\textup{upd}=
\gamma_{\NB,i} = \off{j}{i}$. Hence, the offers remain
unchanged.
Now
consider $(i,j) \notin M$. We have
$\gamma_{\NB,i}+\gamma_{\NB,j} \geq w_{ij}$ and,
$\gamma_{\NB,i} = \max_{k \in \di \backslash j} (w_{ik}-
\gamma_{\NB,k})_+ = \alf{i}{j}$. Similar equation holds for
$\gamma_{\NB,j}$. The validity of this identity can be checked
individually in the cases when $i \in M$ and $i \notin M$.
Hence, $\alf{i}{j}+\alf{j}{i} \geq w_{ij}$. This leads to
$\off{i}{j}^{\textup{upd}}=(w_{ij}-\alf{i}{j})_+=(w_{ij}-\gamma_{\NB,i})_+=\off{i}{j}$.
By symmetry, we know also that $\off{j}{i}^{\textup{upd}}=\off{j}{i}$.

Finally, we show $\bargamma = \bargamma_{\NB}$. For all $(i,j)
\in M$, we already found that $\off{i}{j}=\gamma_j$ and vice
versa. For any edge $(ij) \notin M$, we know $\off{i}{j} =
(w_{ij}-\gamma_{\NB,i})_+ \leq \gamma_{\NB,j}$. This
immediately leads to $\bargamma = \bargamma_{\NB}$. It is worth
noting that making use of the uniqueness of LP optimum we know that $M=M^*$, and  we can
further show that
 $\gamma_i = \off{j}{i} > \alf{i}{j}$ if and only if $(ij) \in M$, i.e.,
 the fixed point reconstructs the pairing $M=M^*$.
\end{proof}

\section{Polynomial convergence on bipartite graphs: Proof of Theorem \ref{thm:bipartite_pairing}}
\label{sec:bipartite_pairing}
Theorem
\ref{thm:bipartite_pairing} says that on a bipartite graph,
under a small random perturbation on any problem instance, the
natural dynamics is likely to quickly find the maximum weight
matching. Now, in light of Corollary \ref{coro:fast_matching},
this simply involves showing that the gap $g$ of the perturbed
problem instance is likely to be sufficiently large. We use a
version of the well known Isolation Lemma to for this. Note
that on bipartite graphs, there is always an integral optimum
to the LP (\ref{prob:mwm_relaxation}).

Next, is our Isolation lemma (recast from \cite{Isolation}). For the proof, see Appendix \ref{app:isolation}.
\begin{lemma}[Isolation Lemma]
\label{lemma:isolation}
Consider a 
bipartite graph $G=(V, E)$. Choose $\eta > 0,\; \xi >0$. Edge
weights are generated as follows: for each $e \in E$,
$\bar{w}_e$ is chosen uniformly in $[w_e, w_e + \eta]$. Denote
by $\mathcal{M}$ the set of matchings in $G$. Let $M^*$ be a
maximum weight matching. Let $M^{**}$ be a matching having the
maximum weight in $\mathcal{M} \backslash M^*$. Denote by
$\bar{w}(M)$ the weight of a matching $M$. Then
\begin{align}
\Pr  [\,\,\,\bar{w}(M^{*})-\bar{w}(M^{**}) \geq\, \eta \xi/(2|E|)\,\,\, ] \; \geq \; 1- \xi
\end{align}
\end{lemma}

\begin{proof}[Proof of Theorem \ref{thm:bipartite_pairing}]
Using Lemma \ref{lemma:isolation}, we know that the gap of the
perturbed problem satisfies $\bar{g} \geq \eta \xi /(2|E|)$ with
probability at least $1-\xi$. Now, the weights in the perturbed
instance are bounded by $\bar{W}=2$. Rescale by dividing all
weights and messages by $2$, and use Corollary
\ref{coro:fast_matching}. The theorem follows from the
following two elementary observations. First, an $(\eps/2)$-NB
solution for the rescaled problem corresponds to an $\eps$-NB
solution for the original problem. Second, induced matchings
are unaffected by scaling.
\end{proof}

We remark that Theorem \ref{thm:bipartite_pairing} does not generalize to any
(non-bipartite) graph with edge weights such that the LP
(\ref{prob:mwm_relaxation}) has an integral optimum, for the
following reason. We can easily generalize the Isolation Lemma to
show that the gap $g$ of the perturbed problem is likely to be
large also in this case. However, there is a probability
arbitrarily close to 1 (depending on the instance) that a
random perturbation will result in an instance for which LP
(\ref{prob:mwm_relaxation}) does \textit{not} have an integral
optimum, i.e. the perturbed instance does not have any Nash
bargaining solutions!



\vspace{0.20cm}

\noindent {\bf Acknowledgements.} We thank Eva Tardos for
introducing us to network exchange theory and Daron Acemoglu
for insightful discussions. We also thank the anonymous referees for their comments.

A large part of this work
was done while Y. Kanoria, M. Bayati and A. Montanari were at Microsoft Research New
England. This research was partially supported by NSF, grants
CCF-0743978 and CCF-0915145, and by a Terman fellowship. Y. Kanoria is
supported by a 3Com Corporation Stanford Graduate Fellowship.

\appendix

\section{Variations of the natural dynamics}
\label{app:variations_of_dynamics}

Now that we have a
reasonable dynamics that converges fast to balanced outcomes,
it is natural ask whether variations of the natural dynamics also
yield balanced outcomes. What happens in the case of asynchronous
updates, different nodes updating at different rates, damping
factors that vary across nodes and in time, and so on? We discuss some of these questions in this section,
focussing on some situations in which we \textit{can} prove
convergence with minimal additional work. Note that we are only concerned with
extending our convergence results since the fixed point properties remain unchanged.

\subsection{Node dependent damping}
 Consider that the damping factor may be different for
different nodes, but unchanging over time. Denote by
$\damp(v)$, the damping factor for node $v$. Assume that
$\damp(v) \in [1-\damp^*,\damp^*] \ \forall v \in V$ for some
$\damp^* \in [0.5, 1)$, i.e. damping factors are uniformly
bounded away from $0$ and $1$. Define operator $\sT : [0,1]^{|E|}
\rightarrow [0,1]^{|E|}$ by
\begin{align}
(\sT \baralf)_{i \backslash j} =
\left (\frac{\damp(i)}{\damp^*} \right)
\max_{k \in \partial i \backslash j} \off{k}{i}
+ \left( \frac{\damp^*- \damp(i)}{ \damp^*} \right)\alf{i}{j}
\end{align}
$\sT$ is nonexpansive.
Now, the dynamics can be written as $\baralf^{t+1} = \damp^* \sT
\baralf^t + (1-\damp^* ) \baralf^t$. Clearly, convergence to
fixed points (Theorem \ref{thm:convergence_to_FP}) holds in
this situation. Note that fixed points of $\sT$ are the same as
fixed points of the natural dynamics. Moreover, we can use
\cite{rate} to assert that $||\baralf^t - \sT\baralf^t||_\infty =
O(1/\sqrt{t})$ and hence $\baralf^t$ is an $O(1/\sqrt{t})$-FP.
In short, we don't lose anything with this generalization!

\subsection{Time varying damping}
Now consider instead that the damping may change over time, but
is the same for all nodes. Denote by $\damp_t$ the damping factor
at time $t$, i.e.
\begin{align}
\baralf^{t+1} = \damp_t \max_{k \in \partial i \backslash j} \off{k}{i}^t
+ (1-\damp_t) \baralf^t
\end{align}
The result of \cite{Ishikawa} implies that, as long as
$\sum_{t=0}^\infty \damp_t = \infty$ and $\lim_{t\rightarrow
\infty } \sup \damp_t < 1$, the dynamics is guaranteed to
converge to a fixed point. Note that again the fixed points are
unchanged. \cite{time_varying_damping} provides a quantitative
estimate of the rate of convergence in this case, guaranteeing
in particular that an $\eps$-FP is reached in time $\exp(
O(1/\eps))$ if $\damp_t$ is uniformly bounded away from $0$ and
$1$. Note that this estimate is much weaker than the one
provided by \cite{rate}, leading to Theorem
\ref{thm:rate_of_convergence}. It seems intuitive that the
stronger $O(1/\sqrt{k})$ bound holds also for the time varying
damping case in the general nonexpansive operator setting, but
a proof has remained elusive thus far.

\subsection{Asynchronous updates}
Finally, we look at the case of asynchronous updates, i.e., one
message $\alf{i}{j}$ is updated in any given step while the
others remain unchanged. Define $\sT_{i \backslash j}:
[0,1]^{|E|} \rightarrow [0,1]^{|E|}$ by
\begin{align}
(\sT_{i \backslash j} \baralf)_{i' \backslash j'} &= \left \{
\begin{array}{ll}
\max_{k \in \partial i' \backslash j'} \off{k}{i'}
& \mbox{if } (i,j)=(i',j')\\
\alf{i'}{j'} & \mbox{otherwise}
\end{array}
\right .
\end{align}
Let $m \equiv |E|$. There are $2m$ such operators, two for each
edge. Clearly, each $\sT_{i \backslash j}$ is nonexpansive in
sup-norm. Now, consider an arbitrary permutation of the $2m$
directed edges $((i_1,j_1), (i_2,j_2), \ldots )$. Consider the
updates induced by $\sT_{i_1\backslash j_1}, \sT_{i_2\backslash
j_2}, \ldots \ $ in order, each with a damping factor of
$1/(2m)$. Consider the resulting product
\begin{align}
\label{eq:asynchronous_T_defined}
&\Big((1/2m)\sT_{i_1\backslash j_1} + \left(1-(1/2m)\right)I\Big ) \,\cdot \\
\cdot \, \Big((1/2m)&\sT_{i_2\backslash j_2} + \left(1-(1/2m)\right)I\Big )\,\cdot \;
\ldots \nonumber\\
&= \Big ( 1- \left(1\!-\!(1/2m)\right)^{2m} \Big ) \sT + \left(1\!-\!(1/2m)\right)^{2m} I \nonumber
\end{align}
Here \eqref{eq:asynchronous_T_defined} defines $\sT$, and $I$ is the identity operator. It is easy to deduce that
$\sT$ above is nonexpansive from the following elementary facts
-- the product of nonexpansive operators in nonexpansive, and
the convex combination of nonexpansive operators is
nonexpansive. Also, $\left(1-(1/2m)\right)^{2m} \in [1/4, 1/e]
\ \forall m$. Thus, if we repeat these asynchronous updates
periodically in a series of `update cycles', we are guaranteed
to quickly converge to an $\eps$-FP of $\sT$ (\cite{rate}).
\begin{proposition}
An $\eps$-FP of $\sT$ is an $O(m\eps)$-FP of the natural
dynamics. \label{propo:async_approx_FP}
\end{proposition}
\begin{proof}
Suppose we start an update cycle at $\baralf$, an $\eps$-FP of
$\sT$. Then we know that at the end of the update cycle, no
coordinate changes by more than $(1-1/4) \eps \leq \eps$. Note
that among the $2m$ steps in a cycle, any particular $i
\backslash j$ `coordinate' only changes in one step. Thus, each
such coordinate change is bounded by $\eps$. Consider the
$s$-th step in the update cycle. The state before the $s$-th
step, call it $\baralf(s-1)$, is $\eps$-close to $\baralf$.
Also, we know that the $(i_s\backslash j_s)$ coordinate changes
by at most $\eps$ in this step. Hence,
\begin{align*}
||\sT_{i_s\backslash j_s} \baralf(s-1) - \baralf(s-1)||_\infty &\leq (2m)\eps \\
\Rightarrow \phantom{||\sT_{i_s\backslash j_s} \baralf(s-1)} ||\sT_{i_s\backslash j_s} \baralf- \baralf||_\infty &\leq (2m + 2)\eps
\end{align*}
This holds for $s=1,2, \ldots, 2m$. Hence the result.
\end{proof}

Note that with $\eps=0$, Proposition
\ref{propo:async_approx_FP} tells us that fixed points of $\sT$
are fixed points of the natural dynamics. Thus, we are
immediately guaranteed convergence to fixed points of the
natural dynamics. Moreover, the quantitative estimate  in
Proposition \ref{propo:async_approx_FP} guarantees that in a
small number of update cycles we reach approximate fixed points
of the natural dynamics.

Finally, we comment that instead of ordering updates by a
permutation of directed edges, we could have an arbitrary
periodic sequence of updates satisfying non-starvation and
obtain similar results. For example, this would include cases
where some nodes update more frequently than others. Also, note
that the damping factors of $(1/2m)$ were chosen for simplicity
and to ensure \textit{fast} convergence. Any non-trivial
damping would suffice to guarantee convergence.

It remains an open question to show convergence for
non-periodic asynchronous updates.

\section{Proof of Isolation lemma}
\label{app:isolation}

Our proof of the isolation lemma is adapted from
\cite{Isolation}.

\begin{proof}[Proof of Lemma \ref{lemma:isolation}]
Fix $e \in E$ and fix $\bar{w}_{e'}$ for all $e' \in E \backslash e$.
Let $M_e$ be a maximum weight matching among matchings that strictly include edge $e$, and let $M_{\sim e}$
be a maximum weight matching among matchings that exclude edge $e$. Clearly, $M_e$ and $M_{\sim e}$
are independent of $\bar{w}_e$. Define
\begin{align*}
f_e(\bar{w}_e) &\equiv \bar{w} (M_e) \ \; = f_e(0)+\bar{w}_e \\
f_{\sim e} &\equiv \bar{w} (M_{\sim e}) = \mbox{const} < \infty
\end{align*}
Clearly, $f_e(0) \leq f_{\sim e}$, since we cannot do worse by forcing exclusion of a zero weight edge. Thus, there is some unique $\theta \geq 0$ such that $f_e(\theta) = f_{\sim e}$.
Define $\delta = \eta \xi /2|E|$.
Let $D(e)$ be the event that $|\bar{w}(M_e) - \bar{w}(M_{\sim e})| < \delta$. It is
easy to see that $D(e)$ occurs if and only if $\bar{w}_e \in (\theta-\delta, \theta+ \delta)$. Thus,
$\Pr [ D(e) ] \leq 2 \delta/\eta = \xi / |E|$. Now,
\begin{align}
\bigg \{\bar{w}(M^{*})-\bar{w}(M^{**}) < \delta \bigg\} \; = \; \bigcup_{e \in E} D(e)
\end{align}
and the lemma follows by union bound.

\end{proof}

\section{Proofs of fixed point properties}\label{sec:fixed_point_prop_proofs}

In this section we state and prove the fixed point properties that were used
for the proof of Theorem \ref{thm:fp_dualopt} in Section \ref{sec:FixedPoint}.
Before that, however, we remark that
the condition: ``LP \eqref{prob:mwm_relaxation} has a unique optimum''
in Theorem \ref{thm:fp_dualopt}(a) is almost always valid.

\begin{Remark}
\label{rem:unique_opt_generic}
We argue that the condition ``LP \eqref{prob:mwm_relaxation} has a unique optimum" is generic in instances with integral optimum:\\
Let ${\sf G}_{\textup{I}} \subset [0,W]^{|E|}$ be the set of instances having an integral optimum.
Let ${\sf G}_{\textup{UI}} \subset {\sf G}_{\textup{I}}$ be the set of instances having a unique integral optimum.
It turns out that ${\sf G}_{\textup{I}}$ has dimension $|E|$ (i.e.
the class of instances having an integral optimum is large) and
that ${\sf G}_{\textup{UI}}$ is both
\emph{open and dense in ${\sf G}_{\textup{I}}$}.
\end{Remark}

\noindent\textbf{Notation.} In proofs of this section and Section \ref{sec:eps_fixed_point_prop_proofs} we denote surplus $w_{ij}-\alf{i}{j}-\alf{j}{i}$ of edge $(ij)$ by $\surp_{ij}$.

\begin{lemma}
$\bargamma$ satisfies the constraints of the dual problem (\ref{prob:mwm_dual}).
\label{lemma:gamma_sat_dual_const}
\end{lemma}
\begin{proof}
Since offers $\off{i}{j}$ are by definition non-negative therefore for all $v\in V$ we have $\gamma_v\geq 0$. So we only need to show $\gamma_i+\gamma_j\geq w_{ij}$ for any edge $(ij)\in E$.  It is easy to see that $\gamma_i\geq \alf{i}{j}$ and $\gamma_j\geq \alf{i}{j}$. Therefore, if $\alf{i}{j}+\alf{i}{j}\geq w_{ij}$ then $\gamma_i+\gamma_j\geq w_{ij}$ holds and we are done.  Otherwise, for $\alf{i}{j}+\alf{i}{j}<w_{ij}$ we have
$\off{i}{j} = \frac{w_{ij}-\alf{i}{j}+\alf{j}{i}}{2}$ and $\off{j}{i} = \frac{w_{ij}-\alf{j}{i}+\alf{i}{j}}{2}$ which gives
$\gamma_i+\gamma_j\geq \off{i}{j}+\off{j}{i}=w_{ij}$.\enp
\end{proof}

Recall that for any $(ij) \in E$, we say that $i$ and $j$ are `partners' if $\gamma_i + \gamma_j = w_{ij}$ and
$P(i)$ denotes the partners of node $i$. In other words $P(i) = \{ j: j \in \partial i, \gamma_i + \gamma_j = w_{ij}\}$.

\begin{lemma}\label{lemma:i-jpartners_equivalence}
The following are equivalent:\\
(a) $i$ and $j$ are partners,\\
(b) $\surp_{ij}\geq 0$.\\
(c) $\gamma_i=\off{j}{i}$ and $\gamma_j=\off{i}{j}$.\\
Moreover, if $\gamma_i=\off{j}{i}$ and $\gamma_j>\off{i}{j}$ then $\gamma_i=0$.
\end{lemma}
\begin{proof} We will prove $(a) \Rightarrow (b) \Rightarrow (c) \Rightarrow (a)$.

$(a)\Rightarrow (b)$: Since $\gamma_i\geq \alf{i}{j}$ and $\gamma_j\geq \alf{j}{i}$ always holds then $w_{ij}=\gamma_i + \gamma_j\geq \alf{i}{j} + \alf{j}{i}$.

$(b)\Rightarrow (c)$: If $\surp_{ij}\geq 0$ then $(w_{ij}-\alf{i}{j}+\alf{j}{i})/2\geq \alf{j}{i}$. But $\off{i}{j}=(w_{ij}-\alf{i}{j}+\alf{j}{i})/2$ therefore $\gamma_j=\off{i}{j}$. The argument for $\gamma_i=\off{j}{i}$ is similar.

$(c)\Rightarrow (a)$: If $\surp_{ij} \geq 0$ then $\off{i}{j}=(w_{ij}-\alf{i}{j}+\alf{j}{i})/2$ and $\off{j}{i}=(w_{ij}-\alf{j}{i}+\alf{i}{j})/2$ which gives $\gamma_i+\gamma_j=\off{i}{j}+\off{j}{i}=w_{ij}$ and we are done.  Otherwise, we have
$\gamma_i+\gamma_j=\off{i}{j}+\off{j}{i}\leq (w_{ij}-\alf{i}{j})_+ + (w_{ij}-\alf{j}{i})_+ < \max \bigg[(w_{ij}-\alf{i}{j})_+, (w_{ij}-\alf{j}{i})_+, 2w_{ij}-\alf{i}{j}-\alf{j}{i}\bigg]\leq w_{ij}$ which contradicts Lemma \ref{lemma:gamma_sat_dual_const}
that $\bargamma$ satisfies the constraints of the dual problem (\ref{prob:mwm_dual}).

Finally, we need to show that $\gamma_i=\off{j}{i}$ and $\gamma_j>\off{i}{j}$ give $\gamma_i=0$. First note that by equivalence of $(b)$ and $(c)$ we should have $w_{ij} < \alf{i}{j} + \alf{j}{i}$. On the other hand $\alf{i}{j}\leq \gamma_i  =\off{j}{i}\leq (w_{ij}-\alf{j}{i})_+$.  Now if  $w_{ij}-\alf{j}{i}>0$ we get $\alf{i}{j}\leq w_{ij}-\alf{j}{i}$ which is a contradiction. Therefore $\gamma_i=(w_{ij}-\alf{j}{i})_+=0$.
\enp
\end{proof}

\begin{lemma}
The following are equivalent:\\
(a) $P(i) = \{ j \}$ and $\gamma_i>0$,\\
(b) $P(j) = \{ i \}$ and $\gamma_j>0$,\\
(c) $w_{ij}-\alf{i}{j} - \alf{j}{i} >0$.\\
(d) $i$ and $j$ receive unique best positive offers from each other.
\label{lemma:strong_dotted}
\end{lemma}
\begin{proof}
$(a)\Rightarrow (c) \Rightarrow (b)$: $(a)$ means that for all $k\in\p i\backslash j$, $\surp_{ik}<0$. This means $\off{k}{i}=(w_{ik}-\alf{k}{i})_+<\alf{i}{k}=\off{j}{i}$ (using $\gamma_i>0$). Hence,
$\alf{i}{j}<\off{j}{i}$. From $(a)$, it also follows that $\off{j}{i}>0$ or
$(w_{ij}-\alf{j}{i})_+=w_{ij}-\alf{j}{i}$. Therefore, $\off{j}{i} \leq (w_{ij}-\alf{j}{i})_+ = w_{ij}-\alf{j}{i}$ which gives $w_{ij}-\alf{i}{j} - \alf{j}{i} >0$ or $(c)$.  From this we can explicitly write $\off{i}{j}=(w_{ij}-\alf{i}{j}+\alf{j}{i})/2$ which is strictly bigger than $\alf{j}{i}$. Hence we obtain $(b)$.

By symmetry $(b)\Rightarrow (c) \Rightarrow (a)$.  Thus, we have shown that $(a)$, $(b)$ and $(c)$ are equivalent.

$(c) \Rightarrow (d)$: $(c)$ implies that $\off{i}{j}= (w_{ij}-\alf{i}{j} + \alf{j}{i})/2 > \alf{j}{i} = \max_{k \in \partial j \backslash i} \off{k}{j}$. Thus, $j$ receives its unique best positive offer from $i$. Using symmetry, it follows that $(d)$ holds.

$(d) \Rightarrow (c)$: $(d)$ implies $\gamma_i = \off{j}{i}$ and $\gamma_j = \off{i}{j}$. By Lemma \ref{lemma:i-jpartners_equivalence}, $i$ and $j$ are partners, i.e. $\gamma_i + \gamma_j = w_{ij}$. Hence, $\off{i}{j} + \off{j}{i} = w_{ij}$. But since $(d)$ holds, $\alf{i}{j} < \off{j}{i}$ and $\alf{j}{i} < \off{i}{j}$. This leads to $(c)$.

This finishes the proof.
\enp
\end{proof}

Recall that $(ij)$ is a weak-dotted edge if $w_{ij}-\alf{i}{j} - \alf{j}{i}=0$, a strong-dotted edge if
$w_{ij}-\alf{i}{j} - \alf{j}{i} > 0$, and a non-dotted edge otherwise.
Basically, for any dotted edge $(ij)$ we have $j\in P(i)$ and $i\in P(j)$.

\begin{corollary} A corollary of Lemmas \ref{lemma:i-jpartners_equivalence}-\ref{lemma:strong_dotted} is that strong-dotted edges are only adjacent to non-dotted edges.  Also each weak-dotted edge is adjacent to at least one weak-dotted edge at each end (assume that the earnings of the two endpoints are non-zero).
\end{corollary}

\begin{lemma}
If $i$ has no adjacent dotted edges, then $\gamma_i=0$
\label{lemma:no_dotted_means_gamma0}
\end{lemma}
\begin{proof}
Assume that the largest offer to $i$ comes from $j$. Therefore, $\alf{i}{j}\leq \off{j}{i}\leq (w_{ij}-\alf{j}{i})_+$.
Now if $w_{ij}-\alf{j}{i}>0$ then $\alf{i}{j}\leq w_{ij}-\alf{j}{i}$ or $(ij)$ is dotted edge which is impossible. Thus, $w_{ij}-\alf{j}{i}=0$ and $\gamma_i=0$.
\enp
\end{proof}

\begin{lemma}
The following are equivalent:\\
(a) $\alf{i}{j}=\gamma_i$,\\
(b) $\surp_{ij}\le 0$,\\
(c) $\off{i}{j}=(w_{ij}-\alf{i}{j})_+$.
\label{lemma:when_alf_equals_gamma}
\end{lemma}
\begin{proof}
$(a)\Rightarrow (b)$: ``not (b)" $\Rightarrow \off{j}{i} = (w_{ij} - \alf{j}{i} + \alf{i}{j})/2 > \alf{i}{j} \Rightarrow $ ``not (a)".

$(b)\Rightarrow (c)$: Follows from the definition of $\off{i}{j}$.

$(c)\Rightarrow (a)$: From $\off{i}{j}=(w_{ij}-\alf{i}{j})_+$ we have $\surp_{ij} \le 0$. Therefore, $\off{j}{i}=(w_{ij}-\alf{j}{i})_+\leq \max\big[w_{ij}-\alf{j}{i},0\big]\leq \alf{i}{j}$.
\enp
\end{proof}
Note that (b) is symmetric in $i$ and $j$, so (a) and (c) can be transformed by interchanging $i$ and $j$.
\begin{corollary}
$\alf{i}{j}=\gamma_i$ if and only if $\alf{j}{i}=\gamma_j$
\label{corr:alf_gamma_eq}
\end{corollary}

\begin{lemma}
$\off{i}{j}=(w_{ij}-\gamma_i)_+$ holds $\forall \; (ij) \in E$
\label{lemma:off_from_gamma}
\end{lemma}
\begin{proof}
If $w_{ij}-\alf{i}{j} - \alf{j}{i} \le 0$ then the result follows from Lemma \ref{lemma:when_alf_equals_gamma}.  Otherwise, $(ij)$ is strongly dotted and $\gamma_i=\off{j}{i}=(w_{ij}-\alf{j}{i} + \alf{i}{j})/2$,  $\gamma_j=\off{i}{j}=(w_{ij}-\alf{i}{j} + \alf{j}{i})/2$. From here we can explicitly calculate $w_{ij}-\gamma_i=(w_{ij}-\alf{i}{j} + \alf{j}{i})/2=\off{i}{j}$.
\enp
\end{proof}

\begin{lemma}
The unmatched balance property, equation \eqref{eq:balance}, holds at every edge $(ij) \in E$, and both sides of the equation
are non-negative.
\label{lemma:balance_at_FP}
\end{lemma}
\begin{proof}
In light of lemma \ref{lemma:off_from_gamma}, \eqref{eq:balance} can be rewritten at a fixed point as
\begin{align}
\gamma_i - \alf{i}{j} = \gamma_j - \alf{j}{i}
\label{eq:FP_balance}
\end{align}
which is easy to verify. The case $\surp_{ij} \le 0$ leads to both
sides of Eq.~(\ref{eq:FP_balance}) being $0$ by Corollary \ref{corr:alf_gamma_eq}. The
other case $\surp_{ij} >0$ leads to
\begin{align}
\off{i}{j}-\alf{j}{i}=\off{j}{i}-\alf{i}{j}=\frac{\surp_{ij}}{2}
\end{align}
Clearly, we have $\gamma_i=\off{j}{i}$ and $\gamma_j=\off{i}{j}$. So Eq.~(\ref{eq:FP_balance}) holds.
\enp
\end{proof}
Next lemmas show that dotted edges are in correspondence with the solid edges that were defined in Section \ref{sec:FixedPoint}.

\begin{lemma}
A non-solid edge cannot be a dotted edge, weak or strong.
\label{lemma:no_extra_dotted}
\end{lemma}

Before proving the lemma let us define alternating paths. A path $P=( {i_1}, {i_2},\ldots, {i_k})$ in $G$ is called \emph{alternating path} if:
(a) There exist a partition {of} edges of $P$ into two sets $A,B$ such that either $A\subset M^*$ or $B\subset M^*$.
Moreover $A$ ($B$) consists of all \emph{odd (even)} edges; i.e. $A=\{( {i_1}, {i_2}), ( {i_3}, {i_4}),\ldots\}$
($B=\{( {i_2}, {i_3}), ( {i_4}, {i_5}),\ldots\}$). (b) The path $P$ might intersect itself or even
repeat its own edges but no edge is repeated immediately.
That is, for any $1\leq r\leq k-2:~~~~ i_r\neq  i_{r+1}$ and $ i_r\neq  i_{r+2}$.  $P$ is called an \emph{alternating cycle} if $ {i_1}= {i_k}$.

Also, consider $\underline{x}^*$ and $\underline{y}^*$ that are optimum solutions for the LP and its dual, \eqref{prob:mwm_relaxation} and \eqref{prob:mwm_dual}. The complementary slackness conditions (see \cite{Schrijver}) for more details)
state that for all $v\in V$, $y_v^*(\sum_{e\in \p v}x_e^*-1)=0$ and for all $e=(ij)\in E$, $x_e^*(y_i^*+y_j^*-w_{ij})=0$.
Therefore, for all solid edges the equality $y_i^*+y_j^*=w_{ij}$ holds. Moreover, any node $v\in V$ is adjacent to a solid edge if and only if $y_v^*>0$.

\begin{proof}[Proof of Lemma \ref{lemma:no_extra_dotted}]
First, we refine the notion of solid edges by calling an edge $e$, 1-$\underline{x}^*$-solid ($\frac{1}{2}$-$\underline{x}^*$-solid) whenever $x_e^*=1$ ($x_e^*=\frac{1}{2}$).

We need to consider two cases:

\textbf{Case (I).} Assume that LP has an optimum solution $\underline{x}^*$ that is integral as well (having a tight LP).

The idea of the proof is that if there exists a non-solid edge $e$ which is dotted, we use a similar analysis to \cite{BayatiB} to construct an alternating path consisting of dotted and $\underline{x}^*$-solid edges that leads to creation of at an optimal solution to LP \eqref{prob:mwm_relaxation} that assigns a positive value to $e$.  This contradicts the non-solid assumption on $e$.

Now assume the contrary: take $(i_1,i_2)$ that is a non-solid edge but it is dotted. Consider an endpoint of $(i_1,i_2)$. For example take $i_2$. Either there is a $\underline{x}^*$-solid edge attached to $i_2$ or not. If there is not, we stop. Otherwise, assume $(i_2,i_3)$ is a $\underline{x}^*$-solid edge. Using Lemma \ref{lemma:no_dotted_means_gamma0}, either $\gamma_{i_3}=0$ or there is a dotted edge connected to $i_3$. But if this dotted edge is $(i_2,i_3)$ then $P(i_2)\supseteq\{i_1,i_3\}$. Therefore, by Lemma \ref{lemma:strong_dotted} there has to be another dotted edge $(i_3,i_4)$ connected to $i_3$. Now, depending on whether $i_4$ has (has not) an adjacent $\underline{x}^*$-solid edge we continue (stop) the construction. A similar procedure could be done by starting at $i_1$ instead of $i_2$. Therefore, we obtain an alternating path $P=({i_{-k}},\ldots, i_{-1}, {i_0}, {i_1}, {i_2},\ldots, {i_\ell})$ with all odd edges being dotted and all even edges being $\underline{x}^*$-solid. Using the same argument as in \cite{BayatiB} one can show that one of the following four scenarios occur.

\noindent \textbf{Path:} Before $P$ intersects itself, both end-points of the path stop. Either the last edge is $\underline{x}^*$-solid (then $\gamma_v =0$ for the last node) or the last edge is a dotted edge.  Now consider a new solution $\underline{x}'$ to LP \eqref{prob:mwm_relaxation} by $x_e'=x_e^*$ if $e\notin P$ and $x_e'=1-x_e^*$ if $e\in P$.  It is easy to see that $\underline{x}'$ is a feasible LP solution at all points $v\notin P$ and also for internal vertices of $P$. The only nontrivial case is when $v=i_{-k}$ (or $v=i_{\ell}$) and the edge $(i_{-k},i_{-k+1})$ (or $(i_{\ell-1},i_{\ell})$ ) is dotted. In both of these cases, by construction $v$ is not connected to an $\underline{x}^*$-solid edge outside of $P$.  Hence, making any change inside of $P$ is safe. Now denote the weight of all solid (dotted) edges of $P$ by $w(P_{\textrm{solid}})$ ($w(P_{\textrm{dotted}})$). Here, we only include edges outside
$P_{\textrm{solid}}$ in $P_{\textrm{dotted}}$.
Clearly,
\begin{align}
\sum_{e \in E} w_e x_e^*- \sum_{e \in E} w_e x_e'=w(P_{\textrm{solid}})-w(P_{\textrm{dotted}}).
\label{eq:w(solid)-w(dotted)}
\end{align}
But $w(P_{\textrm{dotted}})=\sum_{v\in P}\gamma_v$ .  Moreover, from Lemma \ref{lemma:gamma_sat_dual_const},
$\underline{\gamma}$ is dual feasible which gives $w(P_{\textrm{solid}})\leq\sum_{v\in P}\gamma_v$.
We are using the fact that if there is a $\underline{x}^*$-solid edge at an endpoint of $P$ the $\gamma$
of the endpoint should be $0$. Now Eq.~\eqref{eq:w(solid)-w(dotted)} reduces to
$\sum_{e\in E}w_e x_e^*- \sum_{e \in E} w_e x_e'\leq 0.$  This contradicts that $e=(i_1,i_2)$ is non-solid since $x'_e>0$.

\noindent \textbf{Cycle:} $P$ intersects itself and will contain an even cycle $C_{2s}$. This case can be handled very similar to the path by defining $x_e'=x_e^*$ if $e\notin C_{2s}$ and $x_e'=1-x_e^*$ if $e\in C_{2s}$. The proof is even simpler since the extra check for the boundary condition is not necessary.

\noindent \textbf{Blossom:} $P$ intersects itself and will contain an odd cycle $C_{2s+1}$ with a path (stem) $P'$ attached to the cycle at point $u$. In this case let $x_e'=x_e^*$ if $e\notin P'\cup C_{2s+1}$, and $x_e'=1-x_e^*$ if $e\in P'$, and $x_e'=\frac{1}{2}$ if $e\in C_{2s+1}$. From here, we drop the subindex $2s+1$ to simplify the notation. Since the cycle has odd length, both neighbors of $u$ in $C$ have to be dotted.
    Therefore,
\begin{align*}
&\sum_{e \in E} w_e x_e^*- \sum_{e \in E} w_e x_e'\\
= &\,w(P'_{\textrm{solid}})+w(C_{\textrm{solid}})-w(P'_{\textrm{dotted}})\\
&-\frac{w(C_{\textrm{dotted}})+w(C_{\textrm{solid}})}{2}\\
= &\,w(P'_{\textrm{solid}})+\frac{w(C_{\textrm{solid}})}{2}-w(P'_{\textrm{dotted}})
-\frac{w(C_{\textrm{dotted}})}{2}\, .
\end{align*}
Plugging $w(P'_{\textrm{solid}}) \leq \sum_{v\in P'}\gamma_v$, $w(C_{\textrm{solid}}) \leq \sum_{v\in C}\gamma_v-\gamma_u$,
$w(P'_{\textrm{dotted}})= \sum_{v\in P'}\gamma_v -\gamma_u$ and $w(C_{\textrm{dotted}})= \sum_{v\in C}\gamma_v+\gamma_u$, we obtain
\begin{align*}
\sum_{e \in E} w_e x_e^*- \sum_{e \in E} w_e x_e' \leq \, 0\, ,
\end{align*}
    which is again a contradiction.

\noindent \textbf{Bicycle:} $P$ intersects itself at least twice and will contain two odd cycles $C_{2s+1}$ and $C'_{2s'+1}$ with a path (stem) $P'$ that is connecting them. Very similar to Blossom, let $x_e'=x_e^*$ if $e\notin P'\cup C\cup C'$, $x_e'=1-x_e^*$ if $e\in P'$, and $x_e'=\frac{1}{2}$ if $e\in C\cup C'$. The proof follows similar to the case of blossom.

\textbf{Case (II).} Assume that there is an optimum solution $\underline{x}^*$ of LP that is not necessarily integral.

Everything is similar to Case (I) but the algebraic treatments are slightly different. Some edges $e$ in $P$ can be $\f{1}{2}$-$\underline{x}^*$-solid ($x_e^*=\frac{1}{2}$). In particular some of the odd edges (dotted edges) of $P$ can now be $\f{1}{2}$-$\underline{x}^*$-solid. But the subset of $\f{1}{2}$-$\underline{x}^*$-solid edges of $P$ can be only sub-paths of odd length in $P$. On each such sub-path defining $\underline{x}'=1-\underline{x}^*$ means we are not affecting $\underline{x}^*$. Therefore, all of the algebraic calculations should be considered on those sub-paths of $P$ that have no $\f{1}{2}$-$\underline{x}^*$-solid edge which means both of their boundary edges are dotted.

\noindent \textbf{Path:} Define $\underline{x}'$ as in Case (I). Using the discussion above, let $P^{(1)},\ldots,P_{(r)}$ be disjoint sub-paths of $P$ that have no $\f{1}{2}$-$\underline{x}^*$-solid edge. Thus,
    $\sum_{e \in E} w_e x_e^*- \sum_{e \in E} w_e x_e'=\sum_{i=1}^r \big[w(P_{\textrm{solid}}^{(i)})-w(P_{\textrm{dotted}}^{(i)})\big]$.
    Since in each $P^{(i)}$ the two boundary edges are dotted,  $w(P_{\textrm{solid}}^{(i)})\leq \sum_{v\in P^{(i)}}\gamma_v$ and $\sum_{v\in P^{(i)}}\gamma_v=w(P_{\textrm{dotted}}^{(i)})$. The rest can be done as in Case (I).

\noindent \textbf{Cycle, Blossom, Bicycle:} These cases can be done using the same method of breaking the path and cycles into sub-paths $P^{(i)}$ and following the case of path.

\end{proof}

\begin{lemma}
Every strong-solid edge is a strong-dotted edge. Also, every
weak-solid edge is a weak-dotted edge.
\label{lemma:no_extra_solid}
\end{lemma}

\begin{proof}
We rule out all alternative cases one by one. In particular we prove:

(i) \emph{A strong-solid edge cannot be weak-dotted.}  If an edge $(i,j)$ is strong-solid then it cannot be adjacent to another solid edge (weak or strong).  Therefore, using Lemma \ref{lemma:no_extra_dotted} none of adjacent edges to $(i,j)$ are dotted. However, if $(i,j)$ is weak-dotted by Lemma \ref{lemma:strong_dotted} it is adjacent to at least one other weak-dotted edge (since at least one of $\gamma_i$ and $\gamma_j$ is positive) which is a contradiction. Thus $(i,j)$ cannot be weak-dotted.

(ii) \emph{A strong-solid edge cannot be non-dotted.}  Similar to (i), if an edge $(i,j)$ is strong-solid it cannot be adjacent to dotted edges. Now, if $(i,j)$ is non-dotted then $\gamma_i=\gamma_j=0$ using Lemma \ref{lemma:no_dotted_means_gamma0}. Hence $w_{ij}<\gamma_i+\gamma_j=0$ which is contradiction since we assumed all weights are positive.

(iii) \emph{A weak-solid edge cannot be strong-dotted.}  Assume, $(i_1,i_2)$ is weak-solid and strong-dotted.  Then we can show an optimum to LP \eqref{prob:mwm_relaxation} can be improved which is a contradiction. The proof is very similar to proof of Lemma \ref{lemma:no_extra_dotted}.
    Since $(i_1,i_2)$ is weak-solid, there is a half-integral matching $\underline{x}^*$ that is optimum to LP and puts a mass $1/2$ or $0$ on $(i_1,i_2)$. Then either there is an adjacent $\underline{x}^*$-solid edge $(i_2,i_3)$ or an adjacent $\underline{x}^*$-solid edge $(i_0,i_1)$ with mass at least $1/2$ or we stop. In the latter case, increasing the value of $x_{i_1i_2}^*$ increases $\sum_{e \in E} w_e x_e^*$ while keeping it LP feasible which is a contradiction. Otherwise, by strong-dotted assumption on $(i_1,i_2)$ ($(i_0,i_1)$), the new edge $(i_2,i_3)$ is not dotted. Now we select a dotted edge $(i_3,i_4)$ if it exists (otherwise we stop and in that case $\gamma_{i_3}=0$).  This process is repeated as in proof of Lemma \ref{lemma:no_extra_dotted} in both directions to obtain an alternating path $P=({i_{-k}},\ldots, i_{-1}, {i_0}, {i_1}, {i_2},\ldots, {i_\ell})$ with all odd edges being dotted with $\underline{x}^*$ value at most $1/2$ and all even edges being $\underline{x}^*$-solid with mass at least $1/2$.  We discuss the case of $P$ being a simple path (not intersecting itself) here, and other cases: cycle, bicycle and blossom can be treated similar to path as in proof of Lemma \ref{lemma:no_extra_dotted}.

    Construct LP solution $\underline{x}'$ that is equal to $\underline{x}^*$ outside of $P$ and inside it satisfies $x_e'=x_e^*+1/2$ if $e$ is an odd edge that is $e=(i_{2k-1,i_{2k}})$, and $x_e'=x_e^*-1/2$ when
    $e$ is an even edge that is $e=(i_{2k,i_{2k+1}})$.  It is easy to see that $\underline{x}'$ is a feasible LP solution.  And since for all edges $(i_j,i_{j+1})$ we have $\gamma_{i_j}+\gamma_{i_{j+1}}\ge w_{i_ji_{j+1}}$ and on dotted edges we have equality $\gamma_{i_j}+\gamma_{i_{j+1}}= w_{i_ji_{j+1}}$ then
    $
    \sum_{e \in E} w_e x_e^*- \sum_{e \in E} w_e x_e'=\frac{w(P_{\textrm{dotted}})-w(P_{\textrm{solid}})}{2}
    \ge \frac{\gamma_{i_2}+\gamma_{i_3}-w_{i_2i_3}}{2}
    >0
    $
    where the last inequality follows from the fact that $(i_2,i_3)$ is not-dotted.  Hence we reach a contradiction.

(iv) \emph{A weak-solid edge cannot be non-dotted.}  Assume, $(i_1,i_2)$ is weak-solid and non-dotted.  Similar to (iii) we can show the best solution to LP \eqref{prob:mwm_relaxation} can be improved which is a contradiction. Since $(i_1,i_2)$ is weak-solid we can choose a half-integral $\underline{x}^*$ that puts a mass at least $1/2$ on $(i_1,i_2)$.  Also, this time the alternation in $P$ is the opposite of (iii). That is we choose $(i_2,i_3)$ to be dotted (if it does not exist $\gamma_{i_2}=0$ and we stop.)  The solution $\underline{x}'$ is constructed as before:  equal to $\underline{x}^*$ outside of $P$, $x_e'=x_e^*+1/2$ if $e$ is odd and $x_e'=x_e^*-1/2$ if it is even. Hence,
$\sum_{e \in E} w_e x_e^*- \sum_{e \in E} w_e x_e'\ge \frac{\gamma_{i_1}+\gamma_{i_2}-w_{i_1i_2}}{2}>0$,
    using the non-dotted assumption on $(i_1,i_2)$. Hence, we obtain another contradiction.
\end{proof}

\begin{lemma}
$\bargamma$ is an optimum for the dual problem (\ref{prob:mwm_dual})
\label{lemma:gamma_opt}
\end{lemma}
\begin{proof}
Lemma \ref{lemma:gamma_sat_dual_const} guarantees feasibility.
Optimality follows from lemmas \ref{lemma:no_dotted_means_gamma0}, \ref{lemma:no_extra_dotted}
and \ref{lemma:no_extra_solid} as follows. Take any optimum half integral matching $\underline{x}^*$ to LP.
Now using Lemma \ref{lemma:no_extra_solid}:  $\sum_{v}\gamma_v=\sum_{e\in E}w_ex_e^*$ which finishes the proof.
\end{proof}

\begin{theorem} \label{thm:fp_dualopt_unmatched}
Let $\BALOPT$ be the set of optima of the dual problem
(\ref{prob:mwm_dual}) satisfying the unmatched balance
property, Eq. \eqref{eq:balance}, at every edge. If
$(\baralf, \baroff, \bargamma)$ is a fixed point of the natural
dynamics then $\bargamma \in \BALOPT$. Conversely, for every
$\bargamma_{\rm{BO}} \in \BALOPT$, there is a unique fixed
point of the natural dynamics with $\bargamma =
\bargamma_{\rm{BO}}$.
\end{theorem}
\begin{proof}
The direct implication is immediate from Lemmas
\ref{lemma:balance_at_FP} and \ref{lemma:gamma_opt}. The
converse proof here follows the same steps as for Theorem 1,
proved in Section \ref{sec:FixedPoint}. Instead of separately analyzing the cases
$(ij) \in M$ and $(ij) \notin M$, we study the cases $\gamma_i + \gamma_j = w_{ij}$
and $\gamma_i + \gamma_j > w_{ij}$.
\end{proof}

%
%

\section{$\epsilon$-fixed point properties: Proof of Theorem \ref{thm:approx_fp}}
\label{sec:eps_fixed_point_prop_proofs}

In this section we
prove Theorem \ref{thm:approx_fp}, stated in Section
\ref{sec:main_results}.  In this section we assume that
$\baralf$ is an $\epsilon$-fixed point with corresponding
offers $\baroff$ and earnings $\bargamma$. That is, for all $i,j$
\begin{align*}
\eps&\ge|\alf{i}{j} - \max_{k \in \partial i \backslash j} \off{k}{i} \big|\,,\\
\off{i}{j}& =  (w_{ij}-\alf{i}{j})_+ - \frac{(w_{ij}-\alf{i}{j}-\alf{j}{i})_+}{2}\,,\\
\gamma_i &= \max_{k \in \di} \, \off{k}{i}\,.
\end{align*}

\begin{Definition}
An edge $(ij)$ is called $\delta$-dotted ($\delta\ge0$) if $\gamma_i+\gamma_j\le w_{ij}+\delta$.
\end{Definition}

\begin{lemma}\label{lemma:delta_dotted_properties}
For all edge $(ij)\in E$ and all $\delta,\delta_1,\delta_2\in\reals$ the following hold:

(a) If $(ij)$ is $\delta$-dotted then $\surp_{ij}\ge -(2\eps+\delta)$.

(b) If $\surp_{ij}\ge -\delta$ then $\off{i}{j}\ge\gamma_j-(\eps+\delta)$ and $\off{j}{i}\ge\gamma_i-(\eps+\delta)$.

(c) If $\off{i}{j}\ge\gamma_j-\delta_1$ and $\off{j}{i}\ge\gamma_i-\delta_2$ then $(ij)$ is $(\delta_1+\delta_2)$-dotted.

(d) If $\gamma_i-\delta\le\off{j}{i}$ and $\gamma_j>\off{i}{j}+2\eps+\delta$ then $\gamma_i=0$.

(e) If $\gamma_i>0$ and $\off{j}{i}\ge\gamma_i-\delta$ then $(ij)$ is $(2\delta+2\eps)$-dotted.

(f) For $\gamma_i,\gamma_j>0$, $\off{j}{i}\le\alf{i}{j}+\delta$ if and only if $\off{i}{j}\le \alf{j}{i}+\delta$.

(h) For all $(ij)$, $|\off{i}{j}-(w_{ij}-\gamma_i)_+|\le\eps$.

(i) For all $(ij)$, $\gamma_i-(w_{ij}-\gamma_j)_+\ge-\eps$ and $\gamma_i+\gamma_j\ge w_{ij}-\eps$.

(j) For all $i$, if $\gamma_i>0$ then there is at least a $2\eps$-dotted edge attached to $i$.
\end{lemma}
\begin{proof}
(a) Since $\baralf$ is $\eps$-fixed point, $\gamma_i\ge\off{i}{j}-\eps$ and $\gamma_j\ge\off{j}{i}-\eps$. Therefore,
$
\surp_{ij}=w_{ij}-\off{i}{j}-\off{j}{i} \ge w_{ij}-\gamma_i-\gamma_j-(2\eps)\geq-(2\eps+\delta)$.

(b) First consider the case $\surp_{ij}\le 0$. Then,
$
\off{i}{j}=(w_{ij}-\alf{i}{j})_+\ge w_{ij}-\alf{i}{j}\ge \alf{j}{i}-\delta\ge \max_{\ell\in\p j\backslash i}(\off{\ell}{j})-\delta-\eps
$,
which yields $\off{i}{j}\ge \gamma_j -(\eps+\delta)$. The proof of $\off{j}{i}\ge \gamma_i -(\eps+\delta)$ is similar.

For the case $\surp_{ij}> 0$,
$
\off{i}{j}=\f{w_{ij}-\alf{i}{j}+\alf{j}{i}}{2}
=\f{\surp_{ij}}{2}+\alf{i}{j}\ge \max(\f{-\delta}{2},0)+\max_{\ell\in\p j\backslash i}(\off{\ell}{j})-\eps
$,
and the rest follows as above.

(c) Note that $\gamma_i+\gamma_j\le \off{i}{j}+\off{j}{i}+\delta_1+\delta_2$. If $\surp_{ij}\ge 0$ then the result follows from $\off{i}{j}+\off{j}{i}=w_{ij}$.
For $\surp_{ij}<0$ the result follows from
$
\off{i}{j}+\off{j}{i}\le \max[(w_{ij}-\alf{i}{j})_+,(w_{ij}-\alf{j}{i})_+,2w_{ij}-\alf{i}{j}-\alf{j}{i}]\le w_{ij}.
$

(d) We need to show that when $\gamma_i\le \off{j}{i}+\delta$ and $\gamma_j>\off{i}{j}+2\eps+\delta$ then $\gamma_i=0$.  From part $(b)$ that was just shown, the surplus should satisfy $\surp_{ij}<-(\eps+\delta)$. On the other hand
$
\alf{i}{j}-\eps\le\max_{k\in\p i\backslash j}(\off{k}{i}) \le\gamma_i \le \off{j}{i}+\delta\le(w_{ij}-\alf{i}{j})_++\delta
$.
Now, if $\gamma_i>0$ then $w_{ij}-\alf{i}{j}>0$ which gives $\alf{i}{j}-\eps\le w_{ij}-\alf{i}{j}+\delta.$ This is equivalent to $\surp_{ij}\ge -(\eps+\delta)$ which is a contradiction.  Hence $\gamma_i=0$.

(e) Using part $(d)$ we should have $\off{i}{j}\ge \gamma_j-(2\eps+\delta)$. Now applying part $(c)$ the result follows.

(f) If $\surp_{ij}\ge 0$ then
$
\f{w_{ij}-\alf{j}{i}+\alf{i}{j}}{2}=\off{j}{i}\le\alf{i}{j}+\delta
$.
This inequality is equivalent to
$
\off{i}{j}=\f{w_{ij}-\alf{i}{j}+\alf{j}{i}}{2}\le\alf{j}{i}+\delta
$,
which proves the result.  If $\surp_{ij}<0$ then $w_{ij}-\alf{j}{i}\le(w_{ij}-\alf{j}{i})_+\le\alf{i}{j}+\delta$. This is equivalent to
$w_{ij}-\alf{i}{j}\le\alf{j}{i}+\delta$ which yields the result.

(h) If $\surp_{ij}\ge 0$ then by part (b), $\off{i}{j}+\eps\ge \gamma_j$ and $\off{j}{i}+\eps\ge \gamma_i$. Therefore, using $\gamma_j\ge\off{i}{j}$, $\gamma_i\ge\off{j}{i}$ and $\off{j}{i}+\off{i}{j}=w_{ij}$ we have,
$
\off{i}{j}\ge w_{ij}-\off{j}{i}\ge w_{ij}-\gamma_i\ge w_{ij}-\off{j}{i}-\eps\ge \off{i}{j}-\eps$,
which gives the result.

If $\surp_{ij}<0$ then $\off{i}{j}=(w_{ij}-\alf{i}{j})_+< \alf{j}{i}$ this gives $\gamma_j-\eps<\alf{j}{i}$. On the other hand $\alf{j}{i}\le\gamma_j+\eps$ holds. Similarly,
$\gamma_i+\eps\ge\alf{i}{j}\ge\gamma_i-\eps$ that leads to $|(w_{ij}-\alf{i}{j})_+-(w_{ij}-\gamma_i)_+|\le \eps$. Hence, the result follows from $\off{i}{j}=(w_{ij}-\alf{i}{j})_+$.

(i) Using part $(h)$, $\off{j}{i}+\eps\ge (w_{ij}-\gamma_j)_+$. Now result follows using $\gamma_i\ge\off{j}{i}$.

(j) There is at least one neighbor $j\in\p i$ that sends the maximum offer $\off{j}{i}=\gamma_i$. Using part $(d)$ we should have $\off{i}{j}\ge \gamma_j-2\eps$ and now the result follows from part $(c)$.
\enp
\end{proof}

\begin{lemma}\label{lemma:eps-fix-gives-eps-balance}
For any edge $(ij)\in E$ the earnings estimate $\bargamma$ satisfies $6\eps$-balanced property (i.e., Eq. \eqref{eq:eps_balance} holds for $6\eps$ instead of $\eps$).
\end{lemma}
\begin{proof}
Using Lemma \ref{lemma:delta_dotted_properties}(h),
$
\alf{i}{j}-2\eps\le\max_{k\in\p i\backslash j}(\off{k}{i})-\eps\le\max_{k\in\p i\backslash j}[(w_{ik}-\gamma_k)_+]\le\max_{k\in\p i\backslash j}(\off{k}{i})+\eps\le\alf{i}{j}+2\eps
$, or
\begin{align}
\left|\max_{k\in\p i\backslash j}[(w_{ik}-\gamma_k)_+]-\alf{i}{j}\right|\le 2\eps
\label{eq:eps-fix-eq-1}
\end{align}
Now, if $\surp_{ij}\le 0$ then $\off{j}{i}=(w_{ij}-\alf{j}{i})_+\le \alf{i}{j}$ which gives $|\gamma_i-\alf{i}{j}|\le \eps$ or,
$\left|\gamma_i-\max_{k\in\p i\backslash j}[(w_{ik}-\gamma_k)_+]\right|\le 3\eps$.
Therefore, $6\eps$-balance property holds.

And if $\surp_{ij}> 0$, by Lemma \ref{lemma:delta_dotted_properties}(b) we have $\off{j}{i}+\eps\ge\gamma_i$. Hence,
$
\f{\surp_{ij}}{2}+\eps=\off{j}{i}-\alf{i}{j}+\eps\ge\gamma_i-\alf{i}{j}\ge\off{j}{i}-\alf{i}{j}=\f{\surp_{ij}}{2}
$.
Same bound holds for $\gamma_j-\alf{j}{i}$ by symmetry. Therefore, using Eq.~\eqref{eq:eps-fix-eq-1}, $|\gamma_i-\max_{k\in\p i\backslash j}[(w_{ik}-\gamma_k)_+]|$ and $|\gamma_j-\max_{\ell\in\p j\backslash i}[(w_{j\ell}-\gamma_\ell)_+]|$ are within $3\eps\le6\eps$ of each other.
\enp
\end{proof}

\begin{lemma}
\label{lemma:delta-dotted-give-another-dotted}
If $(ij)$ is $\delta$-dotted for $k\in\p i\backslash j$ and if $\gamma_k>\max(\delta,\eps)+6\eps$, then there exists $r\in\p k\backslash i$ such that $(rk)$ is $(\max(\delta,\eps)+6\eps)$-dotted.
\end{lemma}

\begin{proof}
Using, $\gamma_i+\gamma_j\leq w_{ij}+\delta$ and Lemma \ref{lemma:delta_dotted_properties}(i),
\[
-\eps\le\gamma_i-\max_{s\in\p i\backslash k}[(w_{is}-\gamma_s)_+]\le\gamma_i-(w_{ij}-\gamma_j)_+\le \delta.
\]
Therefore, $|\gamma_i-\max_{s\in\p i\backslash k}[(w_{is}-\gamma_s)_+]|\le \max(\delta,\eps)$ which combined with Lemma \ref{lemma:eps-fix-gives-eps-balance} gives
\[
|\gamma_k-\max_{r\in\p k\backslash i}[(w_{rk}-\gamma_r)_+]|\le \max(\delta,\eps)+6\eps.
\]
This fact and $\gamma_k>\max(\delta,\eps)+6\eps$, show that there exists an edge $r\in \p k\backslash i$ with $|\gamma_k-(w_{rk}-\gamma_r)_+|\le\max(\delta,\eps)+6\eps$ and the result follows.
\enp
\end{proof}

\begin{lemma}
A non-solid edge cannot be a $\delta$-dotted edge for $\delta \le4\eps$.
\label{lemma:no_extra_delta_dotted}
\end{lemma}
Note that this Lemma holds even for the more general case of $M^*$ being non-integral.

The proof is a more complex version of proof of Lemma \ref{lemma:no_extra_dotted}. Recall the notion of alternating path from that proof.

Also, consider $\underline{x}^*$ and $\underline{y}^*$ that are optimum solutions for the LP and its dual, \eqref{prob:mwm_relaxation} and \eqref{prob:mwm_dual}. Also recall that by complementary slackness conditions , for all solid edges the equality $y_i^*+y_j^*=w_{ij}$ holds. Moreover, any node $v\in V$ is adjacent to a solid edge if and only if $y_v^*>0$.

\begin{proof}[Proof of Lemma \ref{lemma:no_extra_delta_dotted}]
We need to consider two cases:

\textbf{Case (I).} Assume that the optimum LP solution $\underline{x}^*$ is integral (having a tight LP). Now assume the contrary: take $(i_1,i_2)$ that is a non-solid edge but it is $\delta$-dotted. Consider an endpoint of $(i_1,i_2)$. For example take $i_2$. Either there is a solid edge attached to $i_2$ or not. If there is not, we stop. Otherwise, assume $(i_2,i_3)$ is a solid edge. Using Lemma \ref{lemma:delta-dotted-give-another-dotted}, either $\gamma_{i_3}>10\eps$ or there is a $10\eps$-dotted edge $(i_3,i_4)$ connected to $i_3$. Now, depending on whether $i_4$ has (has not) an adjacent solid edge we continue (stop) the construction. Similar procedure could be done by starting at $i_1$ instead of $i_2$. Therefore, we obtain an alternating path $P=({i_{-k}},\ldots, i_{-1}, {i_0}, {i_1}, {i_2},\ldots, {i_\ell})$ with each $(i_{2k},i_{2k+1})$ being $(6k+4)\eps$-dotted and all $(i_{2k-1},i_{2k}))$ being solid. Using the same argument as in \cite{BayatiB} one can show that one of the following four scenarios occur.

\noindent \textbf{Path:} Before $P$ intersects itself, both end-points of the path stop. At each end of the path, either the last edge is solid (then $\gamma_v < (3n+4)\eps$ for the last node $v$) or the last edge is a $(3n+4)$-dotted edge with no solid edge attached to $v$.  Now consider a new solution $\underline{x}'$ to LP \eqref{prob:mwm_relaxation} by $x_e'=x_e^*$ if $e\notin P$ and $x_e'=1-x_e^*$ if $e\in P$.  It is easy to see that $\underline{x}'$ is a feasible LP solution at all points $v\notin P$ and also for internal vertices of $P$. The only nontrivial case is when $v=i_{-k}$ (or $v=i_{\ell}$) and the edge $(i_{-k},i_{-k+1})$ (or $(i_{\ell-1},i_{\ell})$ ) is $(3n+4)\eps$-dotted. In both of these cases, by construction no solid edge is attached to $v$ outside of $P$ so making any change inside of $P$ is safe. Now denote the weight of all solid (remaining) edges of $P$ by $w(P_{\textrm{solid}})$ ($w(P_{\textrm{dotted}})$). Hence, $\sum_{e \in E} w_e x_e^*- \sum_{e \in E} w_e x_e'=w(P_{\textrm{solid}})-w(P_{\textrm{dotted}})$.

    But $w(P_{\textrm{dotted}})+(3n^2+16n)\eps/4\ge \sum_{v\in P}\gamma_v$ .  Moreover, from Lemma \ref{lemma:delta_dotted_properties}(i),  $\gamma_i+\gamma_j\ge w_{ij}-\eps$ for all $(ij)\in P$ which gives $w(P_{\textrm{solid}})\leq\sum_{v\in P}\gamma_v+n\eps/2$. Now $\sum_{e \in E} w_e x_e^*- \sum_{e \in E} w_e x_e'=w(P_{\textrm{solid}})-w(P_{\textrm{dotted}})$ yields $w_e x_e^*- \sum_{e \in E} w_e x_e'\leq (3n^2+18n)\eps/4 \le n(n+5)\eps.$  For $\eps<g/(6n^2)$ This contradicts the tightness of LP relaxation \eqref{prob:mwm_relaxation} since $x_e'\neq x_e^*$ holds at least for $e=(i_1,i_2)$.

\noindent \textbf{Cycle:} $P$ intersects itself and will contain an even cycle $C_{2s}$. This case can be handled very similar to the path by defining $x_e'=x_e^*$ if $e\notin C_{2s}$ and $x_e'=1-x_e^*$ if $e\in C_{2s}$. The proof is even simpler since the extra check for the boundary condition is not necessary.

\noindent \textbf{Blossom:} $P$ intersects itself and will contain an odd cycle $C_{2s+1}$ with a path (stem) $P'$ attached to the cycle at point $u$. In this case let $x_e'=x_e^*$ if $e\notin P'\cup C_{2s+1}$, and $x_e'=1-x_e^*$ if $e\in P'$, and $x_e'=\frac{1}{2}$ if $e\in C_{2s+1}$. From here, we drop the subindex $2s+1$ to simplify the notation. Since the cycle has odd length, both neighbors of $u$ in $C$ have to be dotted.
    Therefore,
\begin{align*}
\sum_{e \in E} w_e x_e^*- \sum_{e \in E} w_e x_e'&=w(P'_{\textrm{solid}})+w(C_{\textrm{solid}})
-w(P'_{\textrm{dotted}}) -\frac{w(C_{\textrm{dotted}})+w(C_{\textrm{solid}})}{2}\,,
\\
&= w(P'_{\textrm{solid}})+\frac{w(C_{\textrm{solid}})}{2}-w(P'_{\textrm{dotted}})-\frac{w(C_{\textrm{dotted}})}{2}\\ &< \sum_{v\in P'}\gamma_v + \left\lceil\f{|P|}{2}\right \rceil\eps+ \frac{\sum_{v\in C}\gamma_v-\gamma_u}{2} +s\eps - \sum_{v\in P'}\gamma_v +\gamma_u\\
&\phantom{ < \ \ }  +\left (\f{3|P|^2+16|P|}{4}\right)\eps - \frac{\sum_{v\in C}\gamma_v+\gamma_u}{2}+\left(\f{3s^2+16s}{4}\right)\eps\,.
\end{align*}
But the last term is at most $n(n+5)\eps$ which is again a contradiction.

\noindent \textbf{Bicycle:} $P$ intersects itself at least twice and will contain two odd cycles $C_{2s+1}$ and $C'_{2s'+1}$ with a path (stem) $P'$ that is connecting them. Very similar to Blossom, let $x_e'=x_e^*$ if $e\notin P'\cup C\cup C'$, $x_e'=1-x_e^*$ if $e\in P'$, and $x_e'=\frac{1}{2}$ if $e\in C\cup C'$. The proof follows similar to the case of blossom.

\textbf{Case (II).} Assume that the optimum LP solution $\underline{x}^*$ is not necessarily integral.

Everything is similar to Case (I) but the algebraic treatments are slightly different. Some edges $e$ in $P$ can be $\f{1}{2}$-solid ($x_e^*=\frac{1}{2}$). In particular some of the odd edges (dotted edges) of $P$ can now be $\f{1}{2}$-solid. But the subset of $\f{1}{2}$-solid edges of $P$ can be only sub-paths of odd length in $P$. On each such sub-path defining $\underline{x}'=1-\underline{x}^*$ means we are not affecting $\underline{x}^*$. Therefore, all of the algebraic calculations should be considered on those sub-paths of $P$ that have no $\f{1}{2}$-solid edge which means both of their boundary edges are dotted.

\noindent \textbf{Path:} Define $\underline{x}'$ as in Case (I). Using the discussion above,
let $P^{(1)},\ldots,P_{(r)}$ be disjoint sub-paths of $P$ that have no $\f{1}{2}$-solid edge. Thus,
$\sum_{e \in E} w_e x_e^*- \sum_{e \in E} w_e x_e'=\sum_{i=1}^r \big[w(P_{\textrm{solid}}^{(i)})-w(P_{\textrm{dotted}}^{(i)})\big]$.
Since in each $P^{(i)}$ the two boundary edges are dotted,
$w(P_{\textrm{solid}}^{(i)})\leq \sum_{v\in P^{(i)}}\gamma_v+|P^{(i)}|\eps/2$ and $\sum_{v\in P^{(i)}}\gamma_v\le w(P_{\textrm{dotted}}^{(i)})+(3|P^{(i)}|^2+16|P^{(i)}|)\eps/4$. The rest can be done as in Case (I).

\noindent \textbf{Cycle, Blossom, Bicycle:} These cases can be done using the same method of breaking the path and cycles into sub-paths $P^{(i)}$ and following the case of path.

\end{proof}

The direct part of Theorem \ref{thm:approx_fp} follows from the next lemma.
\begin{lemma}
$\baralf$ induces the matching $M^*$.
\label{lemma:no_extra_solid-delta}
\end{lemma}
\begin{proof}
From Lemma \ref{lemma:no_extra_delta_dotted} it follows that the set of $2\eps$-dotted edges is a subset of the solid edges. In particular, when the optimum matching $M^*$ is integral, no node can be adjacent to more than one $2\eps$-dotted edges.
If we define a $\underline{x}'$ to be zero on all edges and $x_e'=1$ for all $2\eps$-dotted edges $(ij)$ with $\gamma_i+\gamma_j>0$. Then clearly $\underline{x}'$ is feasible to \eqref{prob:mwm_relaxation}.  On the other hand, using the definition of $2\eps$-dotted for all $e'$ with $x_{e'}=1$, and
Lemma \ref{lemma:delta_dotted_properties}(j) that each node with $\gamma_i>0$ is adjacent to at least one $2\eps$-dotted edge we can write
$
\sum_{e \in E} w_e x_e'\ge\sum_{v\in V}\gamma_v-n\eps
$.
Separately, from Lemma \ref{lemma:delta_dotted_properties}(i) we have
$
\sum_{v\in V}\gamma_v\ge \sum_{e \in E} w_e x_e^* - \f{n\eps}{2}
$,
which shows that $\underline{x}'$ is also an optimum solution to \eqref{prob:mwm_relaxation} (when $\eps<g/(6n^2)$).
From the uniqueness assumption on $\underline{x}^*$ we obtain that $M^*$ is equal to the set of all $2\eps$-dotted edges with at least one endpoint having a positive earning estimate. We would like to show that for any such edge $(ij)$, both earning estimates $\gamma_i$ and $\gamma_j$ are positive.

Assume the contrary, i.e., without loss of generality $\gamma_i=0$.
Then, $\surp_{ij}\le 0$ and $0=\off{j}{i}=(w_{ij}-\alf{j}{i})_+$ that gives $\alf{j}{i}\ge w_{ij}$ or
\[
\off{\ell}{j}\ge\alf{j}{i}-\eps\ge w_{ij}-\eps\ge (w_{ij}-\alf{i}{j})_+ -\eps= \gamma_j-\eps.
\]
for some $\ell\in\p j\backslash i$.  Now using Lemma \ref{lemma:delta_dotted_properties}(e) the edge $(j\ell)$ is $4\eps$-dotted which contradicts Lemma \ref{lemma:no_extra_delta_dotted}.

Finally, the endpoints of the matched edges provide each other their unique best offers. This latter follows from the fact that each node with $\gamma_i>0$ receives an offer equal to $\gamma_i$ and the edge corresponding to that offer has to be $2\eps$-dotted using Lemma \ref{lemma:delta_dotted_properties}(d).
The nodes with no positive offer $\gamma_i=0$ are unmatched in $M^*$ as well.
\end{proof}

\paragraph{Proof of Theorem \ref{thm:approx_fp}.}
\begin{proof}
For any $\eps<g/(6n^2)$,  an $\eps$-fixed point induces the matching $M^*$ using Lemma \ref{lemma:no_extra_solid-delta}.
Additionally, the earning vector $\bargamma$ is $(6\eps)$-balanced using Lemma \ref{lemma:eps-fix-gives-eps-balance}.
Next we show that $(\bargamma,M^*)$ is a stable trade outcome.
\begin{lemma}
The earnings estimates $\bargamma$ is an optimum solution to the dual \eqref{prob:mwm_dual}. In particular the pair $(\bargamma,M^*)$ is a stable trade outcome.
\label{lemma:gamma_is_dual_opt_delta}
\end{lemma}
\begin{proof}
Using Lemma \ref{lemma:no_extra_delta_dotted}, we can show that for any non-solid edge $(ij)$, stability holds, i.e. $\gamma_i+\gamma_j\ge w_{ij}$.

Now let $(i,j)$ be a solid edge. Then $i$ and $j$ are sending each other their best offers. If $\surp_{ij}\ge0$ we are done using
$
\,\gamma_i+\gamma_j=\off{j}{i}+\off{i}{j}=\,\f{w_{ij}-\alf{i}{j}+\alf{j}{i}}{2}+\f{w_{ij}-\alf{j}{i}+\alf{i}{j}}{2}=w_{ij}$.
And if $\surp_{ij}<0$ then $\gamma_i=\off{j}{i}=(w_{ij}-\alf{j}{i})_+ \le \alf{i}{j}$. Similarly, $\gamma_j\le \alf{j}{i}$. This means there exist $k\in\p i\backslash j$ with $\off{k}{i}\ge \alf{i}{j} -\eps\ge \gamma_i-\eps$. But, from Lemma \ref{lemma:delta_dotted_properties}(e) the edge $(ik)$ would become $4\eps$-dotted which is a contradiction.
\end{proof}

The converse of Theorem \ref{thm:approx_fp} is trivial since any $\eps$-NB solution $(M,\bargamma_{\NB})$ is stable and produces a trade outcome by definition, hence it is a dual optimal solution which means $M=M^*$.
\end{proof}

\section{Proof of Theorem \ref{thm:b_FP_NB}}
\label{app:b_FP_NB_proof}

\begin{theorem}
Let $G=(V,E)$ with edge weights $(w_{ij})_{(ij)\in E}$ and capacity constraints $\mathbf b=(b_i)$ be an instance such that the primal LP (\ref{prob:b_LP_primal_dual}) has a unique optimum that is integral, corresponding to matching $M^*$.
Let $(\baralf,\baroff,\Gamma)$ be a fixed point of
the natural dynamics. Then $\baralf$
induces matching $M^*$ and $(M^*, \Gamma)$ is a Nash bargaining solution.
Conversely, every Nash bargaining solution $(M,\Gamma_{\NB})$ has $M=M^*$ and
corresponds to a unique
fixed point of the natural dynamics with $\Gamma=\Gamma_{\NB}$.\\
\end{theorem}

Let $\mc{S}$ be the set of optimum solutions of LP (\ref{prob:b_LP_primal_dual}).
 As in the one-matching case, we call  $e\in E$ a \emph{strong-solid edge} if $x_e^*=1$ for all $x^* \in \mc{S}$ and a
\emph{non-solid edge} if $x_e^*=0$ for all $x^* \in \mc{S}$.  We call $e\in E$ a \emph{weak-solid edge} if it is neither strong-solid nor non-solid.
%
\vspace{2mm}

\noindent\textbf{Proof of Theorem \ref{thm:b_FP_NB}: From fixed points to NB solutions.}
The direct part follows from the following set of fixed point
properties, similar to those for the one-matching case.
Throughout
$(\baralf,\baroff,\Gamma)$ is a fixed point of the dynamics
\eqref{eq:b_update} (with $\Gamma$ given
by (\ref{eq:b_Gamma}), and $\baroff$ given by \eqref{eq:off_def}).
The properties are proved for the case when the primal LP in \eqref{prob:b_LP_primal_dual} has a unique integral optimum (which implies that there are no weak-solid edges).

(1) Two players $(i,j) \in E$ are called \emph{partners}
    if $\gamma_i + \gamma_j \leq w_{ij}$. Then the following
    are equivalent: (a) $i$ and $j$ are partners, (b)
    $w_{ij}-\alf{i}{j} - \alf{j}{i}\geq 0$, (c)
    $\gamma_i\leq \off{j}{i}$ and $\gamma_j\leq \off{i}{j}$.

(2) The following are equivalent: (a) $w_{ij}-\alf{i}{j} - \alf{j}{i}
    >0$, (b) $\gamma_{j \to i} > \alf{i}{j} = \bmax{i}_{k \in \di \backslash j}\, \off{k}{i}$.
    Denote this set of edges by $M$.

(3) We say that $(i,j)$ is a \emph{weak-dotted edge} if
    $w_{ij}-\alf{i}{j} - \alf{j}{i}=0$, \emph{a
    strong-dotted edge} if $w_{ij}-\alf{i}{j} - \alf{j}{i}
    > 0$, and a \emph{non-dotted edge} otherwise. If $i$
    has less than $b_i$ adjacent dotted edges, then $\gamma_i=0$.

(4) Each strong solid edge is strong dotted, and each non-solid edge is non-dotted.

(5) The balance property \eqref{eq:b_balance}, holds at
    every edge $(i,j) \in M$.

(6) We have
    \begin{align*}
    \off{i}{j} =
    \left \{
    \begin{array}{ll}
    \gamma_{i\to j} & \mbox{for } (ij)\in M \, ,\\
    (w_{ij} - \gamma_i)_+ & \mbox{for } (ij)\notin M\, .
    \end{array} \right .
    \end{align*}

(7) An optimum solution for the dual LP
    in \eqref{prob:b_LP_primal_dual} can be constructed as $y_i = \gamma_i$ for all $i \in V$ and:
    \begin{align*}
    y_{ij} =
    \left \{
    \begin{array}{ll}
    w_{ij} - \gamma_i - \gamma_j & \mbox{for } (ij)\in M\, ,\\
    0 & \mbox{for } (ij)\notin M\, .
    \end{array} \right .
    \end{align*}

\begin{proof}[Proof of Theorem \ref{thm:b_FP_NB}, direct implication]
Assume that the primal LP in (\ref{prob:b_LP_primal_dual}) has a unique optimum that is integral. Then,
by property 4, the set of strong-dotted edges $M$ is the unique
maximum weight matching $M^*$, i.e. $M=M^*$, and all other edges are non-dotted.
By property $3$, for $i$ that has less than $b_i$ partners under $M^*$, we have $\gamma_i=0$.
Hence by property $2$, we know that for $i$ not saturated under $M^*$, for every $(ij)\notin M^*$ since $(ij)$ is a non-dotted edge $\gamma_{j \to i} = \alf{i}{j} = 0$, and for every $(ij) \in M^*$ node $i$ gets a positive incoming offer $\gamma_{j \to i} = \off{j}{i}$.
For $i$ saturated under $M^*$, property 2 yields that the $\bmax{i}$ highest incoming offers to $i$ come from neighbors in $M^*$ (without ties).
It follows that $\baralf$ induces the matching $M^*$. Also, we deduce that $\gamma_{i \to j} = \gamma_{j \to i} = 0$ for $(ij)\notin M^*$.

From property 6 we deduce that $\gamma_{i \to j} + \gamma_{j \to i} = w_{ij}$ for $(ij) \in M$ and from property 1, we deduce that $ \off{j}{i}< \gamma_i$ for $(ij)\notin M$. It follows that $(M, \Gamma)$ is a trade outcome.
Finally, by properties 7 and 5, the pair
$(M^*,\bargamma)$ is stable and balanced respectively, and thus forms a NB
solution.
\end{proof}

%
%
\vspace{2mm}

\noindent\textbf{Proof of Theorem \ref{thm:b_FP_NB}: From NB solutions to fixed points.}~
\begin{proof}
Consider any NB solution $(M,\Gamma_{\NB} )$. Using Proposition \ref{prop:sotomayor},
 $M=M_*$,  the unique maximum weight matching.
Construct a corresponding FP as follows. Set
    \begin{align*}
    \off{i}{j} =
    \left \{
    \begin{array}{ll}
    \gamma_{\NB, i\to j} & \mbox{for } (ij)\in M \, ,\\
    (w_{ij} - \gamma_{\NB,i})_+ & \mbox{for } (ij)\notin M\, .
    \end{array} \right .
    \end{align*}
Compute $\baralf$ using $\alf{i}{j}=\bmax{i}_{k \in \di \backslash
j} \off{k}{i}$. We claim that this is a FP and that the
corresponding $\Gamma$ is $\Gamma_{\NB}$.

To prove that
we are at a fixed point, we imagine updated offers
$\baroff^{\textup{upd}}$ based on $\baralf$, and show
$\baroff^{\textup{upd}}=\baroff$.

Consider a matching edge $(i,j) \in M$. We know that
$\gamma_{\NB,i\to j} + \gamma_{\NB,j\to i}=w_{ij}$. Also, stability and
balance tell us
$$
\gamma_{\NB,j \to i} - \bmax{i}_{k \in \partial i \backslash j} (w_{ik} -\gamma_{\NB,k})_+
=
\gamma_{\NB,i \to j} - \bmax{j}_{l \in \partial j \backslash i} (w_{jl} -\gamma_{\NB,l})_+
$$
and both sides are non-negative. For $(i,k) \in M$, we know that $(w_{ik} -\gamma_{\NB,k})_+ \geq (w_{ik} -\gamma_{\NB,i \to k})_+ = \gamma_{\NB,k \to i} = \off{k}{i}$. It follows that
$$
\bmax{i}_{k \in \partial i \backslash j} (w_{ik} -\gamma_{\NB,k})_+ = \bmax{i}_{k \in \di \backslash j}
\off{k}{i} = \alf{i}{j}
$$
Hence, $\gamma_{\NB,j\to i} - \alf{i}{j} =
\gamma_{\NB,i\to j} - \alf{j}{i} \geq 0$. Therefore
$\alf{i}{j}+\alf{j}{i} \leq w_{ij}$,
\begin{align*}
\off{i}{j}^\textup{upd}&= \frac{w_{ij}-\alf{i}{j}+\alf{j}{i}}{2} = \frac{w_{ij}-\gamma_{\NB,j\to i}+\gamma_{\NB,i \to j}}{2} \\
&= \gamma_{\NB,i\to j} =  \off{i}{j}\, .
\end{align*}
By symmetry, we also have $\off{j}{i}^\textup{upd}=
\gamma_{\NB,j\to i} = \off{j}{i}$. Hence, the offers remain
unchanged.
Now
consider $(i,j) \notin M$. We have
$\gamma_{\NB,i}+\gamma_{\NB,j} \geq w_{ij}$ and,
$\gamma_{\NB,i} = \bmax{i}_{k \in \di \backslash j} \gamma_{\NB, k \to i} = \alf{i}{j}$. A similar equation holds for
$\gamma_{\NB,j}$. The validity of this identity can be checked
individually in the cases when $i$  is saturated under $M$ and $i$ is not saturated  under $M$.
Hence, $\alf{i}{j}+\alf{j}{i} \geq w_{ij}$. This leads to
$\off{i}{j}^{\textup{upd}}=(w_{ij}-\alf{i}{j})_+=(w_{ij}-\gamma_{\NB,i})_+=\off{i}{j}$.
By symmetry, we know also that $\off{j}{i}^{\textup{upd}}=\off{j}{i}$.

Finally, we show $\Gamma = \Gamma_{\NB}$. Note that since we have already established $\baralf$ is a fixed point, we know from the direct part that $\baralf$ induces the matching $M$, so there is no tie breaking required to determine the $b_i$ highest incoming offers to node $i \in V$.
For all $(i,j)
\in M$, we already found that $\off{i}{j}=\gamma_{\NB, i \to j}$ and vice
versa. For any edge $(ij) \notin M$, we know $\off{i}{j} =
(w_{ij}-\gamma_{\NB,i})_+ \leq \gamma_{\NB,j}$. This
immediately leads to $\Gamma = \Gamma_{\NB}$.
\end{proof}

\subsection{Proof of properties used in direct part}

Now we prove the fixed point properties that were used
in the direct part of the proof of Theorem \ref{thm:b_FP_NB}.
Before that, however, we remark that
the condition: ``the primal LP in \eqref{prob:b_LP_primal_dual} has a unique optimum''
in Theorem \ref{thm:b_FP_NB} is almost always valid.

\begin{Remark}
\label{rem:b_unique_opt_generic}
We argue that the condition ``the primal LP in \eqref{prob:b_LP_primal_dual} has a unique optimum'' is generic in instances with integral optimum:\\
Let ${\sf G}_{\textup{I}} \subset [0,1]^{|E|}$ be the set of instances having an integral optimum, given a graph $G$ with capacity constraints $\mathbf b$.
Let ${\sf G}_{\textup{UI}} \subset {\sf G}_{\textup{I}}$ be the set of instances having a unique integral optimum.
It turns out that ${\sf G}_{\textup{I}}$ has dimension $|E|$ (i.e.
the class of instances having an integral optimum is large) and
that ${\sf G}_{\textup{UI}}$ is both
\emph{open and dense in ${\sf G}_{\textup{I}}$}.
\end{Remark}

Again, we denote surplus $w_{ij}-\alf{i}{j}-\alf{j}{i}$ of edge $(ij)$ by $\surp_{ij}$.

\begin{lemma}\label{lemma:b_i-jpartners_equivalence}
The following are equivalent:\\
(a) $\gamma_i + \gamma_j \leq w_{ij}$,\\
(b) $\surp_{ij}\geq 0$.\\
(c) $\gamma_i\leq \off{j}{i}$ and $\gamma_j\leq \off{i}{j}$.\\
Moreover, if $\gamma_i\leq \off{j}{i}$ and $\gamma_j>\off{i}{j}$ then $\gamma_i=0$.
\end{lemma}
\begin{proof} We will prove $(a) \Rightarrow (b) \Rightarrow (c) \Rightarrow (a)$.

$(a)\Rightarrow (b)$: Since $\gamma_i\geq \alf{i}{j}$ and $\gamma_j\geq \alf{j}{i}$ always holds then $w_{ij}\geq \gamma_i + \gamma_j\geq \alf{i}{j} + \alf{j}{i}$.

$(b)\Rightarrow (c)$: If $\surp_{ij}\geq 0$ then $\off{i}{j}=(w_{ij}-\alf{i}{j}+\alf{j}{i})/2\geq \alf{j}{i}$.
So $\off{i}{j}$ is among the $b_j$ best offers received by node $j$, implying $\gamma_j \leq \off{i}{j}$.  The argument for $\gamma_i\leq \off{j}{i}$ is similar.

$(c)\Rightarrow (a)$: If $\surp_{ij} \geq 0$ then $\off{i}{j}=(w_{ij}-\alf{i}{j}+\alf{j}{i})/2$ and $\off{j}{i}=(w_{ij}-\alf{j}{i}+\alf{i}{j})/2$ which gives $\gamma_i+\gamma_j\leq \off{i}{j}+\off{j}{i}=w_{ij}$ and we are done.  Otherwise, we have
$\gamma_i+\gamma_j\leq \off{i}{j}+\off{j}{i}= (w_{ij}-\alf{i}{j})_+ + (w_{ij}-\alf{j}{i})_+ < \max \bigg[(w_{ij}-\alf{i}{j})_+, (w_{ij}-\alf{j}{i})_+, 2w_{ij}-\alf{i}{j}-\alf{j}{i}\bigg]\leq w_{ij}$.

Finally, we suppose $\gamma_i=\off{j}{i}$ and $\gamma_j>\off{i}{j}$. First note that by equivalence of $(b)$ and $(c)$ we should have $w_{ij} < \alf{i}{j} + \alf{j}{i}$. On the other hand $\alf{i}{j}\leq \gamma_i  \leq \off{j}{i}\leq (w_{ij}-\alf{j}{i})_+$.  Now if  $w_{ij}-\alf{j}{i}>0$ we get $\alf{i}{j}\leq w_{ij}-\alf{j}{i}$ which is a contradiction. Therefore $\gamma_i\leq (w_{ij}-\alf{j}{i})_+=0$, implying $\gamma_i = 0$.
\enp
\end{proof}

\begin{lemma}
The following are equivalent:\\
(a) $\off{j}{i} > \alf{i}{j}$,\\
(b) $\off{i}{j} > \alf{j}{i}$,\\
(c) $w_{ij}-\alf{i}{j} - \alf{j}{i} >0$.\\
(d) $i$ and $j$ receive positive offers from each other, with $\off{j}{i} > ((b_{i}+1)^{\rm th}\!\operatorname{-max})_{k \in \di} \off{k}{i}$ and similarly for $i$.

These conditions imply that $\off{j}{i} = \gamma_{j \to i}$ and $\off{i}{j} = \gamma_{i \to j}$.
\label{lemma:b_strong_dotted}
\end{lemma}
\begin{proof}
$(a)\Rightarrow (c) \Rightarrow (b)$: $(a)$ implies that $(w_{ij} - \alf{j}{i})_+ \geq \off{j}{i} > \alf{i}{j}$, which yields (c).  From this we can explicitly write $\off{i}{j}=(w_{ij}-\alf{i}{j}+\alf{j}{i})/2$ which is strictly bigger than $\alf{j}{i}$. Hence we obtain $(b)$.

By symmetry $(b)\Rightarrow (c) \Rightarrow (a)$.  Thus, we have shown that $(a)$, $(b)$ and $(c)$ are equivalent.

$(c) \Rightarrow (d)$: $(c)$ implies that $\off{i}{j}= (w_{ij}-\alf{i}{j} + \alf{j}{i})/2 > \alf{j}{i} = \bmax{j}_{k \in \partial j \backslash i} \off{k}{j}$. Using symmetry, it follows that $(d)$ holds.

$(d) \Rightarrow (a)$ is easy to check.

This finishes the proof of equivalence. The implication follows from the definition of $\gamma_{i \to j}$.
\enp
\end{proof}

Recall that $(ij)$ is a weak-dotted edge if $w_{ij}-\alf{i}{j} - \alf{j}{i}=0$, a strong-dotted edge if
$w_{ij}-\alf{i}{j} - \alf{j}{i} > 0$, and a non-dotted edge otherwise.

\begin{lemma}
If $\gamma_i> 0$, then $i$ has $b_i$ adjacent strong dotted edges, or at least $b_i+1$ adjacent dotted edges.
\label{lemma:b_no_dotted_means_gamma0}
\end{lemma}
\begin{proof}
Suppose $\gamma_i > 0$. Then the $b_i$ largest incoming offers to $i$ are all strictly positive. Suppose one of these offers comes from $j$.
Then, $\alf{i}{j}\leq \off{j}{i}\leq (w_{ij}-\alf{j}{i})_+$. Now $\off{i}{j} > 0$, implying that  $\alf{i}{j}\leq w_{ij}-\alf{j}{i}$ or $(ij)$ is dotted edge. If there is strict inequality for all $j$, this means that we have at least $b_i$ strong dotted edges adjacent to $i$. If we have equality for some $j$, that means there is a tie for the $b_i$ highest offer incoming to $i$. We deduce that at least $b_i+1$ dotted edges adjacent to $i$.
\enp
\end{proof}

\begin{lemma}
The following are equivalent:\\
(a) $\surp_{ij}\le 0$,\\
(b) $\off{i}{j}=(w_{ij}-\alf{i}{j})_+$.\\
Moreover, these conditions  imply $\alf{i}{j}=\gamma_i$ and $\alf{j}{i}=\gamma_j$.
\label{lemma:b_when_alf_equals_gamma}
\end{lemma}
\begin{proof}
Equivalence of (a) and (b) follows from the definitions. (a) implies $\off{i}{j} \leq \alf{j}{i}$ which yields
$\alf{j}{i} = \gamma_j$. By symmetry, we can also deduce $\alf{i}{j}=\gamma_i$.
\enp
\end{proof}
Note that (a) is symmetric in $i$ and $j$, so (b) can be transformed by interchanging $i$ and $j$.

\begin{lemma}
A non-solid edge cannot be a dotted edge, weak or strong.
\label{lemma:b_no_extra_dotted}
\end{lemma}

The proof of this lemma is very similar to that of Lemma \ref{lemma:no_extra_dotted}: we consider optimal solutions of the primal and dual LPs \eqref{prob:b_LP_primal_dual}, and construct an alternating path consisting of alternate (i) non-solid dotted edges, and (ii) strong solid, non-strong dotted edges. We omit the proof.

\begin{lemma}
Every strong-solid edge is a strong-dotted edge.
\label{lemma:b_no_extra_solid}
\end{lemma}
Again, the proof is very similar to the proof of Lemma \ref{lemma:no_extra_solid}, and we omit it.

\begin{lemma}
Consider the set of edges $M\equiv \{(ij): \surp_{ij} > 0\}$. The balance property \eqref{eq:b_balance} is satisfied for all $(ij) \in M$.
\end{lemma}
\begin{proof}
Consider any $(ij) \in M$. From Lemma \ref{lemma:b_strong_dotted}, we know that
$$
\gamma_{j \to i} = \off{j}{i} = \alf{i}{j} + \surp_{ij}/2
$$
To prove balance, it then suffices to establish
\begin{align}
\bmax{i}_{k \in \di \backslash j} \; \off{k}{i} = \alf{i}{j}\, .
\label{eq:b_best_alt_balance}
\end{align}
 Now node $i$ can have at most $b_i$ adjacent strong dotted edges, from Lemma \ref{lemma:strong_dotted} (d). One of these is $(ij)$. \eqref{eq:b_best_alt_balance} follows from the   property $\off{i}{j}=(w_{ij}-\gamma_{i})_+$ on non-strong dotted edges (from Lemma \ref{lemma:b_when_alf_equals_gamma}).
\end{proof}

\begin{lemma}
An optimum solution for the dual LP
    in \eqref{prob:b_LP_primal_dual} can be constructed as $y_i = \gamma_i$ for all $i \in V$ and:
    \begin{align*}
    y_{ij} =
    \left \{
    \begin{array}{ll}
    w_{ij} - \gamma_i - \gamma_j & \mbox{for } (ij)\in M\, ,\\
    0 & \mbox{for } (ij)\notin M\, .
    \end{array} \right .
    \end{align*}
\end{lemma}
\begin{proof}
We first show that this construction satisfies the dual constraints.
$y_{ij}\geq 0$ follows from Lemma \ref{lemma:b_i-jpartners_equivalence} (a) and (b). We have
\begin{align}
y_i + y_j + y_{ij} = w_{ij} \qquad \mbox{for $(ij) \in M$}
\label{eq:b_dual_Medge}
 \end{align}
 by construction. For $(ij) \notin M$, we have $y_i + y_j = \gamma_i + \gamma_j \geq w_{ij}$ from Lemma \ref{lemma:b_when_alf_equals_gamma}. This completes our proof of feasibility.

To show optimality, we establish that the weight of matching $M$ is the same as the dual objective value $\sum_{i \in V} b_i y_i + \sum_{e \in E} y_{ij}$ at the chosen $\underline{y}$. Lemma \ref{lemma:b_no_dotted_means_gamma0} guarantees $\gamma_i = 0$ for $i$ having less than $b_i$ adjacent dotted edges. Using Lemmas \ref{lemma:b_no_extra_dotted} and \ref{lemma:b_no_extra_solid}, we know that $M$ is a valid $\mathbf b$ matching consisting of strong dotted edges, and all other edges are non-dotted. We deduce
\begin{align*}
w(M) = \sum_{(ij)\in M} w_{ij} = \sum_{(ij)\in M} y_i + y_j + y_{ij} = \sum_{i \in V} b_i y_i + \sum_{e \in E} y_{ij} \, ,
\end{align*}
using \eqref{eq:b_dual_Medge}. This completes our proof of optimality.
\end{proof}

\section{The Kleinberg-Tardos construction and the KT gap}
\label{sec:KT+gap}

\begin{figure*}
\centering
\includegraphics[scale=0.3]{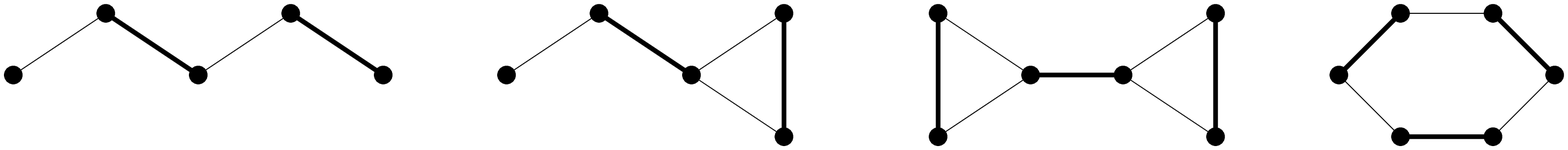}
\caption{Examples of basic structures: path, blossom, bicycle, and cycle
(matched edges in bold).}
\label{fig:structures}
\end{figure*}

Let $G$ be an instance which admits at least one stable
outcome, $M^*$ be the corresponding matching (recall that this
is a maximum weight matching), and consider the
Kleinberg-Tardos (KT) procedure for finding a NB solution
\cite{KT}. Any NB solution $\bargamma^*$ can be constructed by
this procedure with appropriate choices at successive stages.
At each stage, a linear program is solved with variables
$\gamma_i$ attached to node $i$. The linear program maximizes
the minimum `slack' of all unmatched edges and nodes, whose
values have not yet been set (the slack of edge $(i,j)\not\in
M$ is $\gamma_i+\gamma_j-w_{ij}$).

At the first stage, the set of nodes that remain unmatched
(i.e. are not part of $M^*$) is found, if such nodes exist.
Call the set of unmatched nodes $\cC_0$.
After this, at successive stages of the KT procedure, a
sequence of structures $\cC_1$, $\cC_2,\ldots, \cC_k$
characterizing the LP optimum are found. We call this the
\textit{KT sequence}. Each such structure is a pair $\cC_q =
(V(\cC_q),E(\cC_q))$ with $V(\cC_q)\subseteq V$,
$E(\cC_q)\subseteq E$. According to \cite{KT} $\cC_q$ belongs
to one of four topologies: alternating path, blossom, bicycle,
alternating cycle (Figure \ref{fig:structures}). The $q$-th
linear program determines the value of $\gamma_i^*$ for $i\in
V(\cC_q)$. Further, one has the partition
$E(\cC_q)=E_1(\cC_q)\cup E_2(\cC_q)$ with  $E_1(\cC_q)$
consisting of all matching edges along which nodes in
$V(\cC_q)$ trade, and $E_2(\cC_q)$ consists of edges $(i,j)$
such that some $i\in V(\cC_q)$ receives its second-best,
positive offer from $j$.

The $\gamma$ values for nodes on the limiting structure are
uniquely determined if the structure is an alternating path,
blossom or bicycle\footnote{In \cite{KT} it is claimed that the
$\gamma$ values `may not be fully determined' also in the case
of bicycles. However it is not hard to prove that $\gamma$
values are, in fact, uniquely determined in bicycles.}. In case of an alternating cycle there is one degree
of freedom -- setting a value $\gamma_i^*$ for one node $i \in
\cC_q$ fully determines the values at the other nodes.

We emphasize that, within the present definition, $\cC_q$ is
not necessarily a subgraph of $G$, in that it might contain an
edge $(i,j)$ but not both its endpoints. On the other hand,
$V(\cC_q)$ is always subset of the endpoints of $E(\cC_q)$. We
denote by $V_{\textup{ext}}(\cC_q)\supseteq V(\cC_q)$ the set
of nodes formed by all the endpoints of edges in $E(\cC_q)$.

For all nodes $i\in V(\cC_q)$ the second best offer is equal to
$\gamma_i^* -\sigma_q$, where $\sigma_q$ is the slack of the
$q$-th structure. Therefore
\vskip-13pt
\begin{eqnarray*}
\gamma_{i}^*+\gamma_j^*-w_{ij} = \left\{
\begin{array}{ll}
0&\mbox{if $(i,j)\in E_1(\cC_q)$,}\\
\sigma_q &\mbox{if $(i,j)\in E_2(\cC_q)$.}\\
\end{array}\right.
\label{eq:edge_slack}
\end{eqnarray*}
\vskip-3pt
%

The slacks form an increasing sequence ($\sigma_1\le \sigma_2
\le \ldots \le \sigma_k$).

\begin{Definition}
We say that a unique Nash bargaining solution $\baralf^*$
has a KT gap $\sigma$ if
\vskip-15pt
\begin{eqnarray*}
\sigma \le \min\big\{\sigma_1;\, \sigma_2-\sigma_1;\, \dots;\,\sigma_k-
\sigma_{k-1}\big\}\, ,
\label{eq:sigma_cond_1}
\end{eqnarray*}
and if for each edge $(i,j)$ such that $i,j\in
V_\textup{ext}(\cC_q)$ and $(i,j)\not\in E(\cC_q)$,
\begin{eqnarray*}
\gamma_i^*+\gamma_j^* - w_{ij}\ge \sigma_q+\sigma\, .
\label{eq:extra_sigma_condition}
\end{eqnarray*}
\end{Definition}

It is possible to prove that the positive gap condition is
generic in the following sense. The set of all instances such
that the NB solution is unique can be regarded as a subset
${\sf G}\subseteq[0,W]^{|E|}$ ($W$ being the maximum edge
weight). It turns out that ${\sf G}$ has dimension $|E|$ (i.e.
the class of instances having unique NB solution is large) and
that the subset of instances with gap $\sigma>0$ is both
\emph{open and dense in ${\sf G}$}.
\section{Appendix to Section \ref{sec:UD}}

\subsection{UD example showing exponentially slow convergence}
\label{subapp:UD_construction}

We construct a sequence of  instances $I_n$ along with particular initializations such that convergence to an approximate fixed point is exponentially slow, cf. Section \ref{subsec:UD_slow_convergence}. Let us first consider $n= 8 N$, where $N \geq 2$. The construction can be easily extended to arbitrary $n \geq 16$. The graph $G_n=(V_n, E_n)$ we will consider is a simple `ring'. More precisely,
$V_n=\{1, 2, \ldots, n\}$ and $E_n=\{(1,2), (2,3), \ldots, (n-1,n), (n,1) \}$. All edges have the same weight $W$, with the exception $w_{4N,4N+1}= W-1$. (We will choose $W$ later.) Thus, $G_n$ has a unique maximum weight matching $M_*=\{(1,2), (3,4), \ldots, (n-1,n), \}$. Moreover, $M^*$ solves LP \eqref{prob:mwm_relaxation}, implying that $G_n$ possesses a UD solution (using Lemma \ref{lemma:UD_existence}), and the LP gap $g = 1$ provided $W \geq 1$.
Given any $r \in (0, 1/2)$,
we define the split fractions as follows
\footnote{It turns out the values of the split fraction on approximately half the edges is irrelevant for our choice of $\baralf$.}
: We first define the split fractions on edges $E_{n, {\rm left}}=\{(1,2), (2,3), \ldots, (4N-1, 4N) \}$. We set
\begin{align*}
r_{3,2} &= r_{5,4} = \ldots = r_{2N-1, 2N-2} = r \\
r_{2N+1,2N+2} &= r_{2N+3,2N+4} = \ldots = r_{4N-3, 4N-2} = r
\end{align*}
and the split fractions on all other edges in $E_{n, {\rm left}}$ to $1/2$.
As before, $r_{i,i+1}= 1- r_{i+1,i}$ is implicit.

For edges on the `right' side we define split fractions in a symmetrical way. Define `reflection' $\rfl : \{ 4N+1, 4N+2, \ldots, 8N\} \rightarrow \{ 1, 2, \ldots, 4N \}$
as
\begin{align}
\rfl (l) = n - l + 1
\label{eq:reflection}
\end{align}
We set $r_{i, i+1} = r_{\rfl(i), \rfl(i+1)}$ for all $i \in \{ 4N+1, 4N+2, \ldots, 8N-1\}$.

Finally, we set $r_{n,1} = r_{4N, 4N+1} =1/2$.


Now we show how to construct $\baralf$ satisfying properties (a) and (b) above.
Let $\beta \equiv r/(1-r) < 1$. Set $\alf{n}{1}=\alf{1}{n}=W/2-1$.

For $0\leq i \leq N-1$, define\footnote{The summations can be written in closed form if desired, but it turns out to be easier to work with the summations directly.}
\begin{align*}
\alf{2i+1}{2i+2} = W/2 - \sum_{j=0}^{i-1}\beta^j\, ,\qquad
\alf{2i+2}{2i+1} = W/2 + \sum_{j=0}^{i}\beta^j\, .
\end{align*}
For $1\leq i \leq N-1$, define
\begin{align*}
\alf{2i}{2i+1} = W/2 + \sum_{j=0}^{i-2}\beta^j\, ,\qquad
\alf{2i+1}{2i} = W/2 - \sum_{j=0}^{i}\beta^j\, .
\end{align*}

Let $\alf{2N}{2N+1} = W/2 + \sum_{j=0}^{N-1}\beta^j$, $\alf{2N+1}{2N} = W/2 - \sum_{j=0}^{N-1}\beta^j$. For $0\leq i \leq N-2$, define
\begin{align*}
\alf{2N+2i+1}{2N+2i+2} = W/2 - \sum_{j=0}^{N-i-1}\beta^j\, ,\qquad
\alf{2N+2i+2}{2N+2i+1} = W/2 + \sum_{j=0}^{N-i-3}\beta^j\, .
\end{align*}
For $1\leq i \leq N-1$, define
\begin{align*}
\alf{2N+2i}{2N+2i+1} = W/2 + \sum_{j=0}^{N-i-1}\beta^j\, ,\qquad
\alf{2N+2i+1}{2N+2i} = W/2 - \sum_{j=0}^{N-i-2}\beta^j\, .
\end{align*}
Next, let $\alf{4N-1}{4N}=\alf{4N}{4N-1}=W/2-1$.
$\alf{4N}{4N+1}=\alf{4N+1}{4N}=W/2$.
Again, the messages on the `right' side are defined in terms of the reflection Eq.~\eqref{eq:reflection}. We define
$\alf{i}{i+1} = \alf{\rfl(i)}{\rfl(i+1)}$, and $\alf{i+1}{i} = \alf{\rfl(i+1)}{\rfl(i)}$, for all $i \in \{ 4N+1, 4N+2, \ldots, 8N-1\}$.

The corresponding offers $\baroff(\baralf)$ can be easily computed. We have $\off{1}{n}=\off{n}{1}=W/2$.
For $0\leq i \leq N-1$, we have
\begin{align*}
\off{2i+1}{2i+2} = W/2 + \sum_{j=0}^{i-1}\beta^j\, ,\qquad
\off{2i+2}{2i+1} = W/2 + \sum_{j=0}^{i}\beta^j\, .
\end{align*}
For $1\leq i \leq N-1$, we have
\begin{align*}
\off{2i}{2i+1} = W/2 - \sum_{j=0}^{i-1}\beta^j\, ,\qquad
\off{2i+1}{2i} = W/2 + \sum_{j=0}^{i-1}\beta^j\, .
\end{align*}

We have $\off{2N}{2N+1} = W/2 - \sum_{j=0}^{N-1}\beta^j$, $\off{2N+1}{2N} = W/2 + \sum_{j=0}^{N-1}\beta^j$. For $0\leq i \leq N-2$, we have
\begin{align*}
\off{2N+2i+1}{2N+2i+2} = W/2 + \sum_{j=0}^{N-i-2}\beta^j\, ,\qquad
\off{2N+2i+2}{2N+2i+1} = W/2 - \sum_{j=0}^{N-i-2}\beta^j\, .
\end{align*}
For $1\leq i \leq N-1$, we have
\begin{align*}
\off{2N+2i}{2N+2i+1} = W/2 - \sum_{j=0}^{N-i-1}\beta^j\, ,\qquad
\off{2N+2i+1}{2N+2i} = W/2 + \sum_{j=0}^{N-i-2}\beta^j\, .
\end{align*}
Finally, we have $\off{4N-1}{4N}=\off{4N}{4N-1}=W/2$,
$\off{4N}{4N+1}=\off{4N+1}{4N}=W/2-1$. Offers on the `right' side are clearly the same as the corresponding offer on the left.

Importantly, note that it suffices to have $W \geq 2 + 2/(1-\beta)$ to ensure that all messages and offers are bounded below by $1$ (in particular, none of them is negative). Choose $W=2 + 2/(1-\beta)$ (for example).

Next, note that all the fixed point conditions are satisfied, except the update rules for $\alf{2N}{2N+1}$, $\alf{2N+1}{2N}$ and the corresponding messages in the reflection. However, it is easy to check that these update rules are each violated by only $\beta^{N-1}$. Thus, $\baralf$ is an $\eps_n$-FP for $\eps_n = \beta^{N-1} \leq 2^{-cn}$ for appropriate $c=c(r)>0$.

\subsection{Proof of Theorem \ref{thm:FPTAS}}
\label{subapp:UD_FPTAS}

We show how to obtain an FPTAS using the two steps stated in Section \ref{subsec:UD_FPTAS}.

Step 1 can be carried out by finding a maximum weight matching $M^*$ (see, e.g., \cite{PolyMWM}) and also solving the
the dual linear program \eqref{prob:mwm_dual}. For the dual LP, let $v$ be the optimum value and let
$\bargamma$ be an optimum solution. We now use Proposition \ref{prop:sotomayor}.
If the weight of $M^*$ is smaller than $v$, we return {\sc unstable}, since we know
that no stable outcome exists, hence no UD solution (or $\eps$-UD solution) exists. Else,
$(\bargamma,M^*)$ is a stable outcome. This completes step 1! The computational effort involved is
polynomial in the input size.
All unmatched nodes have earnings of $0$. 

In step 2, we fix the matching $M^*$, and rebalance the matched edges iteratively. It turns out to be crucial that our iterative
updates preserve stability.  An example similar to the one in Section \ref{subapp:UD_construction} demonstrates that the rebalancing procedure can take an
exponentially large time to reach an approximate UD solution if stability is not preserved \cite{myWINE}. 

We now motivate the rebalancing procedure
briefly, before we give a detailed description and state results. Imagine an edge $(i,j) \in M^*$. Since we start with a
stable outcome, the edge weight $w_{ij}$ is at least the sum of the best alternatives, i.e. $\surp_{ij}\geq0$. Suppose we change the
division of $w_{ij}$ into $\gamma_i'$, $\gamma_j'$ so that the $\surp_{ij}$ is divided as per the prescribed split
fraction $r_{ij}$. Earnings of all other nodes are left unchanged.
Since $r_{ij} \in (0,1)$, $\gamma_i'$ is at least as large as the best alternative of $i$, as was the case for $\gamma_i$.
This leads to $\gamma_i'+\gamma_k \geq w_{ik}$ for all $k \in \partial i \backslash j$. A similar argument holds for node $j$.
In short, \emph{stability is preserved}!

It turns out that the analysis of convergence is simpler if we analyze synchronous updates, as opposed to
asynchronous updates as described above. Moreover, we find that simply choosing an appropriate `damping factor' allows us to ensure that stability
is preserved even with synchronous updates. We use a powerful technique introduced in our recent work \cite{KT} to prove
convergence.


Table \ref{alg:edge_reb} shows the algorithm {\sc Edge Rebalancing} we use to complete step 2.
Note that each iteration
of the loop can requires $O(|E|)$ simple operations.\\

\begin{table}[t]
\caption{Local algorithm that converts stable outcome to $\eps$-UD solution}
\begin{center}
{\fontfamily{cmtt}\selectfont
\begin{tabular}{ll}
\hline
\multicolumn{2}{l}{  {\sc Edge Rebalancing}$\,\big(\;$Instance $I$, Stable outcome $(\bargamma,M)$,}\\
\multicolumn{2}{l}{ \phantom{{\sc Edge Rebalancing}( I} Damping factor $\damp$, Error target $\eps\;\big)$}\\[2pt]
\hline
1: & Check $\damp \in (0, 1/2]$, $\eps > 0$, $(\bargamma,M)$ is stable outcome \sT \\[2pt]
2: & If (Check fails)\hspace{0.1cm} Return {\sc error}\\[2pt]
3: & $\bargamma^0 \leftarrow \bargamma$\\[2pt]
4: & $t \leftarrow 0$\\[2pt]
5: & Do\\[2pt]
6: & \hspace{0.4cm} ForEach $(i,j) \in M$\\[2pt]
7: & \hspace{0.9cm} $\gamma^{\reb}_i \leftarrow \max_{k \in \partial i \backslash j} (w_{ik} - \gamma_k^t)_+ + r_{ij} \surp_{ij}(\bargamma^t)$\\[2pt]
8: & \hspace{0.9cm} $\gamma^{\reb}_j \leftarrow \max_{l \in \partial j \backslash i}\, (w_{jl} \, - \gamma_l^t)_+ + r_{ji} \surp_{ij}(\bargamma^t)$\\[2pt]
9: & \hspace{0.4cm} End ForEach\\[2pt]
10: & \hspace{0.4cm} ForEach $i\in V$ that is unmatched under $M$\\[2pt]
11: & \hspace{0.9cm} $\gamma^{\reb}_i \leftarrow 0$\\[2pt]
12: & \hspace{0.4cm} End ForEach\\[2pt]
13: & \hspace{0.4cm} If $\left(\lVert\bargamma^{\reb} -\bargamma^t\rVert_\infty \leq \eps \right)$\hspace{0.1cm} Break Do \\[2pt]
14: & \hspace{0.4cm} $\bargamma^{t+1} =  \damp \bargamma^{\reb} + (1-\damp) \bargamma^t$\\[2pt]
15: & \hspace{0.4cm} $t \leftarrow t+1$ \\[2pt]
16: & End Do \\[2pt]
17: & Return $(\bargamma^t, M)$\\[2pt]
\hline
\end{tabular}}
\end{center}
\label{alg:edge_reb}
\end{table}

\noindent{\bf Correctness  of {\sc Edge Rebalancing}:}\\
It is easy to check that $(\bargamma^{\reb}, M)$ and $(\bargamma^t, M)$ are valid outcomes for all $t$.
We show that $\bargamma^t$ computed by {\sc Edge Rebalancing} is, in fact, a stable allocation:
\begin{lemma}
\label{lemma:ER_preserves_stability}
If {\sc Edge Rebalancing} is given a valid input satisfying the `Check' on line 1, then
$(\gamma^t, M)$ is a stable outcome for all $t \geq 0$.
\end{lemma}
\begin{proof}
We prove this lemma by induction on time $t$. Clearly $(\gamma^0, M)$ is a stable outcome, since the input is valid. Suppose $(\gamma^t, M)$ is a stable outcome.

Consider any $(i,j) \in M$. It is easy to verify that $\gamma_i^\reb + \gamma_j^\reb = w_{ij}$, for
$\bargamma^{\reb}$ computed from $\bargamma^t$ in
Lines 6-9 of {\sc Edge Rebalancing}. Also, we know that $\gamma_i^t + \gamma_j^t = w_{ij}$.
It follows that $\gamma_i^{t+1} + \gamma_j^{t+1} = w_{ij}$ as needed. For $i \in V$ unmatched
under $M$, $\gamma_i^t=0$ by hypothesis and as per Lines 10-12, $\gamma_i^\reb=0 \, \Rightarrow \, \gamma_i^{t+1}=0$ as needed.

Consider any $(i,k) \in E \backslash M$. We know that $\gamma_i^t + \gamma_k^t \geq w_{ik}$. We want to show the corresponding inequality
at time $t+1$. Define $\sigma_{ik}^t \equiv \gamma_i^t + \gamma_k^t - w_{ik} \geq 0$.

\begin{claim}
$\gamma_i^{\reb} \geq \gamma_i^t - \sigma_{ik}^t$.
\end{claim}

If we prove the claim, it follows that a similar inequality holds for $\gamma_k^{\reb}$, and hence
$\gamma_i^{\reb}+\gamma_k^{\reb} \geq\gamma_i^t + \gamma_k^t - 2\sigma_{ik}= w_{ik}- \sigma_{ik}^t$.
It then follows from the definition in Line 14
that $\gamma_i^{t+1} + \gamma_k^{t+1} \geq w_{ik}$,
for any $\damp \in (0,1/2]$. This will complete our proof  that $(\bargamma^{t+1}, M)$ is a stable outcome.

Let us now prove the claim. Suppose $i$ is matched under $M$.
Using the definition in Line 7 (Line 8 contains a symmetrical definition), 
$\gamma_i^\reb \geq \max_{k' \in \partial i \backslash j} (w_{ik'} - \gamma_{k'}^t)_+$ since $\surp_{ik}(\bargamma^t) \geq 0$. Hence,
\begin{align*}
\gamma_i^\reb \geq (w_{ik} - \gamma_k^t)_+ \geq (w_{ik} - \gamma_k^t) = \gamma_i^t - \sigma_{ik}^t \, ,
\end{align*}
as needed.  If $i$ is not matched under $M$, then $\gamma_i^t=\gamma_i^\reb=0$, so the claim follows from $\sigma_{ik}^t \geq 0$.
\end{proof}

\noindent{\bf Convergence  of {\sc Edge Rebalancing}:}\\
Note that the termination condition $\lVert\bargamma^{\reb} -\bargamma^t\rVert_\infty \leq \eps$ on Line 13 is
equivalent to $\eps$-correct division. We show that the rebalancing
algorithm terminates fast:
%
\begin{lemma}
\label{lemma:ER_converges_fast}
For any instance with weights bounded by $1$, i.e. $(w_e, e \in E) \in (0,1]^{|E|}$, if {\sc Edge Rebalancing}
is given a valid input, it terminates
in $T$ iterations, where
\begin{align}
T \leq \left \lceil \frac{1}{\pi \damp (1-\damp) \eps^2} \right \rceil \, .
\end{align}
and returns an outcome satisfying $\eps$-correct division (cf. Definition \ref{def:eps_correct_div}). Here $\pi = 3.14159\ldots$
\end{lemma}
\begin{proof}[sketch only]
This result is proved using the same technique as in the proof of Theorem \ref{thm:rate_of_convergence}.
The iterative updates of {\sc Edge Rebalancing} can be written as
\begin{align}
\bargamma^{t+1} = \damp \sT\bargamma^t + (1-\damp)\bargamma^t \, ,
\label{eq:mann_iterations}
\end{align}
where $\sT$ is a self mapping of the (convex) set of allocations corresponding to matching $M$, $\cA_M \subseteq [0,1]^{|E|}$.
The `edge balancing' operator $\sT$ essentially corresponds to Lines 6-9 in Table \ref{alg:edge_reb}.
It is fairly straightforward to show that $\sT$ is nonexpansive with respect to
sup norm.
The main theorem in \cite{rate} then tells us that
\begin{align}
\lVert \sT \bargamma^t - \bargamma^t \rVert _\infty \leq \frac{1}{\sqrt{\pi \damp (1-\damp) t}}\, .
\label{eq:BB_rate}
\end{align}
The result follows.

For the full proof, see \cite{myWINE}.
\end{proof}

Using Lemmas \ref{lemma:ER_preserves_stability} and Lemmas \ref{lemma:ER_converges_fast}, we immediately obtain
our main result, Theorem \ref{thm:FPTAS}.

\begin{proof}[of Theorem \ref{thm:FPTAS}]
We showed that step 1 can be completed in polynomial time.
If the instance has no UD solutions then the algorithm returns {\sc unstable}. Else
we obtain a stable outcome and proceed to step 2.

Step 2 is performed using {\sc Edge Rebalancing}. The input is the instance, the stable outcome obtained from step 1,
$\damp = 1/2$ (for example) and the target error value $\eps > 0$. Lemmas \ref{lemma:ER_preserves_stability} and \ref{lemma:ER_converges_fast}
show that {\sc Edge Rebalancing} terminates after at most $\lceil 1/(\pi \damp (1-\damp) \eps^2) \rceil$ iterations, returning a outcome that is
stable and satisfies $\eps$-correct division, i.e. an $\eps$-UD solution.
Moreover, each iteration requires $O(|E|)$ simple operations. Hence, step 2 is completed in  $O(|E|/\eps^2)$ simple operations.

The total number of operations required by the entire algorithm is thus polynomial in the input and in $(1/\eps)$.
\end{proof}

\end{document}